\documentclass[12pt]{article}
\usepackage{amsmath}
\usepackage{times}
\usepackage{graphicx}
\usepackage{color}
\usepackage{multirow}
\usepackage[authoryear]{natbib}
\usepackage{rotating}
\usepackage{bbm}
\usepackage{latexsym}

\usepackage{amsmath}
\usepackage{amsthm}
\usepackage{times}
\usepackage{graphicx}
\usepackage{color}
\newcommand{\rev}[1]{\textcolor{black}{#1}}

\newcommand{\SP}[1]{{\color{blue} #1}}
\usepackage{multirow}
\usepackage[toc,page]{appendix}
\usepackage{subcaption}

\usepackage[authoryear]{natbib}
\usepackage{rotating}
\usepackage{bbm}
\usepackage{enumitem}
\usepackage{latexsym}

\usepackage{pgfplots,caption}
\pgfplotsset{compat=1.9}

\usepackage{amssymb,amsmath,amsfonts}
\usepackage{epsf}
\usepackage{wrapfig}
\usepackage[normalem]{ulem}
\usepackage{caption}
\usepackage{fixmath}
\usepackage{hyperref}
\usepackage{tikz}
\usepackage{tabularx}
\usepackage{booktabs,caption}
\usepackage[flushleft]{threeparttable}
\usepackage{natbib}

\newtheorem{theorem}{Theorem}

\newtheorem{lemma}[theorem]{Lemma}
\newtheorem{remark}{Remark}

\newcommand{\R}{\mathbb{R}}

\newcommand{\Na}{Na$^+$} 
\newcommand{\mtot}{M_\text{tot}}
\newcommand{\ntot}{N_\text{tot}}

\newcommand{\K}{K$^+$}

\renewcommand{\Pr}{\mathbb{P}}

\newtheorem{innercustomgeneric}{\customgenericname}
\providecommand{\customgenericname}{}
\newcommand{\newcustomtheorem}[2]{%
  \newenvironment{#1}[1]
  {%
   \renewcommand\customgenericname{#2}%
   \renewcommand\theinnercustomgeneric{##1}%
   \innercustomgeneric
  }
  {\endinnercustomgeneric}
}

\newcustomtheorem{customthm}{Theorem}
\newcustomtheorem{customlemma}{Lemma}

\newcommand{\mbc}{\mathbf{c}}
\newcommand{\mbx}{\mathbf{x}}

\newcommand{\mbf}{\mathbf{f}}

\newcommand{\mbX}{\mathbf{X}}
\newcommand{\mbdX}{\mathbf{dX}}

\def\given{\:|\:}

\newcommand{\mbN}{\mathbf{N}}
\newcommand{\mbn}{\mathbf{n}}
\newcommand{\mbm}{\mathbf{m}}

\newcommand{\mbdW}{\mathbf{dW}}
\newcommand{\mbW}{\mathbf{W}}

\newcommand{\mbM}{\mathbf{M}}

\definecolor{blue}{rgb}{0,0,1}
\definecolor{darkgreen}{rgb}{0,.5,0}
\definecolor{darkred}{rgb}{.75,0,0}
\definecolor{red}{rgb}{1,0,0}

\newcommand{\bigzero}{\mbox{\normalfont\Large\bfseries 0}}

\newcommand{\captionfonts}{\normalsize}

\makeatletter  
\long\def\@makecaption#1#2{%
  \vskip\abovecaptionskip
  \sbox\@tempboxa{{\captionfonts #1: #2}}%
  \ifdim \wd\@tempboxa >\hsize
    {\captionfonts #1: #2\par}
  \else
    \hbox to\hsize{\hfil\box\@tempboxa\hfil}%
  \fi
  \vskip\belowcaptionskip}
\makeatother   

\begin{document}
\hspace{13.9cm}1

\ \vspace{20mm}\\

{\LARGE Fast and Accurate Langevin Simulations of Stochastic Hodgkin-Huxley Dynamics}

\ \\
{\bf \large Shusen Pu$^{\displaystyle 1}$ and Peter J.~Thomas$^{\displaystyle 1-\displaystyle 4}$}\\
Case Western Reserve University Departments of 
$^{\displaystyle 1}$Mathematics, Applied Mathematics, and Statistics,
{$^{\displaystyle 2}$Biology, {$^{\displaystyle 3}$Cognitive Science, and {$^{\displaystyle 4}$Electrical, Computer, and Systems Engineering.}\\
}}


\thispagestyle{empty}
%
\ \vspace{-0mm}\\
\begin{center} {\bf Abstract} \end{center}
Fox and Lu introduced a Langevin framework for discrete-time stochastic models of randomly gated ion channels such as the Hodgkin-Huxley (HH) system. 
They derived a Fokker-Planck equation with state-dependent diffusion tensor $D$ and suggested a Langevin formulation with noise coefficient matrix $S$ such that $SS^\intercal=D$. 
Subsequently, several authors introduced a variety of Langevin equations for the HH system. 
In this paper, we present a natural 14-dimensional dynamics for the HH system 
in which each \emph{directed} edge in the ion channel state transition graph acts as an independent noise source, leading to a $14\times 28$ noise coefficient matrix $S$. 
We show that 
(i) the corresponding 14D  system of ordinary differential \rev{equations} is consistent with the classical 4D representation of the HH system; 
(ii) the 14D representation leads to a noise coefficient matrix $S$ that can be obtained cheaply on each timestep, without requiring a matrix decomposition; 
(iii) sample trajectories of the 14D representation are pathwise equivalent to trajectories of Fox and Lu's system, as well as trajectories of several existing Langevin models; 
(iv) our 14D representation (and those equivalent to it) give the most accurate interspike-interval distribution, not only with respect to moments but under both the $L_1$ and $L_\infty$ metric-space norms; 
and
(v) the 14D representation gives an approximation to exact Markov chain simulations that are as fast and as  efficient as all equivalent models.     
Our approach goes beyond existing models, in that it supports a stochastic shielding decomposition that dramatically simplifies $S$ with minimal loss of accuracy under both voltage- and current-clamp conditions.

\section{Introduction}

Many natural phenomena exhibit stochastic fluctuations arising at the molecular scale,  the effects of which impact macroscopic quantities.  Understanding when and how microscale fluctuations will significantly contribute to macroscale behavior is a fundamental problem spanning the sciences. 
The impact of random ion channel fluctuations on the timing of action potentials in nerve cells provides an important example. 
Channel noise can have a significant effect on spike generation \citep{Mainen1995Science, SchneidmanFreedmanSegev1998NECO}, propagation along axons \citep{Faisal2007PLoS}, and spontaneous (ectopic) action potential generation in the absence of stimulation \citep{ODonnel_van_Rossum2015NECO}.
At the network level, channel noise can drive endogenous variability of vital rhythms such as  respiratory activity \citep{Yu2017JNeuro}.

Hodgkin and Huxley's quantitative model for active sodium and potassium currents producing action potential generation in the giant axon of \textit{Loligo} \citep{jp:Hodgkin+Huxley:1952d} suggested an underlying system of gating variables consistent with a multi-state Markov process description \citep{HillChen1972BPJ}.
The discrete nature of individual ion channel conductances was confirmed experimentally  \citep{Neher1976Nature}.  Subsequently, numerical studies of high-dimensional discrete-state, continuous-time Markov chain models produced insights into the effects of fluctuations in discrete ion channel populations on action potentials 
 \citep{SkaugenWalloe1979ActaPhysiolScand,StrassbergDeFelics1993NC}, aka \emph{channel noise}  \citep{White1998APSB, White2000Elsevier}. 

In the standard molecular-level HH model, which we adopt here, the \K~channel comprises four identical ``$n$'' gates that open and close independently, giving a five-vertex channel-state diagram with eight directed edges; the channel conducts a current only when in the rightmost state (Fig.~\ref{plot:HHNaKgates}, top).  The \Na~channel comprises three identical ``$m$" gates and a single ``$h$" gate, all independent, giving an eight-vertex diagram with twenty directed edges, of which one is conducting (Fig.~\ref{plot:HHNaKgates}, bottom).

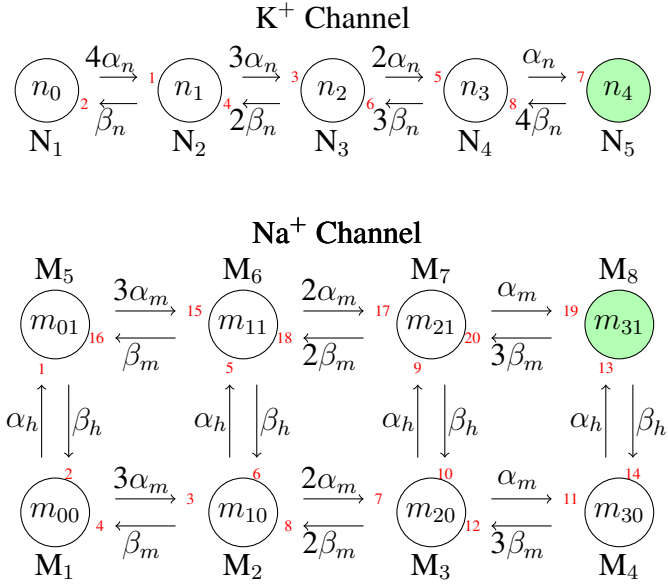
\begin{figure}
    \centering

  \begin{tikzpicture}
 
 \foreach \x/\w/\y/\z/\num/\id in {1/ /1/2/4/0,2/2/3/4/3/1,3/3/5/6/2/2,4/4/7/8/ /3}
 {
 \draw (1.9*\x,0) circle (12pt);
 \draw [color=black,->] (1.9*\x+0.7,5pt) -- (1.9*\x+1.2,5pt) node[right, color=red]{\tiny \y};
 \node at (1.9*\x+0.85,12pt) {\num$\alpha_n$};
 \draw [color=black,<-] (1.9*\x+0.7,-5pt) node[left, color=red]{\tiny \z} -- (1.9*\x+1.2,-5pt);
 \node at (1.9*\x+0.85,-12pt) {\w$\beta_n$};
 \node at (1.9*\x,-20pt) {$\text{N}_\x$};
 \node at (1.9*\x,0) {$n_\id$}; 
 }
 \draw [fill=green!30] (1.9*5,0) circle (12pt);
 \node  at (1.9*5,-20pt) {$\text{N}_5$} node at (1.9*5,0) {$n_4$} node  at (1.9*3,1) {$\text{K}^+\  \text{Channel}$};
 \end{tikzpicture}
\vspace{0.5cm}%

  \begin{tikzpicture}
 \foreach \x/\w/\y/\z/\num/\p in {1/ /3/4/3/00,2/2/7/8/2/10,3/3/11/12/ /20}
 {
 \draw (2.5*\x,0) circle (13pt) node at (2.5*\x,0) {$m_{\p}$};
 \draw [color=black,->] (2.5*\x+0.8,0.18) -- (2.5*\x+1.6,0.18) node[right, color=red]{\tiny \y};
 \node at (2.5*\x+1.15,12*0.036) {\num$\alpha_m$};
 \draw [color=black,<-] (2.5*\x+0.8,-0.18) node[left, color=red]{\tiny \z} -- (2.5*\x+1.6,-0.18);
 \node at (2.5*\x+1.15,-12*0.036) {\w$\beta_m$};
 \node at (2.5*\x,-0.72) {$\text{M}_\x$} node  at (2.5*2.5,2.5+1.25) {$\text{Na}^+\  \text{Channel}$};
 }
 \draw (2.5*4,0) circle (13pt);
 \node at (2.5*4,-0.72) {$\text{M}_4$} node at(2.5*4,0) {$m_{30}$};
 \foreach \x/\w/\y/\z/\num/\a/\p in {1/ /15/16/3/5/01,2/2/17/18/2/6/11,3/3/19/20/ /7/21}
 {
 \draw (2.5*\x,2.5) circle (13pt) node at (2.5*\x,2.5) {$m_{\p}$};
 \draw [color=black,->] (2.5*\x+0.8,0.18+2.5) -- (2.5*\x+1.6,0.18+2.5) node[right, color=red]{\tiny \y};
 \node at (2.5*\x+1.15,12*0.036+2.5) {\num$\alpha_m$};
 \draw [color=black,<-] (2.5*\x+0.8,-0.18+2.5) node[left, color=red]{\tiny \z} -- (2.5*\x+1.6,-0.18+2.5);
 \node at (2.5*\x+1.15,-12*0.036+2.5) {\w$\beta_m$};
 \node at (2.5*\x,0.72+2.5)  {$\text{M}_\a$};
 }
 \draw [fill=green!30] (2.5*4,2.5) circle (13pt);
 \node  at (2.5*4,0.72+2.5) {$\text{M}_8$} node at (2.5*4,2.5) {$m_{31}$};
 \foreach \x/\a/\b in {1/1/2,2/5/6,3/9/10,4/13/14}
 {
  \draw [->] (2.5*\x-0.18,0.73) -- (2.5*\x-0.18,0.7+1) node[above, color=red]{\tiny \a};
  \node at (2.5*\x-12*0.036,1.2) {$\alpha_h$};
  \draw [->] (2.5*\x+0.18,0.7+1) -- (2.5*\x+0.18,0.73) node[below, color=red]{\tiny \b};
  \node at (2.5*\x+12*0.036,1.2) {$\beta_h$};
 }
 \end{tikzpicture}
 
\caption{Molecular potassium (\K) and \rev{sodium} (\Na) channel states for the Hodgkin-Huxley model.  Filled circles mark conducting states $n_4$ and $m_{31}$. Per capita transition rates for each directed edge ($\alpha_n$, $\beta_n$, $\alpha_m$, $\beta_m$, $\alpha_h$ and $\beta_h$) are voltage dependent (cf.~eqns.~\eqref{eq:rate4}-\eqref{eq:rate9}). Directed edges are numbered 1-8 (\K~channel) and 1-20 (\Na-channel), marked in small red numerals.}
\label{plot:HHNaKgates}
\end{figure}
Discrete-state channel noise models are numerically intensive, whether implemented using discrete-time binomial approximations to the underlying continuous-time Markov process \citep{SkaugenWalloe1979ActaPhysiolScand,SchmandtGalan2012PRL} or continuous-time hybrid Markov models with exponentially distributed state transitions and continuously varying membrane potential.  The latter were introduced by 
\citep{ClayDeFelice1983BiophysJ} and are in principle exact \citep{AndersonErmentroutThomas2015JCNS}.
Under voltage-clamp conditions the hybrid conductance-based model reduces to a time-homogeneous Markov chain \citep{Colquhoun1981} that can be simulated using standard methods such as Gillespie's exact algorithm \citep{Gillespie1977,gillespie2007stochastic}.  
Even with this simplification, such Markov Chain (MC) algorithms are numerically expensive to simulate with realistic population sizes of 1000s of channels or greater.
Therefore, there is an ongoing need for efficient and accurate approximation methods. 

Following Clay and DeFelice's exposition of continuous time Markov chain implementations, \citep{FoxLu1994PRE} introduced a Fokker-Planck equation (FPE) framework that captured the first and second order statistics of HH ion channel populations in a 14-dimensional representation. 
Taking into account conservation of probability, one needs four variables to represent the population of \K~channels, seven for \Na, and one for voltage, leading to a 12-dimensional state space description.  
The resulting high-dimensional partial differential equation is impractical to solve numerically. However, as Fox and Lu observed, ``to every Fokker-Planck description, there is associated a Langevin description" \citep{FoxLu1994PRE}.  They therefore introduced a Langevin stochastic differential equation of the form:
\begin{eqnarray}
C \frac{dV}{dt}&=&I_\text{app}(t)-\bar{g}_\text{Na}\mbM_8\left(V-V_\text{Na}\right)-\bar{g}_\text{K}\mbN_5\left(V-V_\text{K}\right)-g_\text{leak}(V-V_\text{leak}), \label{eq:FoxandLu_dv}\\
\frac{d\mbM}{dt}&=&A_\text{Na}\mbM+S_1\xi_1, \label{eq:FoxandLu_dNa}\\
\frac{d\mbN}{dt}&=&A_\text{K}\mbN+S_2\xi_2, \label{eq:FoxandLu_dK}
\end{eqnarray}
where $C$ is the capacitance, $I_\text{app}$ is the applied current, maximal conductances are denoted $\bar{g}_\text{ion}$, with $V_\text{ion}$ being the associated reversal potential,  and  ohmic leak current $g_\text{leak}(V-V_\text{leak})$. 
$\mbM\in\R^8$ and $\mbN\in\R^5$ are  vectors for the fractions of \Na~and \K~channels in each state, with $\mbM_8$ representing the open channel 
fraction for \Na, and $\mbN_5$ the open channel fraction for \K~(Fig.~\ref{plot:HHNaKgates}). 
Vectors $\xi_1(t)\in\R^8$ and $\xi_2(t)\in\R^5$ are independent Gaussian white noise processes with
zero mean and unit variances 
$\langle\xi_1(t)\xi_1^\intercal(t')\rangle=I_8\,\delta(t-t')$ and
$\langle\xi_2(t)\xi_2^\intercal(t')\rangle=I_5\,\delta(t-t')$. 
The state-dependent rate matrices  $A_\text{Na}$ and $A_\text{K}$ are given in eqns.~\eqref{matrix:ANa} and \eqref{Matrix:AK}. In Fox and Lu's formulation, $S$ must satisfy $S=\sqrt{D}$, where $D$ is a symmetric, positive semi-definite $k\times k$ ``diffusion matrix" (see Appendix \ref{app:SNaSK} for the $D$ matrices for the standard HH \K~and \Na~channels). 
\rev{We will refer to the  14-dimensional Langevin equations  \eqref{eq:FoxandLu_dv}-\eqref{eq:FoxandLu_dK}, with $S=\sqrt{D}$, as the ``Fox-Lu" model.}

The original \rev{Fox-Lu model, later called} the ``conductance noise model" by \citep{GoldwynSheaBrown2011PLoSComputBiol}, did not see widespread use until gains in computing speed made the square root calculations more feasible. Seeking a more efficient approximation, \rev{\citep{FoxLu1994PRE} also introduced a four-dimensional  Langevin version of  the  HH model.} \rev{This model was systematically studied in} \citep{Fox1997BiophysicalJournal} 
which can be written as follows:
\begin{eqnarray}
C \frac{dV}{dt}&=&I_\text{app}(t)-\bar{g}_\text{Na}m^3h\left(V-V_\text{Na}\right)-\bar{g}_\text{K}n^4\left(V-V_\text{K}\right)-g_\text{leak}(V-V_\text{leak}) \label{eq:Fox_dv}\\
\frac{dx}{dt}&=&\alpha_x(1-x)-\beta_xx+\xi_x(t), \ \text{where}\ x=m,h,\ \text{or},\ n.\label{eq:Fox_dx}
\end{eqnarray}
where $\xi_x(t)$ are Gaussian processes with covariance function
\begin{equation}\label{eq:covariance_fcn}
E[\xi_x(t),\xi_x(t')]=\frac{\alpha_x(1-x)+\beta_xx}{N}\delta(t-t').
\end{equation}
Here $N$ represents the total number of \Na channels (respectively, the total number of \K channels) and $\delta(\cdot)$ is the Dirac delta function. 
This model, referred as the ``subunit noise model" by \citep{GoldwynSheaBrown2011PLoSComputBiol}, has been widely used as an approximation to MC ion channel models (see references in \cite{Bruce2009ABE, GoldwynSheaBrown2011PLoSComputBiol}). 
\rev{For example, \cite{SchmidGoychukHanggi2001EPL} used this approximation to investigate stochastic resonance and coherence resonance in forced and unforced versions of the HH model (e.g.~in the excitable regime).}
However, the numerical accuracy of this method was criticized by several studies \citep{MinoRubinsteinWhite2002AnnBiomedEng, Bruce2009ABE}, which found that its accuracy does not improve even with increasing numbers of channels.

Although more accurate approximations based on Gillespie's algorithm (using a piecewise constant propensity approximation,  \cite{Bruce2009ABE,MinoRubinsteinWhite2002AnnBiomedEng}) and even based on exact simulations \citep{ClayDeFelice1983BiophysJ,NewbyBressloffKeener2013PRL,AndersonErmentroutThomas2015JCNS} became available, they remained prohibitively expensive for large network simulations.  Meanwhile, Goldwyn and Shea-Brown's rediscovery of Fox and Lu's earlier conductance based model \citep{GoldwynSheaBrown2011PLoSComputBiol,Goldwyn2011PRE} launched a flurry of activity seeking the best Langevin-type approximation.  
\cite{GoldwynSheaBrown2011PLoSComputBiol} introduced a faster decomposition algorithm to simulate equations \eqref{eq:FoxandLu_dv}-\eqref{eq:FoxandLu_dK}, and showed that Fox and Lu's method accurately captured the fractions of open channels and the inter-spike intervla (ISI) statistics, in comparison with Gillespie-type Monte Carlo (MC) simulations.  However, despite the development of efficient singular value decomposition based algorithms for  solving $S=\sqrt{D}$, this step still causes a bottleneck in the algorithms  based on \citep{FoxLu1994PRE,GoldwynSheaBrown2011PLoSComputBiol, Goldwyn2011PRE}.

Many variations on Fox and Lu's 1994 Langevin model have been proposed in recent years \citep{Dangerfield2010PCS, Linaro2011PublicLibraryScience, Dangerfield2012APS, OrioSoudry2012PLoS1, Guler2013MITPress, Huang2013APS, Pezo2014Frontiers, Huang2015PhBio, Fox2018arXiv} including Goldwyn et al's work \citep{GoldwynSheaBrown2011PLoSComputBiol, Goldwyn2011PRE},
each with its own strengths and weaknesses.
One class of methods imposes  projected boundary conditions  \citep{Dangerfield2010PCS,Dangerfield2012APS}; as we will show in \S \ref{sec:modelcomp}, this approach leads to inaccurate interspike interval distribution, and is inconsistent with a natural multinomial invariant manifold structure for the ion channels. 
Several methods implement correlated noise at the subunit level, as in \eqref{eq:Fox_dx}-\eqref{eq:covariance_fcn} \citep{Fox1997BiophysicalJournal,Linaro2011PublicLibraryScience,Guler2013NC,Guler2013MITPress}.  However, if one recognizes that, at the molecular level, the \emph{individual directed edges} represent the independent noise sources in ion channel dynamics, then the approach incorporating noise at the subunit level obscures the biophysical origin of ion channel fluctuations. 
Some methods introduce the noisy dynamics at the level of edges rather than nodes, but lump reciprocal edges together into pairs  \citep{OrioSoudry2012PLoS1,Dangerfield2012APS,Huang2013APS,Pezo2014Frontiers}.  This approach implicitly assumes, in effect, that the ion channel probability distribution satisfies a  detailed balance (or microscropic reversibility) condition. 
However, while detailed balance holds for the HH model under stationary voltage clamp, this condition is violated during active spiking. 
Finally, the stochastic shielding approximation \citep{SchmandtGalan2012PRL,SchmidtThomas2014JMN,SchmidtGalanThomas2018PLoSCB} does not have a natural formulation in the representation associated with an $n\times n$ noise coefficient matrix $S$; in the cases of rectangular $S$ matrices used in  \citep{OrioSoudry2012PLoS1,Dangerfield2012APS} stochastic shielding can only be applied to reciprocal pairs of edges. We will elaborate on these points in \S \ref{sec:discussion}.

In this paper, we introduce a new variation of Fox and Lu's conductance-noise model that avoids the limitations described above.  We show that preserving each directed edge in the channel transition graph (Fig.~\ref{plot:HHNaKgates}) as an independent noise source leads to a natural, biophysically motivated Langevin model that does not require any matrix decomposition step.  Our construction lends itself to direct application of stochastic shielding methods, leading to faster simulations that retain the accuracy of Fox and Lu's method. 

As an additional benefit, our method answers an open question in the literature, arising from the fact that the decomposition $D=SS^\intercal$  is not unique. 
As Fox recently pointed out, sub-block determinants of the $D$ matrices play a major role in the structure of the $S$ matrix elements.  \citep{Fox2018arXiv} conjectured that ``a universal form for $S$ may exist". In this paper we obtain the universal form for the noise coefficient matrix $S$. Moreover, we prove that our model is equivalent to Fox and Lu's 1994 model in the strong sense of \emph{pathwise equivalence}.

The remainder of the paper is organized as follows. In \S  \ref{sec:determ_14D}, we review the canonical \emph{deterministic} 14D  version of the HH model.  We prove a series of lemmas which show (1) the multinomial submanifold $\mathcal{M}$ is an invariant manifold within the 14D space and (2) the velocity on the 14D space and the pushforward of the velocity on the 4D space are identical.  Moreover, we show (numerically) that (3) the submanifold $\mathcal{M}$ is globally attracting, even under current clamp conditions.
\rev{Fig.~\ref{fig:HHcomp} illustrates the relationship between the 4D and 14D deterministic HH models.}
\S  \ref{sec:stochastic_14DHH} lays out our $14\times 28$ Langevin HH model.  Like \citep{OrioSoudry2012PLoS1,Dangerfield2012APS,Pezo2014Frontiers}, we avoid matrix decomposition by computing $S$ directly. The key difference between our approach and its closest relative \citep{Pezo2014Frontiers} is to use a rectangular $n\times k$ matrix $S$  for which \emph{each directed edge} is treated as an independent noise source, rather than lumping reciprocal edges together in pairs. 
In the new Langevin model, the form of our $S$ matrix reflects the biophysical origins of the underlying channel noise, and  allows us to apply the stochastic shielding approximation by neglecting the noise on selected individual directed edges.  
As we prove in \S \ref{sec:path_equiv}, our model (without the stochastic shielding approximation) is pathwise equivalent to all those in a particular class of biophysically derived Langevin models, including those used in \citep{FoxLu1994PRE, Goldwyn2011PRE, GoldwynSheaBrown2011PLoSComputBiol,OrioSoudry2012PLoS1,Pezo2014Frontiers,Fox2018arXiv}.
\rev{In addition to 4D and 14D deterministic trajectories, Fig.~\ref{fig:HHcomp} also shows a stochastic trajectory generated by our Langevin model.
Finally, we} compare our Langevin model to several alternative stochastic neural models in terms of accuracy (of the full ISI distribution) and numerical efficiency in \S  \ref{sec:modelcomp}. 

\section{The Deterministic 4D and 14D HH Models}
\label{sec:determ_14D}
In this section, we review the classical four-dimensional model of  \cite{jp:Hodgkin+Huxley:1952d} (HH), as well as its natural fourteen-dimensional version (\cite{Dayan+Abbott:2001}, \S5.7), with variables comprising membrane voltage and the occupancies of five potassium channel states and eight sodium channel states. 
The deterministic 14D model is the mean field of \rev{the} channel-based Langevin model proposed by \citep{FoxLu1994PRE}\rev{; this paper describes both the Langevin and the mean field versions of the 14D Hodgkin-Huxley system.}
For completeness of exposition, we briefly review the 4D deterministic HH system and its 14D deterministic counterpart.  In \S\ref{sec:path_equiv} we will prove that the sample paths of a class of Langevin stochastic HH models are equivalent; in 
\S\ref{subsec:14D4D}  
we review  analogous results relating trajectories of the 4D and 14D deterministic ODE systems.

In particular, we will show that the deterministic 14D model and the original 4D HH model are dynamically equivalent, in the sense that every flow (solution) of the 4D model corresponds to a flow of the 14D model.  
The consistency of trajectories between of the 14D and 4D models is easy to verify for initial data on a 4D submanifold of the 14D space given by choosing multinomial distributions for the gating variables \citep{Dayan+Abbott:2001,Goldwyn2011PRE}.  
Similarly, Keener established results on multinomial distributions as invariant submanifolds of Markov models with ion channel kinetics under several circumstances \citep{Keener:2009:JMathBiol,Keener:2010:JMathBiol,EarnshawKeener2010bSIADS,EarnshawKeener2010SIADS_2}, but without treating the general current-clamped case.
Consistent with these results, we show below that the set of all 4D flows maps to an invariant submanifold of the state space of the 14D model.  Moreover, we show numerically that solutions of the 14D model with arbitrary initial conditions converge to this submanifold.  \rev{Thus the original HH model ``lives inside" the 14D deterministic model in the sense that the former embeds naturally and consistently within the latter (cf.~Fig.~\ref{fig:HHcomp}).}

In the stochastic case, the 14D model has a natural interpretation as a hybrid stochastic system with independent noise forcing along each edge of the potassium (8 directed edges) and sodium (20 directed edges) channel state transition graphs.  
The hybrid model leads naturally to a biophysically grounded Langevin model that we describe in section \S \ref{sec:stochastic_14DHH}.  
In contrast to the ODE case, the stochastic versions of the 4D and 14D models are \emph{not} equivalent \citep{GoldwynSheaBrown2011PLoSComputBiol}.  
\begin{figure}
    \centering
    \includegraphics[scale=0.26]{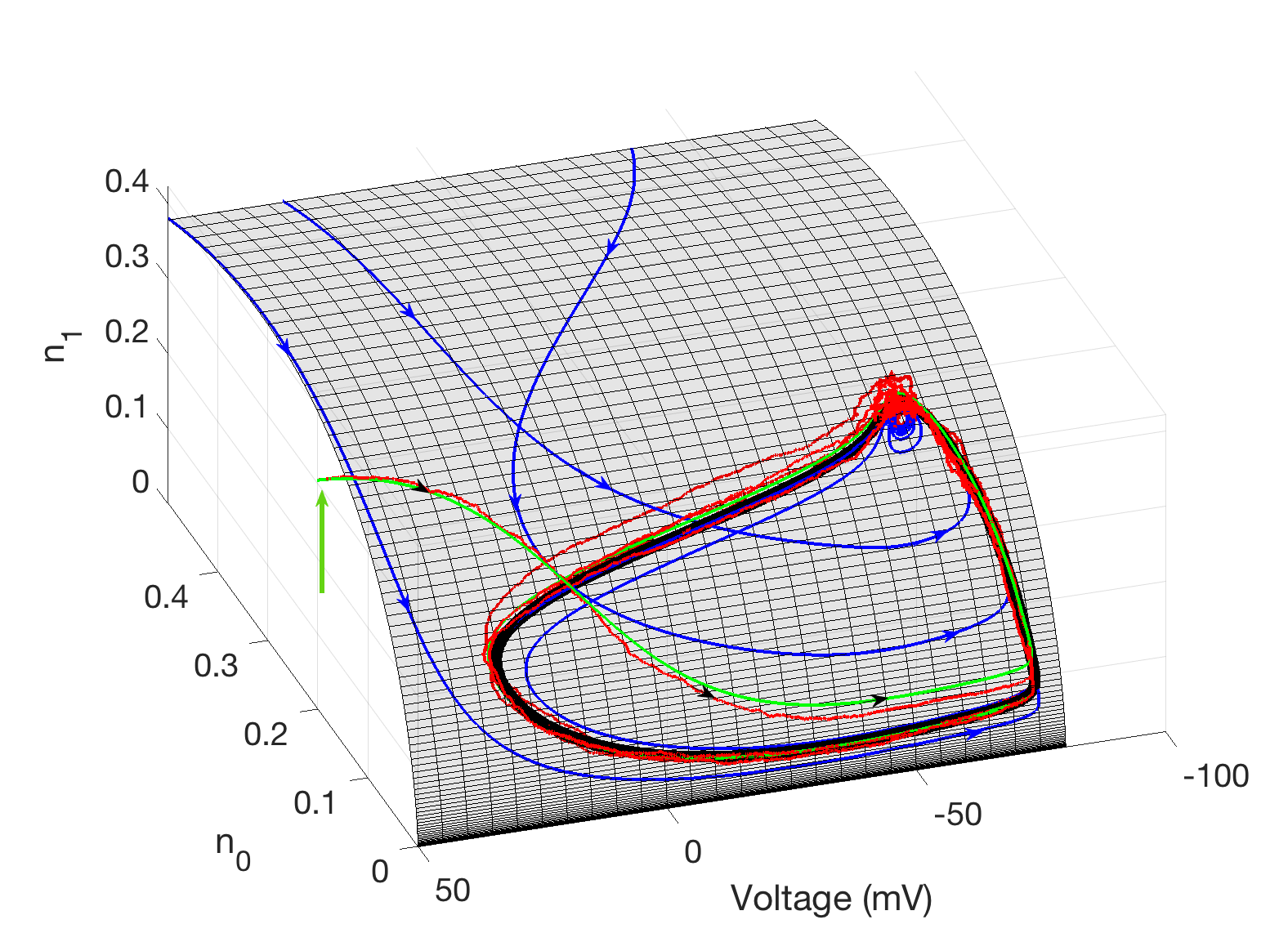}
    \caption{\rev{4D and 14D HH models. The meshed surface is a three dimensional projection of the 14D state space onto three axes representing the voltage, $v$, the probability of all four potassium gates being in the closed state, $n_0$, and the probability of exactly one potassium gate being in the open state, $n_1$. 
    \textbf{Blue curves:} Trajectories of the deterministic 14D HH model with initial conditions located on the 4D multinomial invariant submanifold, $\mathcal{M}$.  We prove that $\mathcal{M}$ is an invariant submanifold in \S \ref{subsec:14D4D}.  
    \textbf{Black curve:} The deterministic limit cycle solution for the 14D HH model, which forms a closed loop within $\mathcal{M}$.
     \textbf{Green curve:} A trajectory of the deterministic 14D HH model with initial conditions (vertical green arrow) off the multinomial submanifold.
     \textbf{Red curve:} A trajectory of the stochastic 14D HH model (cf.~\S\ref{sec:stochastic_14DHH}) with the same initial conditions as the green trajectory.
     The blue and black arrows mark the directions of the trajectories.
    Note that trajectories starting away from $\mathcal{M}$ converge to $\mathcal{M}$; and \emph{all} deterministic trajectories converge to the deterministic limit cycle.
     Parameters of the simulation are given in Tab.~\ref{tab:parameters}.}}
    \label{fig:HHcomp}
\end{figure}

\subsection{The 4D Hodgkin-Huxley Model}
The 4D voltage-gated ion \rev{channel} HH model is a set of four ordinary differential equations
\begin{align}\label{eq:HH4D_V}
 C\frac{dv}{dt}&=-\bar{g}_{\text{Na}}m^3h(v-V_{\text{Na}})-\bar{g}_{\text{K}}n^4(v-V_\text{K})-g_\text{L}(v-V_\text{L})+I_{\text{app}}, \\
    \frac{dm}{dt}&=\alpha_m(v)(1-m)-\beta_m(v)m, \label{eq:rate1}  \\
    \frac{dh}{dt}&=\alpha_h(v)(1-h)-\beta_h(v)h, \label{eq:rate2}  \\
     \frac{dn}{dt}&=\alpha_n(v)(1-n)-\beta_n(v)n, \label{eq:rate3}
\end{align}
where $v$ is the membrane potential, $I_\text{app}$ is the applied current, and  $0\le m,n,h \le1$ are dimensionless gating variables associated with \Na~and \K~ channels. 
The constant $\bar{g}_{\text{ion}}$ is the maximal value of the conductance for the sodium and potassium channel, respectively. 
Parameters $V_\text{ion}$ and $C$ are the ionic reversal potentials and capacitance, respectively.
The quantities $\alpha_x$ and $\beta_x$, $x\in\{m,n,h\}$ are \rev{the voltage-dependent per capita transition rates,} defined in Appendix \ref{axppend_alpha_beta}.

This system is a $C^\infty$ vector field on a four-dimensional manifold (with boundary) contained in $\R^4$: $\mathcal{X}=\{-\infty<v<\infty,0\le m,h,n \le 1\}=\R\times[0,1]^3.$  The manifold is forward and backward invariant in time.  If $I_\text{app}$ is constant then $\mathcal{X}$ has an invariant subset given by $\mathcal{X}\cap\{v_\text{min}\le v\le v_\text{max}\}$, where $v_\text{min}$ and $v_\text{max}$ are calculated in  Lemma~\ref{Lemma:vmin_vmax}.

As pointed out by (\cite{Keener1998MathPhy}, \S 3, p.~106) and \citep{Keener:2009:JMathBiol}, for voltage either fixed or given as a prescribed function of time, the equations for $m, h$ and $n$ can be interpreted as the parametrization of an invariant manifold embedded in a higher-dimensional time-varying Markov system.  Several papers developed this idea for a variety of ion channel models and related systems \citep{Keener:2009:JMathBiol, EarnshawKeener2010SIADS_2}
but the theory developed is restricted to the voltage-clamped case.  

Under fixed voltage clamp, the ion channels form a time-homogeneous Markov process with a unique (voltage-dependent) stationary probability distribution.  Under \emph{time-varying} current clamp the ion channels nevertheless form a Markov process, albeit no longer time-homogeneous. Under these conditions the ion channel state converges rapidly to a multinomial distribution indexed by a low-dimensional set of time-varying parameters ($m(t),h(t),n(t)$) \citep{Keener:2010:JMathBiol}. 
\label{page:multinomial_mentioned_first}
In the current-clamped case, the ion channel process, considered alone, is neither stationary nor Markovian, making the analysis of this case significantly more challenging, from a mathematical point of view.  

\subsection{The Deterministic 14D Hodgkin-Huxley Model}
For the HH kinetics given in Fig.~\ref{plot:HHNaKgates}  (on page \pageref{plot:HHNaKgates}), we define the eight-component state vector $\mbM$ for the \Na~gates, and the five-component state vector  $\mbN$ for the \K~gates, respectively, as
\begin{align}\label{eq:define_M}
\mbM&=[m_{00},m_{10},m_{20},m_{30},m_{01},m_{11},m_{21},m_{31}]^\intercal \in [0,1]^{8} \\
\label{eq:define_N}
\mbN&=[n_0,n_1,n_2,n_3,n_4]^\intercal\in [0,1]^5,
\end{align}
where $\sum_{i=0}^3\sum_{j=0}^1 m_{ij}=1$ and $\sum_{i=0}^4 n_i=1$.
The open probability for the \Na~channel is  $\mbM_8=m_{31}$, and is $\mbN_5=n_4$ for the \K~channel.
The deterministic 14D HH equations may be written (compare \eqref{eq:HH4D_V}-\eqref{eq:rate3})
\begin{eqnarray}\label{14dhh1}
C\frac{dV}{dt}&=&-\bar{g}_{\text{Na}}\text{M}_8(V-V_{\text{Na}})-\bar{g}_{\text{K}}\text{N}_5(V-V_\text{K})-g_\text{L}(V-V_\text{L})+I_\text{app},\\
\frac{d\mbM}{dt}&=&A_\text{Na}(V)\mbM, \label{14dhh2}\\
\frac{d\mbN}{dt}&=&A_\text{K}(V)\mbN,\label{14dhh3}
\end{eqnarray}
where the voltage-dependent drift matrices $A_\text{Na}$ and $A_\text{K}$ are given by
 {\footnotesize{\begin{equation}
 A_\text{Na}(V) =\begin{bmatrix}
A_\text{Na}(1) & \beta_m&0 &0 &\beta_h&0&0&0\\
3\alpha_m&A_\text{Na}(2)&2\beta_m&0&0&\beta_h&0&0\\
0&2\alpha_m&A_\text{Na}(3) &3\beta_m&0&0&\beta_h&0 \\
0&0&\alpha_m&A_\text{Na}(4)&0&0&0&\beta_h \\
\alpha_h&0&0&0&A_\text{Na}(5)&\beta_m&0&0\\
0&\alpha_h&0&0&3\alpha_m&A_\text{Na}(6)&2\beta_m&0\\
0&0&\alpha_h&0&0&2\alpha_m&A_\text{Na}(7)&3\beta_m\\
0&0&0&\alpha_h&0&0&\alpha_m&A_\text{Na}(8)\\
\end{bmatrix},
\label{matrix:ANa}
\end{equation}

\begin{equation}
   A_\text{K}(V) =\begin{bmatrix}
   A_\text{K}(1)& \beta_n(V)             & 0                & 0                  & 0\\
   4\alpha_n(V)& A_\text{K}(2)&   2\beta_n(V)              & 0&                   0\\
    0&        3\alpha_n(V)&        A_\text{K}(3)& 3\beta_n(V)&          0\\
    0&        0&               2\alpha_n(V)&          A_\text{K}(4)& 4\beta_n(V)\\
    0&        0&               0&                 \alpha_n(V)&          A_\text{K}(5)
\end{bmatrix},
\label{Matrix:AK}
\end{equation}}}
and the diagonal elements $$A_\text{ion}(i)=-\sum_{j\::\:j\neq i}A_\text{ion}(j,i),\text{ for ion }\in\{\text{Na},\text{K}\}.$$

\subsection{Relation Between the 14D and 4D Deterministic HH Models}\label{subsec:14D4D}
\cite{EarnshawKeener2010SIADS_2} suggests that it is reasonable to expect that the global flow of the 14D system should converge to \rev{the} 4D submanifold but also that it is far from obvious that it must.  Existing theory applies to the voltage-clamped case \cite{Keener:2009:JMathBiol, EarnshawKeener2010SIADS_2}.  Here, we consider instead the current-clamped case, in which the fluctuations of the ion channel state influences the voltage evolution, and {\it{vice-versa}}.

In the remainder of this section  we will (1) define a \emph{multinomial submanifold} $\mathcal{M}$ and show that it is an invariant manifold within the 14D space, and (2)  show that the velocity on the 14D space and the pushforward of the velocity on the 4D space are identical.  \rev{In \S \ref{subsec:converge} we will} (3) provide numerical evidence that $\mathcal{M}$ is globally attracting within the higher-dimensional space.  

In order to compare the trajectories of the 14D HH equations with trajectories of the standard 4D equations, we define lower-dimensional and higher-dimensional domains $\mathcal{X}$ and $\mathcal{Y}$, respectively, as
\begin{align}\nonumber
    \mathcal{X}&=\{-\infty<v<\infty,0\le m\le 1,0\le h\le 1,0\le n\le 1\} =  
    \R\times[0,1]^3
    \subset \R^{4}\\
    \mathcal{Y}&=\{-\infty<v<\infty\}
    \cap 
    \rev{
    \left\{0\le m_{ij},\,  \,\sum_{i=0}^3\sum_{j=0}^1 m_{ij}=1 \right\}
    \cap
    \left\{0\le n_i,\,\,  \sum_{i=0}^4n_i=1 \right\}
    }\nonumber 
    \\
    &=\R\times\Delta^7\times \Delta^4 \subset \R^{14}, 
    \label{eq:XYdefinitions}
\end{align}
where $\Delta^k$ is the $k$-dimensional simplex in $\R^{k+1}$ given by $y_1+\ldots+y_{k+1}=1, y_i\ge 0.$  
The 4D HH model $\frac{dx}{dt}=F(x)$, equations \eqref{eq:HH4D_V}-\eqref{eq:rate3}, is defined for $x\in \mathcal{X}$, and the 14D HH model $\frac{dy}{dt}=G(y)$, equations \eqref{14dhh1}-\eqref{14dhh3}, is defined for $y\in \mathcal{Y}$.
\begin{table}[htbp]\centering
   \begin{tabular}{cc} 
   14D model  & 4D model \\
   \hline
    $(v,m_{00},\ldots,m_{31},n_0,\ldots,n_4)$& $(v,m,h,n)$ 
   \\ \hline \hline
   $v$ & $v$\\
   $\frac13 (m_{11}+m_{10})+\frac23 (m_{21}+m_{20})+m_{31}+m_{30}$ &  $m$ \\
   $m_{01}+m_{11}+m_{21}+m_{31}$ & $h$\\
   $n_1/4+n_2/2+3n_3/4+n_4$ & $n$
`   \\ \hline \hline
   \end{tabular}
   \caption{$R$: Map from the 14D HH model $(m_{00},\ldots,m_{31},n_0,\ldots,n_4)$ to the 4D HH model $(m,h,n)$. Note that $\{m_{00},\ldots,m_{31}\}$ and \rev{$\{n_0,\ldots,n_4\}$} both follow multinomial distributions.}
   \label{tab:14D-to-4D-for-HH}
\end{table}
We introduce a dimension-reducing mapping $R:\mathcal{Y}\to\mathcal{X}$ as in Table \ref{tab:14D-to-4D-for-HH}, and a mapping from lower to higher dimension, $H:\mathcal{X}\to\mathcal{Y}$ as in Table \ref{tab:4D-to-14D-for-HH}. 
We construct $R$ and $H$ in such a way that
$R\circ H$ acts as the identity on $\mathcal{X}$, that is,
for all $x\in\mathcal{X}$, \rev{$x=R(H(x))$}.  
The maps $H$ and $R$ are  consistent with a multinomial structure for the ion channel state distribution, in the following sense.  The space $\mathcal{Y}$ covers all possible probability distributions on the eight sodium channel states and the five potassium channel states.  Those distributions which are products of one multinomial distribution on the \K-channel \footnote{That is, distributions indexed by a single open probability $n$; with the five states having probabilities $\binom{4}{i}n^i(1-n)^{4-i}\text{ for }0\le i\le 4$.
} 
and a second multinomial distribution on the \Na-channel\footnote{That is, distributions indexed by two open probabilities $m$ and $h$, with the eight  states having probabilities $\binom{3}{i}m^i(1-m)^{3-i}h^j(1-h)^{1-j},\text{ for }0\le i\le 3,\text{ and }0\le j\le 1$.} form a submanifold of $\Delta^7\times\Delta^4$.  In this way we define a submanifold,  denoted $\mathcal{M}=H(\mathcal{X})$,  the image of $\mathcal{X}$ under $H$. 

\begin{table}[htbp]\centering
   \begin{tabular}{cc} 
   4D model  & 14D model \\
   \hline
   $(v,m,h,n)$ & $(v,m_{00},\ldots,m_{31},n_0,\ldots,n_4)$
   \\ 
   $v$ & $v$\\
   \hline \hline
   $(1-n)^4$ & $n_0$\\
   $4(1-n)^3n$ & $n_1$\\
   $6(1-n)^2n^2$ & $n_2$\\
   $4(1-n)n^3$ & $n_3$\\
   $n^4$ & $n_4$\\ 
   \hline \hline
   $(1-m)^3(1-h)$ & $m_{00}$\\
   $3(1-m)^2m(1-h)$ & $m_{10}$\\
   $3(1-m)m^2(1-h)$ & $m_{20}$\\
   $m^3(1-h)$ & $m_{30}$\\ \hline
   $(1-m)^3h$ & $m_{01}$\\
   $3(1-m)^2mh$ & $m_{11}$\\
   $3(1-m)m^2h$ & $m_{21}$\\
   $m^3h$ & $m_{31}$\\ \hline \hline
   \end{tabular}
   \caption{$H$: Map from the 4D HH model $(m,h,n)$ and the 14D HH model $(m_{00},\ldots,m_{31},n_0,\ldots,n_4)$.}
   \label{tab:4D-to-14D-for-HH}
\end{table}

Before showing that the multinomial submanifold $\mathcal{M}$ is an invariant manifold within the 14D space, we first show that \rev{the 
deterministic 14D} HH model is defined on a bounded domain. Having a bounded forward-invariant manifold is a general property of conductance-based models, which may be written in the form
\begin{align}\label{eq:generalCBM}
&\frac{dV}{dt}=f(V,\mathcal{N}_\text{open})=\frac{1}{C}\left\{I_\text{app}-g_\text{leak}(V-V_\text{leak})-\sum_{i\in\mathcal{I}}\left[g_i N_{\text{open}}^i(V-V_i)\right]\right\}\\
&\frac{d\mathcal{N}}{dt}=A(V)\mathcal{N} 
\label{eq:dndt}\text{ and }\\
\label{eq:pro}
&\mathcal{N}_\text{open}=\mathcal{O}[\mathcal{N}].
\end{align} 
Here, $C$ is the membrane capacitance,  $I_\text{app}$ is an applied current with upper and lower bounds $I_{\pm}$ respectively, and $g_i$ is the conductance for the $i$th ion channel. 
The index $i$ runs over the set of distinct  ion channel types in the model, $\mathcal{I}$.
The gating  vector $\mathcal{N}$ represents the fractions of each ion channel population in  various ion channel states, and the operator $\mathcal{O}$ gives the fraction of each ion channel population in the open (or conducting) channel states.
The following lemma establishes that any conductance-based model (including the 4D or 14D HH model) is defined on a bounded domain.
\begin{lemma}\label{Lemma:vmin_vmax}
For a conductance-based model of the form
\eqref{eq:generalCBM}-\eqref{eq:pro}, and for any bounded applied current $I_-\le I_\text{app}\le I_+$, there exist  upper and lower bounds $V_\text{max}$ and $V_\text{min}$ such that trajectories with initial voltage condition $V\in[V_\text{min},V_\text{max}]$ remain within this interval for all times $t>0$, regardless of the initial channel state.
\end{lemma}

\begin{proof}
Let $V_1=\underset{i\in \mathcal{I}}{\text{min}}\{V_i\}\land V_\text{leak}$, and 
$V_2=\underset{i\in \mathcal{I}}{\text{max}}\{V_i\}\lor V_\text{leak}$, where the index $i$ runs over $\mathcal{I}$, the set of distinct ion channel types. 
Note that for all $i$, $0\leq N_{\text{open}}^i\leq1$, and $g_i>0,\ g_\text{leak}>0$. Therefore  when $V\le V_1$
\begin{align}\label{eq:vmin1}
\frac{dV}{dt}&=\frac{1}{C}\left\{I_\text{app}-g_\text{leak}(V-V_\text{leak})-\sum_{i\in\mathcal{I}}\left[g_i N_{\text{open}}^i(V-V_i)\right]\right\}\\
\label{eq:vmin2}
&\ge  \frac{1}{C}\left\{I_\text{app}-g_\text{leak}(V-V_1)-\sum_{i\in\mathcal{I}}\left[g_i N_{\text{open}}^i(V-V_1)\right]\right\}\\
\label{eq:vmin3}
&\ge \frac{1}{C}\left\{I_\text{app}-g_\text{leak}(V-V_1)-\sum_{i\in\mathcal{I}}\left[g_i\times 0 \times(V-V_1)\right]\right\}\\
\label{eq:vmin4}
&=\frac{1}{C}\left\{I_\text{app}-g_\text{leak} (V-V_1)\right\}.
\end{align} 
Inequality \eqref{eq:vmin2} follows because $V_1 =\underset{i\in \mathcal{I}}{\text{min}}\{V_i\}\land V_\text{leak}$, and  inequality \eqref{eq:vmin3} follows because $V-V_1\le 0$, $g_i>0$ and $N^i_\text{open}\ge 0$. 
Let $V_\text{min}:=\text{min}\left\{\frac{I_-}{g_\text{leak}} +V_1,V_1\right\}$.  When $V<V_\text{min}$, $\frac{dV}{dt}>0$. Therefore,  
$V$ will not decrease beyond $V_\text{min}$. 

Similarly, when $V\ge V_2$
\begin{align}\label{eq:vmax1}
\frac{dV}{dt}&=\frac{1}{C}\left\{I_\text{app}-g_\text{leak}(V-V_\text{leak})-\sum_{i\in\mathcal{I}}\left[g_i N_{\text{open}}^i(V-V_i)\right]\right\}\\
\label{eq:vmax2}
&\le \frac{1}{C}\left\{I_\text{app}-g_\text{leak}(V-V_2)-\sum_{i\in\mathcal{I}}\left[g_i N_{\text{open}}^i(V-V_2)\right]\right\}\\
\label{eq:vmax3}
&\leq \frac{1}{C}\left\{I_\text{app}-g_\text{leak}(V-V_2)-\sum_{i\in\mathcal{I}}\left[g_i\times 0\times(V-V_2)\right]\right\}\\
\label{eq:vmax4}
&=\frac{1}{C}\left\{I_\text{app}-g_\text{leak} (V-V_2)\right\}.
\end{align} 
Inequality \eqref{eq:vmax2} holds because $V_2=\underset{i\in \mathcal{I}}{\text{max}}\{V_i\}\lor V_\text{leak},$ and inequality \eqref{eq:vmax3} holds because $V-V_2\ge 0$, $g_i>0$ and $N^i_\text{open}\ge 0$. 
Let $V_\text{max}=\text{max}\left\{\frac{I_\text{app}}{g_\text{leak}}+V_2, V_2 \right\}$. When $V>V_\text{max},$ $\frac{dV}{dt}<0$.   Therefore, $V$ will not go beyond $V_\text{max}$.  

We conclude that if $V$ takes an initial condition in the interval $[V_\text{min},V_\text{max}],$ then $V(t)$ remains within this interval for all $t\ge 0$.
\end{proof}

Given that $y\in\mathcal{Y}$ has a bounded domain, Lemma \ref{LemmaXinvari} follows directly, and establishes that the multinomial submanifold $\mathcal{M}$ is a forward-time--invariant manifold within the 14D space. The proof of Lemma \ref{LemmaXinvari} is  in Appendix \ref{append_LemmaXinvari}.

\begin{lemma}\label{LemmaXinvari} Let $\mathcal{X}$ and $\mathcal{Y}$ be the lower-dimensional and higher-dimensional Hodgkin-Huxley manifolds given by \eqref{eq:XYdefinitions}, and let $F$ and $G$ be the vector fields on $\mathcal{X}$ and $\mathcal{Y}$ defined by \eqref{eq:HH4D_V}-\eqref{eq:rate3} and \eqref{14dhh1}-\eqref{14dhh3}, respectively. 
Let $H:\mathcal{X}\to\mathcal{M}\subset\mathcal{Y}$ and $R:\mathcal{Y}\to\mathcal{X}$ be the mappings given in Tables  \ref{tab:4D-to-14D-for-HH} and \ref{tab:14D-to-4D-for-HH}, respectively, 
and define the multinomial submanifold $\mathcal{M}=H(\mathcal{X})$.  Then $\mathcal{M}$ is forward-time--invariant under the flow generated by $G$.  Moreover, the vector field $G$,  when restricted to $\mathcal{M}$, coincides  with the vector field induced by $F$ and the map $H$.  That is, for all $y\in\mathcal{M}$, $G(y)=D_x H(R(y))\cdot F(R(y)).$
\end{lemma}

Lemma \ref{LemmaXinvari} establishes that the  14D HH model given by \eqref{14dhh1}-\eqref{14dhh3} is dynamically consistent with the original 4D HH model given by  \eqref{eq:HH4D_V}-\eqref{eq:rate3}.

In \rev{\S\ref{subsec:converge}} we provide numerical evidence that the flow induced by $G$ on $\mathcal{Y}$ converges to $\mathcal{M}$ exponentially fast.    Thus, an initial probability distribution over the ion channel states that is not multinomial quickly approaches a multinomial distribution with dynamics induced by the 4D HH equations.  Similar results, restricted to the voltage-clamp setting, were established by Keener and Earnshaw  \citep{Keener1998MathPhy,Keener:2009:JMathBiol, EarnshawKeener2010SIADS_2}.

\subsection{Local Convergence Rate}\label{subsec:converge}
Keener and Earnshaw \citep{Keener1998MathPhy,Keener:2009:JMathBiol,EarnshawKeener2010SIADS_2} showed that for Markov chains with constant (even time varying) transition rates: (i) the multinomial probability distributions corresponding to mean-field models (such as the HH sodium or potassium models) form  invariant submanifolds within the space of probability distributions over the channel states, and (ii) arbitrary initial probability distributions converged exponentially quickly to the invariant manifold.   
For systems with prescribed time-varying transition rates, such as for an ion channel system under voltage clamp with a prescribed voltage $V(t)$  as a function of time, the distribution of channel states had an invariant submanifold again corresponding to the multinomial distributions, and the flow on that manifold induced by the evolution equations was consistent with the flow of the full system.  

In the preceding section we established the dynamical consistency of the 14D and 4D models with enough generality to cover both the voltage-clamp and current-clamp systems;  the latter is distinguished by NOT having a prescribed voltage trace, but rather having the voltage coevolve along with the (randomly fluctuating) ion channel states.  
\rev{Here,} we give numerical evidence for exponential convergence under current clamp  similar to that established under voltage clamp by Keener and Earnshaw.


Rather than providing a rigorous proof, we give numerical evidence for the standard deterministic HH model that $y\to\mathcal{M}$ under current clamp (spontaneous firing conditions) in the following sense: if $y(t)$ is a solution of $\dot{y}=G(y)$ with arbitrary initial data $y_0\in\mathcal{Y}$, then   $||y(t)-H(R(y(t)))||\to 0$ as $t\to\infty$, exponentially quickly. 
Moreover, the convergence rate is bounded by $\lambda=\text{max}(\lambda_\text{v},\lambda_\text{Na},\lambda_\text{K})$, where $\lambda_\text{ion}$ is the least negative nontrivial eigenvalue of the channel state transition matrix (over the voltage range $V_\text{min}\le v \le V_\text{max}$) for a given ion, and $-1/\lambda_\text{v}$ is the largest value taken by the membrane time constant (for $V_\text{min}\le v \le V_\text{max}$).  In practice, we find that the membrane time constant does not determine the slowest time scale for convergence to $\mathcal{M}$.  
In fact it appears that the second-least-negative eigenvalues (not the least-negative eigenvalues) of the ion channel matrices set the convergence rate.

Note that $y\in \mathcal{Y}$ can be written as
$y=[V;\mbM;\mbN]$. As shown in Appendix \ref{append_LemmaXinvari}, the Jacobian matrix $\frac{\partial H}{\partial x}$ consists of three block matrices: one for the voltage terms, $\frac{\partial V}{\partial v}$, one associated to the \Na~gates, given by $\frac{\partial \mbM}{\partial m}$ and $\frac{\partial \mbM}{\partial h}$, and one corresponding to the \K~gates, $\frac{\partial \mbN}{\partial n}$. 
Fixing a particular voltage $v$, let $\lambda_i,\ i\in\{0,1,2,\ldots,7\}$ be the eight eigenvalues of $A_\text{Na}$  and $v_i$ be the associated eigenvectors, i.e., $A_\text{Na}v_i=\lambda_i v_i$ for the rate matrix in equations \eqref{14dhh2}.
Similarly, let $\eta_i,\ w_i,\  i\in\{0,1,2,\ldots,4\}$ be the five eigenvalues and the associated eigenvectors of $A_\text{K}$, i.e., $A_\text{K}w_i=\eta_i w_i$, for the rate matrix in equations \eqref{14dhh3}. 
If we rank the eigenvalues of either matrix in descending order, the leading eigenvalue is always zero (because the sum of each column for $A_\text{Na}$ and $A_\text{K}$ is zero for every $V$) and the remainder are real and negative.  
Let $\lambda_1$ and $\eta_1$ denote the largest (least negative) nontrivial eigenvalues of $A_\text{Na}$ and $A_\text{K}$, respectively, and let $v_1\in\R^8$ and $w_1\in\R^5$ be the corresponding eigenvectors.


The eigenvectors of the full 14D Jacobian are not simply related to the eigenvectors of the component submatrices, because the first (voltage) row and column contain nonzero off-diagonal elements. 
However, the eigenvectors associated to the largest nonzero eigenvectors of $A_\text{Na}$ and $A_\text{K}$ 
(respectively $v_2$ and $w_2$) are parallel to $\partial M/\partial h$ and $\partial N/\partial n$, regardless of voltage.   In other words, the slowest decaying directions for each ion channel, $v_1$ and $w_1$, transport the flow along the multinominal sub-manifold of $\mathcal{Y}$. Therefore, it is reasonable to make the hypothesis that if $Y(t)$ is a solution of $\dot{y}=G(y)$ with arbitrary initial data $y\in\mathcal{Y}$, then
\begin{equation} \label{ineq:converge}
    \frac{||y(t)-H(R(y(t)))||}{||y(0)-H(R(y(0)))||}\lesssim e^{-\lambda_2 t} 
\end{equation}
for $\lambda_2$ being the \emph{second largest} nonzero eigenvalue of $A_\text{K}$ and $A_\text{Na}$ over all $v$ in the range $v_\text{min}<v<v_\text{max}$. 
The convergence behavior is plotted numerically in Fig.~\ref{fig:converge}, and is consistent with the Ansatz \eqref{ineq:converge}. We calculate the distance from a point $y$ to $\mathcal{M}$ as  
\begin{equation} \label{eq:initial_p}
    y_\text{max}=\operatorname*{argmax}_{y\in\mathcal{Y}} \left\Vert y-H(R(y))\right\Vert^2.
\end{equation}
In order to obtain an upper bound on the distance as a function of time, we begin with the furthest point in the simplex from $\mathcal{M}$, by numerically finding the solution to the argument \eqref{eq:initial_p}, which is 
$$y_\text{max}=[v, 0.5, 0, 0,0.5,0,0,0,0,0.5,0,0,0,0.5].$$ 
This vector represents the furthest possible departure from the multinomial distribution: all probability equally divided between the extreme states $m_{00}$ and $m_{03}$ for the sodium channel, and the extremal states $n_0$ and $n_4$ for potassium.  The maximum distance from the multinomial submanifold $\mathcal{M}$, $d_\text{max}$, is calculated using this point. As shown in Fig.~\ref{fig:converge},
the function $d_\text{max}\,e^{-\lambda_2 t}$ provides a tight upper bound for the convergence rate from arbitrary initial data $y\in\mathcal{Y}$ to the invariant submanifold $\mathcal{M}$.

\begin{figure}
\begin{center}
\includegraphics[scale=0.45]{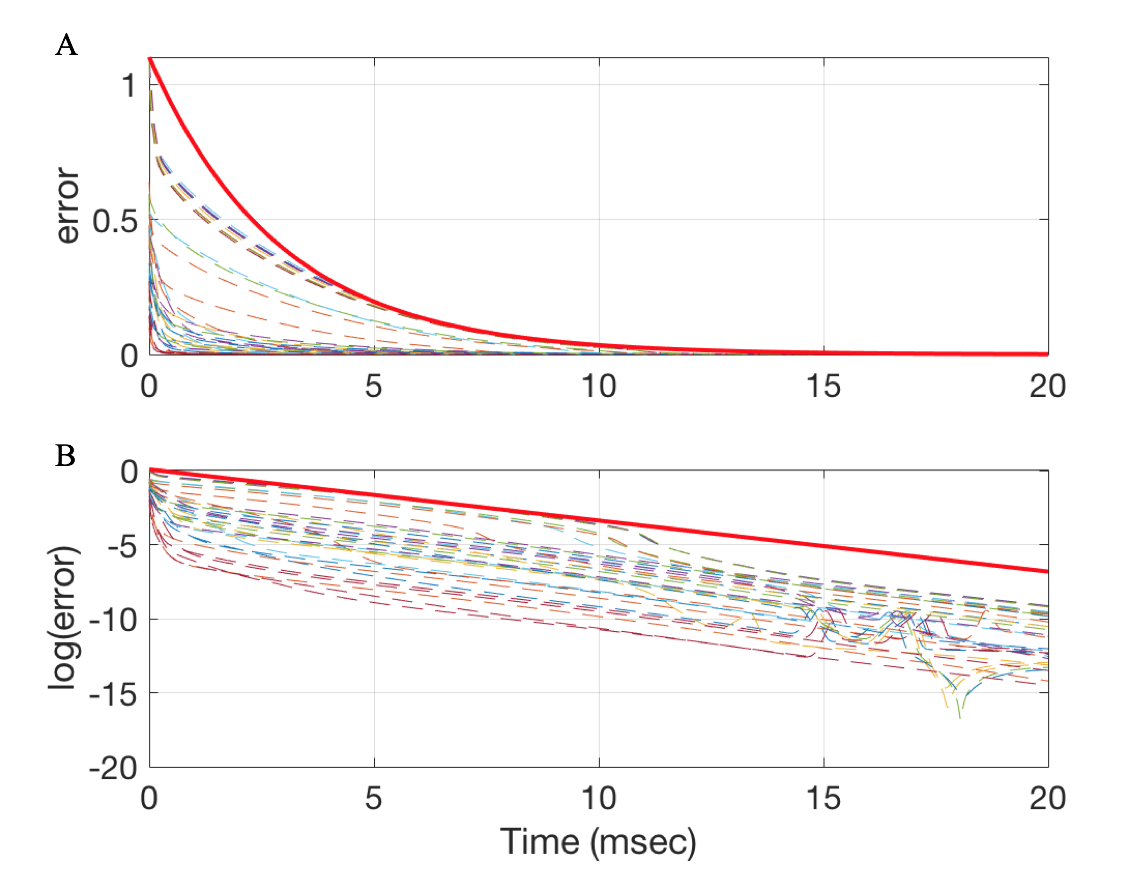}
\vspace*{-4mm}
\caption{\small{Convergence of trajectories $y(t)$, for arbitrary initial conditions $y_0\in\mathcal{Y}$, to the multinomial submanifold $\mathcal{M}$, for an ensemble of random initial conditions. A: distance \eqref{eq:initial_p} between $y(t)$ and  $\mathcal{M}$. B: Logarithm of the distance in panel A. The red solid line  shows $d_\text{max}e^{-\lambda_2 t}$ in panel A and \rev{$\log(d_\text{max})-\lambda_2 t$} in panel B.}}
\label{fig:converge}
\end{center}
\end{figure}



\section{Stochastic 14D Hodgkin-Huxley Models}\label{sec:stochastic_14DHH}

Finite populations of ion channels generate stochastic fluctuations (``channel noise") in ionic currents that influence action potential initiation and timing \citep{White1998APSB,SchneidmanFreedmanSegev1998NECO}.  
At the molecular level, fluctuations arise because transitions between individual ion channel states are stochastic \citep{HillChen1972BPJ,Neher1976Nature,SkaugenWalloe1979ActaPhysiolScand}.  
Each directed edge in the ion channel state transition diagrams (cf.~Fig.~\ref{plot:HHNaKgates}) introduces an independent noise source.  
It is of interest to be able to attribute variability of the interspike interval timing distribution to specific molecular noise sources, 
\rev{specifically} individual directed edges in each channel state graph.
In order to explore these contributions, we develop a system of  Langevin equations for the Hodgkin-Huxley equations, set in a 14-dimensional phase space.

Working with a higher-dimensional stochastic model may appear inconvenient, but in fact has several advantages.
First, any projection of an underlying 14D model onto a lower (e.g.~4D) stochastic model generally entails loss of the Markov property.
Second, the higher-dimensional representation allows us to assess the contribution of individual molecular transitions to the macroscopically observable variability of timing in the interspike interval distribution.  
Third, by using a rectangular noise coefficient matrix constructed directly from the transitions in the ion channel graphs, we avoid a matrix decomposition step.  This approach leads  to a fast algorithm that is equivalent to the slower algorithm due to Fox and Lu \citep{FoxLu1994PRE,GoldwynSheaBrown2011PLoSComputBiol} in a strong sense (pathwise equivalence) that we detail in \S \ref{sec:path_equiv}.

\subsection{Exact Stochastic Simulation of HH Kinetics: the Random--Time-Change Representation}\label{subsec:RTC}

An ``exact" representation of the Hodgkin-Huxley system with a population of $\mtot$ sodium channels and $\ntot$ potassium channels treats each of the 20 directed edges in the sodium channel diagram, and each of the 8 directed edges in the potassium channel diagram, as independent Poisson processes, with voltage-dependent per capita intensities.
As in the deterministic case, the sodium and potassium  channel population vectors $\mbM$ and $\mbN$ satisfy $\sum_{i=0}^3\sum_{j=0}^{1}\mbM_{ij}\equiv 1\equiv\sum_{i=0}^4\mbN_i$.\footnote{We annotate the stochastic population vector $\mbM$ either as $[M_{00},M_{10},\ldots,M_{31}]$ or as $[M_1,\ldots,M_8]$, whichever is more convenient. In either notation $M_{31}\equiv M_8$ is the conducting state of the \Na channel. For the \K channel, $N_4$ denotes the conducting state.} 
Thus they are constrained, respectively, to a 7D simplex embedded in $\R^8$ and a 4D simplex embedded in $\R^5$.
In the random--time-change representation \citep{anderson2015stochastic} the exact evolution equations are written in terms of sums over the directed edges $\mathcal{E}$ for each ion channel, $\mathcal{E}_\text{Na}=\{1,\ldots,20\}$ and $\mathcal{E}_\text{K}=\{1,\ldots,8\}$, cf.~Fig.~\ref{plot:HHNaKgates}.
\begin{align}
    \mbM(t)&=\mbM(0)+\frac1\mtot \sum_{k\in\mathcal{E}_\text{Na}} \zeta^\text{Na}_kY^\text{Na}_k\left(\mtot
    \int_0^t \alpha^\text{Na}_k(V(s))\mbM_{i(k)}(s)\,ds
    \right) \label{eq:RTCNa}\\
    \mbN(t)&=\mbN(0)+\frac1\ntot
    \sum_{k\in\mathcal{E}_\text{K}} \zeta^\text{K}_k Y^\text{K}_k\left(\ntot
    \int_0^t \alpha^\text{K}_k(V(s))\mbN_{i(k)}(s)\,ds
    \right) \label{eq:RTCK}.
\end{align}
Here $\zeta^\text{ion}_k$ is the stoichiometry vector for the $k$th directed edge. If we write $i(k)$ for the source node and $j(k)$ for the destination node of edge $k$, then $\zeta^\text{ion}_k=e^\text{ion}_{j(k)}-e^\text{ion}_{i(k)}$.\footnote{We write $e^\text{Na}_i$ and $e^\text{K}_i$ for the $i$th standard unit vector in $\R^8$ or $\R^5$, respectively.}  
Each $Y^\text{ion}_k(\tau)$ is an independent unit-rate Poisson process, evaluated at ``internal time'' (or integrated intensity) $\tau$, representing the independent channel noise arising from transitions along the $k$th edge.  The voltage-dependent per capita transition rate along the $k$th edge is $\alpha^\text{ion}_k(v)$, and $\mbM_{i(k)}(s)$ (resp.~$\mbN_{i(k)}(s)$) is the fractional occupancy of the source node for the $k$th transition at time $s$. Thus, for example, the quantity $\mtot\alpha_k^\text{Na}(V(s))\mbM_{i(k)}(s)$ gives the \emph{net} intensity along the $k$th directed edge in the \Na~channel graph at time $s$.

\begin{remark}
Under ``voltage-clamp'' conditions,  with the voltage $V$  held fixed,   \eqref{eq:RTCNa}-\eqref{eq:RTCK} reduce to a time-invariant first-order transition process on an directed graph \citep{SchmidtThomas2014JMN,bmb:Cadgil+Lee+Othmer:2005}.
\end{remark}

Under ``current-clamp" conditions, the voltage evolves according to a conditionally deterministic current balance equation of the form
\begin{align}\label{eq:RTR_dV}
\frac{dV}{dt}&=\frac1C\left\{I_\text{app}(t)-\bar{g}_\text{Na}\mbM_{31}\left(V-V_\text{Na}\right)-\bar{g}_\text{K}\mbN_{4}\left(V-V_\text{K}\right)-g_\text{leak}(V-V_\text{leak})    \right\}.
\end{align} 
Here, $C$ ($\mu F/cm^2$) is the capacitance, $I_\text{app}$ ($nA/cm^2$) is the applied current, the maximal conductance is $\bar{g}_\text{chan}$ ($mS/cm^2$),  $V_\text{chan}$ ($mV$) is the associated reversal potential,   and the ohmic leak current is $g_\text{leak}(V-V_\text{leak})$. 

The random--time-change representation \eqref{eq:RTCNa}-\eqref{eq:RTR_dV} leads to an exact stochastic simulation algorithm, given in \citep{anderson2015stochastic}; equivalent simulation algorithms have been used previously \citep{ClayDeFelice1983BiophysJ,NewbyBressloffKeener2013PRL}.  
Many authors substitute a simplified Gillespie algorithm that imposes a piecewise-constant propensity approximation, ignoring the voltage dependence of the transition rates $\alpha^\text{ion}_k$ between channel transition events  \citep{Goldwyn2011PRE,GoldwynSheaBrown2011PLoSComputBiol,OrioSoudry2012PLoS1, Pezo2014Frontiers}.  
The two methods give similar moment statistics, provided $\ntot,\mtot\gtrsim 40$ \citep{anderson2015stochastic};  their similarity regarding path-dependent properties (including interspike interval distributions) has not been studied in detail.  
Moreover, both Markov chain algorithms are prohibitively slow for modest numbers (e.g.~thousands) of channels; the exact algorithm may be even slower than the approximate Gillespie algorithm.    For consistency with previous studies, in this paper we use the piecewise-constant propensity Gillespie algorithm with $\mtot=6000$ \Na~and $\ntot=1800$ \K~channels  as our ``gold standard" Markov chain (MC) model, as in  \citep{GoldwynSheaBrown2011PLoSComputBiol}.

In \rev{\S\ref{subsec:14DHH}} we develop a 14D conductance-based Langevin model with 28 independent noise sources -- one for each directed edge -- derived from the random--time-change representation \eqref{eq:RTCNa}-\eqref{eq:RTR_dV}.   
In previous work \citep{SchmidtThomas2014JMN} we established a quantitative measure of ``edge importance'', namely the contribution of individual transitions (directed edges) to the variance of channel state occupancy under steady-state voltage-clamp conditions.
Under voltage clamp, the edge importance was identical for each reciprocal pair of directed edges in the graph, a consequence of detailed balance.   
Some Langevin models lump the noise contributions of each pair of edges  \citep{Dangerfield2010PCS, OrioSoudry2012PLoS1,Dangerfield2012APS,Pezo2014Frontiers}.  Under conditions of detailed balance, this simplification is well justified. However, as we will show  (cf.~Fig.~\ref{fig:logSSNaK}) under current-clamp conditions, e.g.~for an actively spiking neuron, detailed balance is violated, the reciprocal edge symmetry is broken, and each pair of directed edges makes a distinct contribution to ISI variability.

\subsection{Langevin Equations of the 14D HH Model}
\label{subsec:14DHH}
For sufficiently large number of channels,  \citep{SchmidtThomas2014JMN,SchmidtGalanThomas2018PLoSCB} showed that under voltage clamp, equations \eqref{eq:RTCNa}-\eqref{eq:RTCK} can be approximated by a multidimensional Ornstein-Uhlenbeck (OU) process (or Langevin equation) in the form\footnote{The convergence of the discrete channel system to a Langevin system under voltage clamp is a special case of Kurtz' theorem \citep{Kurtz1981}.}
\begin{align}\label{eq:RTR_dM}
d\mbM&=\sum_{k=1}^{20}\zeta_k^{\text{Na}}\left\{\alpha_k^\text{Na}(V)\mbM_{i(k)}dt +\sqrt{\epsilon^\text{Na}\alpha_k^\text{Na}(V)\mbM_{i(k)}}\,dW^\text{Na}_k\right\}\\
\label{eq:RTR_dN}
d\mbN&=\sum_{k=1}^{8}\zeta_k^{\text{K}}\left\{\alpha_k^\text{K}(V)\mbN_{i(k)}dt +\sqrt{\epsilon^\text{K}\alpha_k^\text{K}(V)\mbN_{i(k)}}\,dW^\text{K}_k\right\}.
\end{align}
Here, $\mbM$, $\mbN$, $\zeta_k^\text{ion}$, and $\alpha_k^\text{ion}$
have the same meaning as in \eqref{eq:RTCNa}-\eqref{eq:RTCK}.
The channel state \emph{increments} in a short time interval $dt$ are $d\mbM$ and $d\mbN$, respectively.
The finite-time increment in the Poisson process $Y^\text{ion}_k$ is now approximated by a Gaussian process, namely the increment $dW_k^\text{ion}$ in a Wiener (Brownian motion) process associated with each directed edge.
These independent noise terms are scaled by $\epsilon^\text{Na}=1/\mtot$ and $\epsilon^\text{K}=1/\ntot$, respectively.

Equations \eqref{eq:RTR_dV}-\eqref{eq:RTR_dN} comprise a system of Langevin equations for the HH system (under current clamp) on a 14-dimensional phase space driven by 28 independent white noise sources, one for each directed edge. These equations may be written succinctly in the form 
\begin{equation}\label{eq:langevin-main}
  d\mbX=\mbf(\mbX)\,dt+\sqrt{\epsilon}\mathcal{G}(\mbX)\,d\mbW(t)  
\end{equation}
where we define the 14-component vector $\mbX=(V;\mbM;\mbN)$, and $\mbW(t)$ is a Wiener process with 28 independent components.
The deterministic part of the evolution equation $\mbf(\mbX)\,=\left[\frac{dV}{dt};\frac{d\mbM}{dt};\frac{d\mbN}{dt}\right]$ is the same as the mean-field, equations \eqref{14dhh1}-\eqref{14dhh3}. The state-dependent noise coefficient matrix $\mathcal{G}$ is $14\times28$ and can be written as 
\[\sqrt{\epsilon}\mathcal{G} =\left(\begin{array}{@{}c|c@{}}
  \bigzero_{1\times 20} & 
  \bigzero_{1\times 8}  \\
\hline
  \begin{matrix}
  S_\text{Na}
  \end{matrix} &\bigzero_{8\times 8} \\
  \hline
  \bigzero_{5\times 20} & \begin{matrix}
  S_\text{K}
  \end{matrix}
\end{array}\right).
\]
The coefficient matrix $S_\text{K}$
is
\begin{align*}
S_\text{K}=
\frac{1}{\sqrt{\ntot}}
&\left[
\begin{array}{cccc}
-\sqrt{4\alpha_n n_0}& \sqrt{\beta_n n_1}&0&0\\
   \sqrt{4\alpha_n n_0}& -\sqrt{\beta_n n_1}&-\sqrt{3\alpha_n n_1}&\sqrt{2\beta_n n_2} \\
  0 &0 &\sqrt{3\alpha_n n_1}&-\sqrt{2\beta_n n_2}\\
    0 &0&0&0 \\
    0&0&0&0\end{array} 
\right.\cdots\\
&\quad\quad\quad\quad\quad\quad\cdots \left. 
\begin{array}{cccc}
0&0&0&0\\
0&0&0&0 \\
  -\sqrt{2\alpha_n n_2}&\sqrt{3\beta_n n_3}&0&0\\
    \sqrt{2\alpha_n n_2}&-\sqrt{3\beta_n n_3}&-\sqrt{\alpha_n n_3}&\sqrt{4\beta_n n_4} \\
    0&0&\sqrt{\alpha_n n_3}&-\sqrt{4\beta_n n_4} 
    \end{array}
    \right],
\end{align*}
and $S_\text{Na}$ is given in Appendix \ref{app:SNaSK}. 
Note that each of the 8 columns of $S_\text{K}$ corresponds to the flux vector along a single directed edge in the \K~channel transition graph.
Similarly, each of the 20 columns of $S_\text{Na}$ corresponds to the flux vector along a directed edge in the \Na~graph (cf.~App.~\S\ref{app:SNaSK}).

\begin{remark}
Although the ion channel state trajectories generated by equation \eqref{eq:langevin-main} are not strictly bounded to remain within the nonnegative simplex, empirically, the voltage nevertheless remains within specified limits with overwhelming probability.
\end{remark}
To facilitate comparison of the model \eqref{eq:RTR_dV}-\eqref{eq:RTR_dN}  with prior work  \citep{FoxLu1994PRE, Fox1997BiophysicalJournal,  GoldwynSheaBrown2011PLoSComputBiol}, we may rewrite the $14\times28$D Langevin description in the equivalent form 
\begin{align}
C \frac{dV}{dt}&=I_\text{app}(t)-\bar{g}_\text{Na}\mbM_8\left(V-V_\text{Na}\right)-\bar{g}_\text{K}\mbN_5\left(V-V_\text{K}\right)-g_\text{leak}(V-V_\text{leak}), \label{eq:14dHH_L_dV}\\
    \frac{d\mbM}{dt}&=A_\text{Na}\mbM+S_\text{Na}\xi_\text{Na}, \label{eq:14dHH_L_dNa} \\
    \frac{d\mbN}{dt}&=A_\text{K}\mbM+S_\text{K}\xi_\text{K}, \label{eq:14dHH_L_dK}
\end{align}
The drift matrices $A_\text{Na}$ and $A_\text{K}$ remain the same as in \citep{FoxLu1994PRE}, and are the same as in the
14D deterministic model \eqref{matrix:ANa}-\eqref{Matrix:AK}. $S_\text{Na}$ and $S_\text{K}$ are constructed from direct transitions of the underlying kinetics in Fig.~\ref{plot:HHNaKgates},  $\xi_\text{Na}\in\R^{20}$ and  $\xi_\text{K}\in\R^8$ are vectors of independent Gaussian white noise processes with zero mean and unit variance.

Fox and Lu's original approach \citep{FoxLu1994PRE} requires solving a matrix square root equation $SS^\intercal=D$ to obtain a square ($8\times 8$ for Na$^+$ or $5\times 5$ for K$^+$) noise coefficient matrix consistent with the state-dependent diffusion matrix $D$. As an advantage, the ion channel representation  \eqref{eq:14dHH_L_dV}-\eqref{eq:14dHH_L_dK} uses sparse, nonsquare noise coefficient matrices ($8\times 20$ for the Na$^+$ channel and $5\times 8$ for the K$^+$ channel), which exposes the independent sources of noise for the system. 

The new Langevin model in \eqref{eq:14dHH_L_dV}-\eqref{eq:14dHH_L_dK} does not require detailed balance, which gives more insights to the underlying kinetics.
Review papers such as \citep{GoldwynSheaBrown2011PLoSComputBiol, Pezo2014Frontiers, Huang2015PhBio}, did systematic comparison of various stochastic versions of the HH model.  
In \S \ref{sec:path_equiv} and \S \ref{sec:modelcomp}, we quantitaviely analyze the connection between the new model and other existing models \citep{FoxLu1994PRE, Goldwyn2011PRE, GoldwynSheaBrown2011PLoSComputBiol, Dangerfield2010PCS, OrioSoudry2012PLoS1, Dangerfield2012APS, Huang2013APS,Pezo2014Frontiers,Huang2015PhBio,Fox2018arXiv}. Problems such as the boundary constrains are beyond the scope of this paper, however, we would like to connect the new model to another type of approximation to the MC model, namely the stochastic shielding approximation.

\subsection{Stochastic Shielding for the 14D HH Model} \label{subsec:SS}
 The stochastic shielding (SS) approximation was  introduced by Schmandt and Gal$\acute{\text{a}}$n \citep{SchmandtGalan2012PRL}, in order to \rev{approximate} the Markov process using fluctuations from only a subset of the transitions, namely those corresponding to changes in the 
 observable states.
 In \citep{SchmidtThomas2014JMN}, we showed that, under voltage clamp, each directed edge makes a distinct contribution to the steady-state variance of the ion channel conductance, with the total variance being a sum of these contributions.  We call the variance due to the $k$th directed edge the \emph{edge importance}; assuming detailed balance, it is given by 
 \begin{equation}
 \label{eq:edgeimportance}
     R_k=J_k\sum_{i=2}^n\sum_{j=2}^n\left(\frac{-1}{\lambda_i+\lambda_j}\right)
     \left(\mbc^\intercal v_i \right)
     \left( w_i^\intercal \zeta_k \right)
     \left( \zeta_k^\intercal w_j  \right)
     \left(v_j^\intercal \mbc \right).
 \end{equation}
Here, $J_k$ is the steady-state probability flux along the $k$th directed edge; $\lambda_n<\lambda_{n-1}\le \ldots\le \lambda_2< 0$ are the eigenvalues of the drift matrix ($A_\text{Na}$ or $A_\text{K}$, respectively), and $v_i$ (resp.~$w_i$) are the corresponding right (resp.~left) eigenvectors of the drift matrix.  Each edge's stoichiometry vector $\zeta_k$ has components summing to zero; consequently the columns of $A_\text{Na}$ and $A_\text{K}$ all sum to zero.  Thus each drift matrix has a leading trivial eigenvalue $\lambda_1\equiv 0$.  The  vector $\mbc$ specifies the unitary conductance of each ion channel state; for the HH model it is proportional to $e_8^\text{Na}$ or $e_5^\text{K}$, respectively.

Fig.~\ref{fig:SS_volt} shows the edge importance for each pair of edges in the HH \Na~and \K~ion channel graph, as a function of voltage in the range
$[-100, 100]$ mV. 
Note that reciprocal edges have identical $R_k$ due to detailed balance. Under voltage clamp, the largest value of $R_k$ for the HH channels always corresponds to directly observable transitions, i.e.~edges $k$ such that $|\mbc^\intercal\zeta_k|>0$, although this condition need not hold in general \citep{SchmidtGalanThomas2018PLoSCB}.

\begin{figure}
\begin{center}
\includegraphics[scale=0.65]{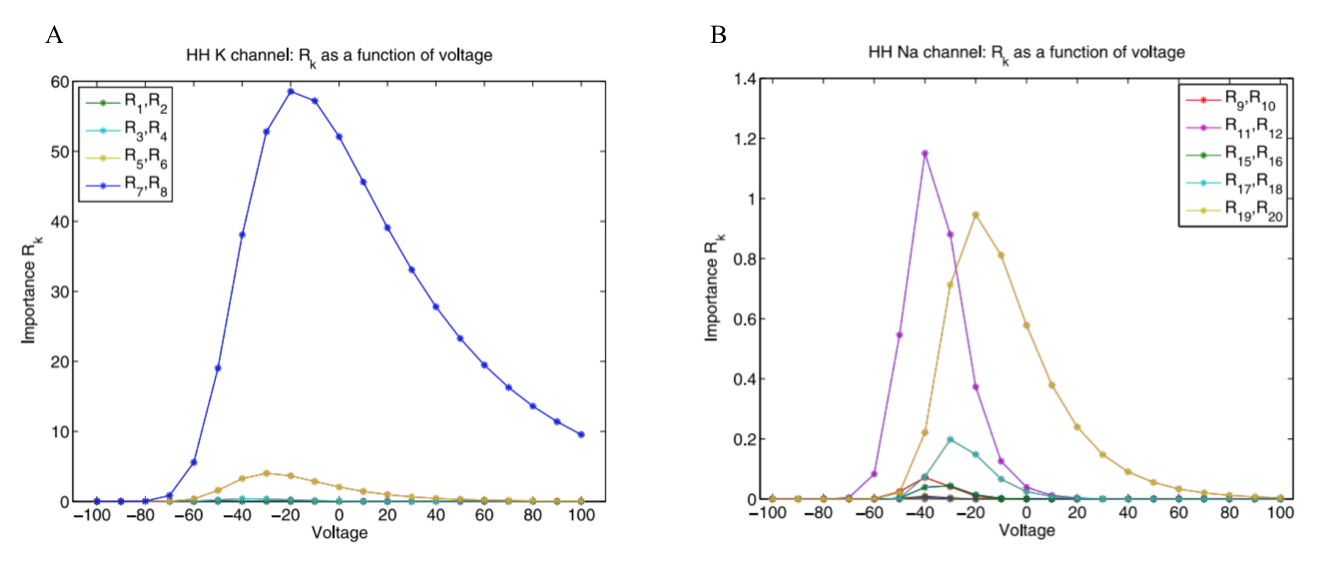}
\caption{Stochastic shielding under voltage clamp. Redrawn (with permission) from Figs.~10 \& 13 of \citep{SchmidtThomas2014JMN}. Each curve shows the
edge importance $\text{R}_k$ (equation \eqref{eq:edgeimportance}) as a
function of voltage in the range
$[-100, 100]$ mV for a different
 edge pair.  For the \K~kinetics, $\text{R}_7=\text{R}_8$ are the largest $\text{R}_k$ value in the voltage range above. For the \Na~kinetics, $\text{R}_{11}=\text{R}_{12}$ have the largest $\text{R}_k$ values in the subthreshold voltage range (c.~$[-100, -25]$ mV), and $\text{R}_{19}=\text{R}_{20}$ have the largest $\text{R}_k$ values in the suprathreshold voltage range (c.~$[-25, 100]$ mV). }\label{fig:SS_volt}
\end{center}
\end{figure}

To apply the stochastic shielding method under current clamp, we simulate the model with noise from only a selected subset  $\mathcal{E'}\subset\mathcal{E}$ of directed edges, replacing \eqref{eq:14dHH_L_dNa}-\eqref{eq:14dHH_L_dK} with
\begin{align}
    \frac{d\mbM}{dt}&=A_\text{Na}\mbM+S'_\text{Na}\xi_\text{Na}, \label{eq:14dHH_L_dNa_SS} \\
    \frac{d\mbN}{dt}&=A_\text{K}\mbM+S'_\text{K}\xi_\text{K}, \label{eq:14dHH_L_dK_SS},
\end{align}
where $S'_\text{Na}$ (resp.~$S'_\text{K}$) is a reduced matrix containing only the noise coefficients from the most important edges $\mathcal{E'}$.  \rev{That is, $\mathcal{E'}$ contains a subset of edges with the largest edge-importance values $R_k$.}

\cite{SchmandtGalan2012PRL} assumed that the edges with the largest contribution contribution to current fluctuations under voltage clamp would also make the largest contributions to variability in voltage and timing under current clamp, and included edges $7-8$ of the \K channel ($\mathcal{E}'_\text{K}=\{7,8\}$) and edges $11-12$ and $19-20$ of the \Na~channel ($\mathcal{E}'_\text{Na}=\{11,12,19,20\}$), yielding an
$8\times 4$ matrix $S'_\text{Na}$ and an $\rev{5}\times 2$ matrix $S'_\text{K}$.
They demonstrated numerically that restricting stochastic forcing to these edges gave a significantly faster simulation with little appreciable change in statistical behavior: 
under voltage clamp, the mean current remained the same, with a small (but noticeable) decrease in the current variance; meanwhile similar inter-spike interval (ISI) statistics were observed.

Under current clamp, detailed balance is violated, and it is not clear from mathematical principles whether the edges with the largest $R_k$ under voltage clamp necessarily make the largest contribution under other circumstances. 
In order to evaluate the contribution of the fluctuations driven by each directed edge on ISI variability, we test the stochastic shielding method by removing all but one column of $S_\text{ion}$ at a time.
That is, we restrict to a single noise source and observe the resulting ISI variance empirically.
For example, to calculate the importance of the $k^{th}$ direct edge in the \Na~channel, we suppress the noise from all other edges by setting $S'_\text{K}\xi_\text{K}=\bigzero_{5\times 1}$ and   
$$S'_\text{Na}=\left[\bigzero_{8\times 1},\cdots,S_\text{Na}(:,k),\cdots,\bigzero_{8\times 1}\right]$$
i.e., only include the k$^{th}$ column of $S_\text{Na}$ and set other columns to be zeros. The ISI variance was calculated from an ensemble of $10^4$  voltage traces, each spanning c.~500 ISIs.

Fig.~\ref{fig:logSSNaK}A plots the logarithm of the ISI variance for each edge in $\mathcal{E}_\text{K}$. Vertical bars (cyan) show the ensemble mean of the ISI variance, with a 95\% confidence interval \rev{superimposed} (magenta).  
Several observations are in order.  
First, the ISI variance driven by the noise in each edge decreases rapidly, the further the edge is from the observable transitions (edges 7,8), reflecting the underlying ``stochastic shielding" phenomenon.  Second, the symmetry of the edge importance for reciprocal edge pairs ((1,2), (3,4), (5,6) and (7,8)) that is observed under voltage clamp is broken under current clamp.  The contribution of individual directed edges to timing variability under current clamp has another important difference compared with the edge importance (current fluctuations) under voltage clamp.
A similar breaking of symmetry for reciprocal edges is seen for the \Na~channel, again reflecting the lack of detailed balance during active spiking. 

Fig.~\ref{fig:logSSNaK}B shows the ISI variance when channel noise is included on individual edges of $\mathcal{E}_\text{Na}.$
Here the difference between voltage and current clamp is striking.  Under voltage clamp, the four most important edges are always those representing ``observable transitions", in the sense that the transition's stoichiometry vector $\zeta$ is not orthogonal to the conductance vector $\mbc$.  That is,  the four most important  pairs are always 11-12 and 19-20, regardless of voltage (Fig.~\ref{fig:SS_volt}).  
Under current clamp,  the most important edges are 17, 18, 19 and 20. Although edges 11 and 12 are among the four most important sources of channel population fluctuations under voltage clamp, they are not even among the top ten contributors to ISI variance, when taken singly. 
Even though edges  17 \rev{and} 18 are ``hidden", meaning they do not directly change the instantaneous channel conductance, these edges are nevertheless the second most important pair under current clamp.
Therefore, when we implement the stochastic-shielding based approximation, we include the pairs 17-18 and 19-20 in equation \eqref{eq:14dHH_L_dNa_SS}.
\rev{We refer to the approximate SS model driven by these six most important edges as the $14\times 6$D HH model.}

\rev{Given the other parameters we use for the HH model (cf.~Tab.~\ref{tab:parameters} in Appendix \ref{axppend_alpha_beta}), the input current of $I_\text{app}=10$ nA is slightly beyond the region of multistability associated with a subcritical Andronov-Hopf bifurcation.  
In order to make sure the results are robust against increases in the applied current, we tried current injections ranging from 20 to 100 nA. While injecting larger currents decreased the ISI variance, it did not change the rank order of the contributions from the most important edges.}


\begin{figure}
\begin{center}
\includegraphics[scale=0.56]{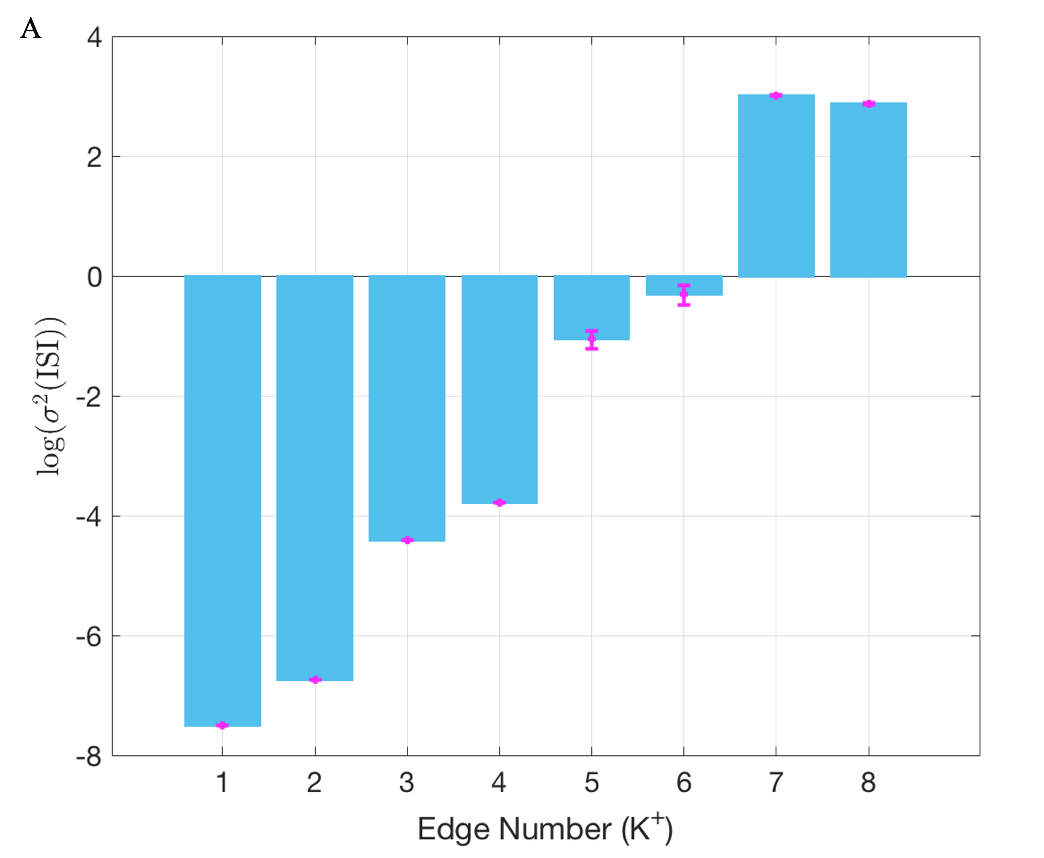}
\includegraphics[scale=0.55]{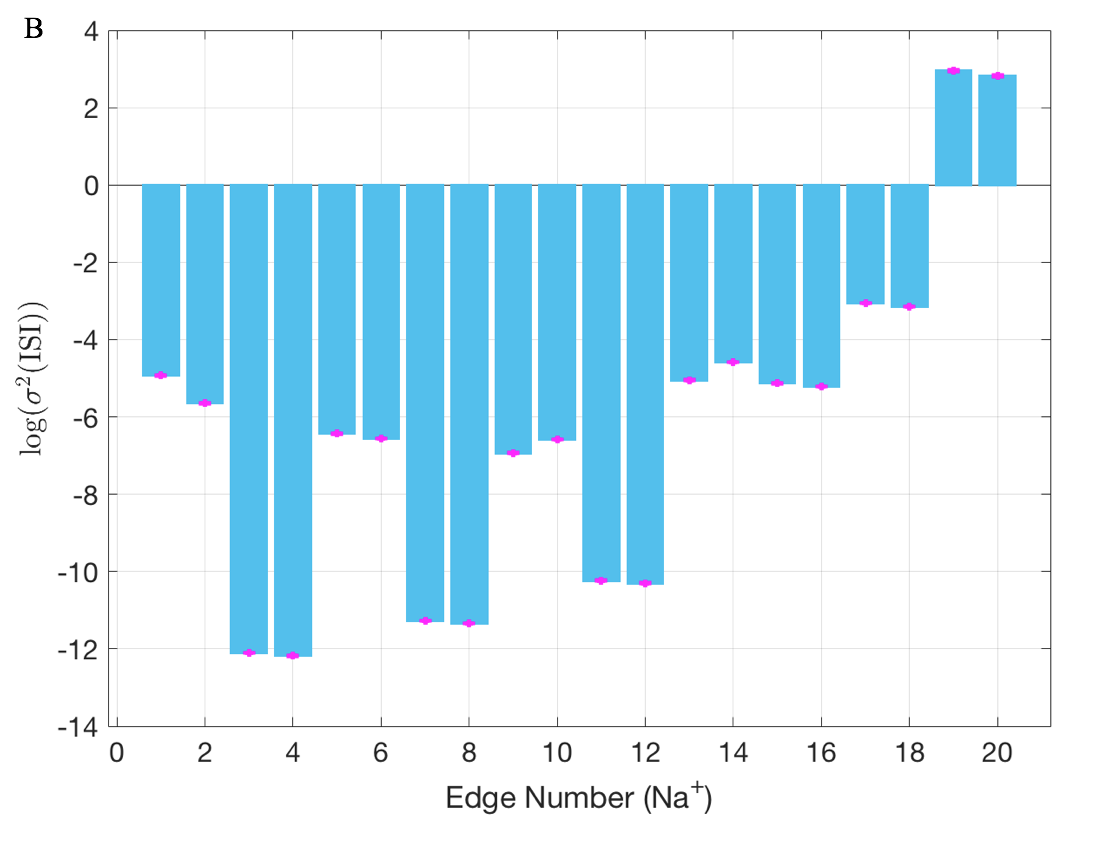}
\caption{Logarithm of variance of ISI for stochastic shielding under current clamp. Cyan bar is the mean of ISI, and magenta plots 95\% confidence interval of the mean ISI (see text for details). The applied current is 10 nA with other parameters specified in the Appendices. For the \K~ kinetics, the largest contribution edge is 7, and 8 is slightly smaller ranking the second largest. 
For the \Na~ kinetics, the largest contribution pair is 19 and 20, with 20 slightly smaller than 19. Moreover, edge 17 and 18 is the second largest pair.}\label{fig:logSSNaK}
\end{center}
\end{figure}

\section{Pathwise Equivalence for a Class of Langevin Models}
\label{sec:path_equiv}
Fox and Lu's method was widely used since its appearance (see references in \cite{Bruce2009ABE, GoldwynSheaBrown2011PLoSComputBiol, Huang2015PhBio}), and the ``best" approximation for the underlying Markov Chain (MC) model has been a subject of ongoing discussion for decades.  
Several studies \citep{MinoRubinsteinWhite2002AnnBiomedEng, Bruce2009ABE, Sengupta2010PRE} attested to discrepancies between Fox's later approach in  \citep{Fox1997BiophysicalJournal} and the discrete-state MC model, raising the question of whether Langevin approximations could  ever accurately represent the underlying fluctuations seen in the ``gold standard" MC models. 
An influential review paper  \citep{GoldwynSheaBrown2011PLoSComputBiol} found that these
discrepancies were due to the way in which noise is added to the stochastic differential equations \eqref{eq:FoxandLu_dv}-\eqref{eq:FoxandLu_dK}. 
Recent studies including \cite{Dangerfield2010PCS, Linaro2011PublicLibraryScience, GoldwynSheaBrown2011PLoSComputBiol, Goldwyn2011PRE, Dangerfield2012APS, OrioSoudry2012PLoS1, Guler2013MITPress, Huang2013APS, Pezo2014Frontiers, Huang2015PhBio, Fox2018arXiv} discussed various ways of incorporating channel noise into HH kinetics based on the original work by Fox and Lu \citep{FoxLu1994PRE, Fox1997BiophysicalJournal}, some of which have the same SDEs but with different boundary conditions. 
Different boundary conditions (BCs) are not expected to have much impact on computational efficiency. 
Indeed, if BCs are neglected, the main difference between channel-based (or conductance-based) models \rev{is} the diffusion matrix $S$ in the Langevin euqations \eqref{eq:FoxandLu_dNa} and \eqref{eq:FoxandLu_dK}. 
As the discussion about where and how to incorporate noise into the HH model framework goes on, \cite{Fox2018arXiv} recently asked whether there is a way of relating different models with different $S$ matrices.  We give a positive answer to this question below.

In \rev{\S\ref{subsec:equivalence}} we will demonstrate the equivalence (neglecting the boundary conditions) of a broad class of previously proposed channel-based Langevin models including: \citep{FoxLu1994PRE,  Dangerfield2010PCS, GoldwynSheaBrown2011PLoSComputBiol, Dangerfield2012APS, OrioSoudry2012PLoS1,  Huang2013APS, Pezo2014Frontiers, Fox2018arXiv} and the 14D Langevin HH model with 28 independent noise sources (one for each directed edge in the channel state transition graph), i.e.~our ``$14\times28$D'' Langevin model.

\subsection{When are Two Langevin Equations Equivalent?}
\label{subsec:equivalence}

Two Langevin models are pathwise equivalent if the sample paths (trajectories) of one model can be made to be identical to the sample paths of the other, under an appropriate choice of Gaussian white noise samples for each.  To make this notion precise, consider two channel-based Langevin models of the form $\mbdX=\mbf(\mbX)\,dt+G(\mbX)\,\mbdW$ with the same mean dynamics $\mbf\in\R^d$ and two different $d\times n$ matrices (possibly with different values of $n_1$ and $n_2$), $G_1$ and $G_2$. 
Denote 
\begin{eqnarray}
&& \mbf:\R^d \rightarrow \R^d, \\
&& G_1:\R^d \rightarrow \R^{d\times n_1},\\
&& G_2:\R^d \rightarrow \R^{d\times n_2}.
\end{eqnarray}
Let $\mbX(t)=[X_1(t),X_2(t),\ldots,X_d(t)]^\intercal$ and  $\mbX^*(t)=[X^*_1(t),X^*_2(t),\ldots,X^*_d(t)]^\intercal$ be  trajectories produced by the two models and let 
$\mbW(t)=[W_1(t),W_2(t),\ldots,W_{n_1}(t)]^\intercal$ and $\mbW^*(t)=[W^*_1(t),W^*_2(t),\ldots,W^*_{n_2}(t)]^\intercal$ be vectors of Weiner processes.  
That is, 
$W_i(t),\ i=1,2,\ldots,n_1$ and $W^*_j(t),\ j=1,2,\ldots,n_2$ are independent Wiener processes with $\langle W_i(s)W_j(t)\rangle=\delta_{ij}\delta(t-s)$ and $\langle W_i^*(s)W_j^*(t)\rangle=\delta_{ij}\delta(t-s)$. Note that  $n_1$ and $n_2$ need not be equal. 
As defined in \citep{AllenAllenArciniegoGreenwood2008StochAnalApp}, the stochastic differential equation (SDE) models 
\begin{equation}
  \mbdX=\mbf(t,\mbX(t))dt+G_1(t,\mbX(t))\mbdW(t) \label{SDE1}
\end{equation}
and
\begin{equation}
 \mbdX^*=\mbf(t,\mbX^*(t))dt+G_2(t,\mbX^*(t))\mbdW^*(t) \label{SDE2}
\end{equation}
are \emph{pathwise equivalent} if systems \eqref{SDE1} and \eqref{SDE2} posses the same probability distribution, and moreover, a sample path solution of one equation is also a sample solution to the other one.
\cite{AllenAllenArciniegoGreenwood2008StochAnalApp} proved a theorem giving general conditions under which the trajectories of two SDEs are equivalent. We follow their construction closely below, adapting it to the case of two different Langevin equations for the Hodgkin-Huxley system represented in a 14-dimensional state space.

As in \S\ref{sec:stochastic_14DHH}, channel-based Langevin models for the stochastic dynamics of HH can be written as 
\begin{equation}\label{eq:langevin_main}
  \mbdX=\mbf(\mbX)\,dt+\mathcal{S}(\mbX)\,\mbdW(t)  
\end{equation}
where the 14-component random vector $\mbX=(V;\mbM;\mbN)$ and $\mbf(\mbx)\,=\left[\frac{dV}{dt};\frac{d\mbM}{dt};\frac{d\mbN}{dt}\right]$ is the same as the mean-field, eqns.~\eqref{14dhh1}-\eqref{14dhh3}. 
Recall that $\mbx=[v,\mbm,\mbn]^\intercal$. Here we write
\[
\mathcal{S}(\mbx) =\left(\begin{array}{@{}c|c@{}}
  \bigzero_{1\times m} & 
  \bigzero_{1\times n}  \\
\hline
  \begin{matrix}
  S_\text{Na}(\mbm)
  \end{matrix} &\bigzero_{8\times n} \\
  \hline
  \bigzero_{5\times m} & \begin{matrix}
  S_\text{K}(\mbn)
  \end{matrix}
\end{array}\right),\quad\text{with}
\]
\begin{eqnarray}
&& S_\text{Na}:\R^8 \rightarrow \R^{8\times m},
\end{eqnarray}
for the \Na~channel, and 
\begin{eqnarray}
&& S_\text{K}:\R^5 \rightarrow \R^{5\times n},
\end{eqnarray}
for the \K~channel. 
Here, $m$ is the number of independent white noise forcing terms affecting the sodium channel variables, while $n$ is the number of independent noise sources affecting the potassium gating variables.  
We write $$\mbW(t) =[W_1(t),W_2(t), \ldots,W_{m+n}(t)]^\intercal$$ for a Wiener process incorporating both the sodium and potassium noise forcing.
Given two channel-based models with diffusion matrices
\begin{eqnarray}
&& S_\text{Na,1}:\R^8 \rightarrow \R^{8\times m_1},\\
&& S_\text{Na,2}:\R^8 \rightarrow \R^{8\times m_2},
\end{eqnarray}
for the \Na~channel, and 
\begin{eqnarray}
&& S_\text{K,1}:\R^5 \rightarrow \R^{5\times n_1},\\
&& S_\text{K,2}:\R^5 \rightarrow \R^{5\times n_2},
\end{eqnarray}
for the \K~channel, we construct  the diffusion matrix $\mathcal{D}=\mathcal{S}\mathcal{ S}^\intercal$. 
In order for the two models to generate equivalent sample paths, it suffices that they have the same diffusion matrix, i.e.
\[
\mathcal{D}=\mathcal{S}_1\mathcal{S}_1^\intercal =\left(\begin{array}{@{}c|c |c@{}}
  \bigzero_{1\times 1} & 
  \bigzero_{1\times 8} & \bigzero_{1\times 5}   \\
\hline
\bigzero_{8\times 1}&
  \begin{matrix}
  D_\text{Na}
  \end{matrix} &\bigzero_{8\times 5} \\
  \hline
  \bigzero_{5\times 1}&\bigzero_{5\times 8} & \begin{matrix}
  D_\text{K}
  \end{matrix}
\end{array}\right)=\mathcal{S}_2\mathcal{S}_2^\intercal.
\]
The SDEs corresponding to the two channel-based Langevin models are
\begin{eqnarray}
\mbdX&=&\mbf(t,\mbX(t))dt+\mathcal{S}_1(t,\mbX(t))\mbdW(t),\label{SDE3}\\
\mbdX^*&=&\mbf(t,\mbX^*(t))dt+\mathcal{S}_2(t,\mbX^*(t))\mbdW^*(t).\label{SDE4}
\end{eqnarray}

The probability density function $p(t,\mbx)$ for random variable $\mbX$ in  eqn.~\eqref{SDE3} satisfies the Fokker-Planck equation
\begin{eqnarray}
\frac{\partial p(t,\mbx)}{\partial t}&=&\frac{1}{2}\sum_{i=1}^8\sum_{j=1}^{8}\frac{\partial^2}{\partial x_ix_j}\Big[p(t,\mbx)\sum_{l=1}^{m_1+n_1}\mathcal{S}_1^{(i,l)}(t,\mbx)\mathcal{S}_1^{(j,l)}(t,\mbx)\Big] \nonumber  \\
&&-\sum_{i=1}^8\frac{\partial}{\partial x_i}\Big[\mbf_i(t,\mbx)p(t,\mbx)\Big]  \nonumber \\
&=&\frac{1}{2}\sum_{i=1}^8\sum_{j=1}^{8}\frac{\partial^2}{\partial x_ix_j}\Big[ \mathcal{D}^{(i,j)}(t,\mbx)p(t,\mbx)\Big]-\sum_{i=1}^8\frac{\partial}{\partial x_i}\Big[\mbf_i(t,\mbx)p(t,\mbx)\Big] \label{FPeqn1} 
\end{eqnarray}
where $\mathcal{S}_1^{(i,j)}(t,\mbx)$ is the $(i,j)^{th}$ entry of the diffusion matrix $\mathcal{S}_1(t,\mbx)$. Eqn.~\eqref{FPeqn1} holds because $$\mathcal{D}^{(i,j)}(t,\mbx)=\sum_{l=1}^{m_1+n_1}\mathcal{S}_1^{(i,l)}(t,\mbx)\mathcal{S}_1^{(j,l)}(t,\mbx).$$
If $\mathbf{z}_1$, $\mathbf{z}_2 \ \in \mathbb{R}^{14}$ and $\mathbf{z}_1 \leqslant \mathbf{z}_2$, then
$$P(\mathbf{z}_1 \leqslant \mbX(t)  \leqslant \mathbf{z}_2)=\int_{z_{1,14}}^{z_{2,14}}\int_{z_{1,13}}^{z_{2,13}}\cdots\int_{z_{1,1}}^{z_{2,1}}p(t,\mbx)dx_1dx_2\cdots dx_8.$$
 Note that \eqref{SDE3} and \eqref{SDE4} have the same expression \eqref{FPeqn1} for the Fokker-Planck equation, therefore, $\mbX$ and $\mbX^*$ possess the same probability density function. In other words, the probability density function of $\mbX$ in eqn.~\eqref{eq:langevin_main} is invariant for different choices of the diffusion matrix $\mathcal{S}$.

\subsection{Map Channel-based Langevin Models to Fox and Lu's Model}
\rev{We now} explicitly construct a  mapping between Fox and Lu's 14D model \citep{FoxLu1994PRE} and any channel-based model (given the same boundary conditions).
We begin with a channel-based Langevin description
\begin{equation}
\mbdX=\mbf(t,\mbX(t))dt+\mathcal{S}(t,\mbX(t))\mbdW(t), \label{eq:LA}
\end{equation}
and Fox and Lu's model
 \rev{\citep{FoxLu1994PRE}}
\begin{equation}
 \mbdX^*=\mbf(t,\mbX^*(t))dt+\mathcal{S}_0(t,\mbX^*(t))\mbdW^*(t) \label{eq:Fox},
\end{equation}
where $S$ is a $d$ by $m$ matrix satisfying $\mathcal{S}\mathcal{S}^\intercal=\mathcal{D}$ (note that $\mathcal{S}$ is not necessarily a square matrix), and $\mathcal{S}_0=\sqrt{\mathcal{D}}$.

 Let $T$ be the total simulation time of the random process in equations \eqref{eq:LA} and \eqref{eq:Fox}. 
 For $0\leqslant t\leqslant T$, denote the singular value decomposition (SVD) of $S$ as
$$\mathcal{S}(t)=P(t)\Lambda(t)Q(t)$$
 where $P(t)$ is an $d \times d$ orthogonal matrix (i.e., $P^\intercal P=PP^\intercal=I_d$)  and $Q(t)$ is an $m\times m$ orthogonal matrix, and $\Lambda(t)$ is a $d\times m$ matrix with $\text{rank}(\Lambda)= r\leqslant d$ positive diagonal entries and $d-r$ zero diagonal entries.
 
 First, we prove that given a Wiener trajectory, $\mbW(t),\ t\in[0,T]$ and the solution to eqn.~\eqref{eq:LA},  $\mbX(t)$, there exists a Wiener trajectory $\mbW^*(t)$ such that the solution to eqn.~\eqref{eq:Fox}, $\mbX^*$, is also a solution to  eqn.~\eqref{eq:LA}. 
 In other words, for a Wiener process $\mbW(t)$ we can construct a $\mbW^*(t)$, such that $\mbX^*(t)=\mbX(t)$, for $0\le t\le T$.
 
 Following \citep{AllenAllenArciniegoGreenwood2008StochAnalApp}, we construct the vector $\mbW^*(t)$ of $d$ independent Wiener processes as follows:
 \begin{equation}\label{eq:14D_Fox}
     \mbW^*(t)=\int_0^tP(s)\Big[\big(\Lambda(s)\Lambda^\intercal(s)\big)^{\frac{1}{2}}\Big]^+\Lambda(s)Q(s)\mbdW(s)+\int_0^tP(s)\mbdW^{**}(s)
 \end{equation}
for $0\leqslant t \leqslant  T$, where $\mbW^{**}(t)$ is a vector of length $d$ with the first $r$ entries equal to 0 and the next $d-r$ entries independent Wiener processes, and $\Big[\big(\Lambda(s)\Lambda^\intercal(s)\big)^{\frac{1}{2}}\Big]^+$ is the pseudoinverse of $\big(\Lambda(s)\Lambda^\intercal(s)\big)^{\frac{1}{2}}$. Consider that
\begin{eqnarray}
\mathcal{D}(t)&=&\mathcal{S}(t)\mathcal{S}^\intercal(t)=P(t)\Lambda(t)Q(t)\Big[P(t)\Lambda(t)Q(t)\Big]^\intercal \\
 & = &  P(t)\Lambda(t)\Lambda^\intercal(t)P^\intercal(t) \\
 &=& [\mathcal{S}_0(t)]^2,
\end{eqnarray}
where $\mathcal{S}_0(t)=P(t)\Big(\Lambda(t)\Lambda^\intercal(t)\Big)^{\frac{1}{2}}P^\intercal(t)$ is a square root of $\mathcal{D}$, by construction.

The diffusion term on the right side of \eqref{eq:Fox} with $\mbX^*(t)$ replaced by $\mbX(t)$ satisfies
\begin{align}
\mathcal{S}_0(t,\mbX(t))&\mbdW^*(t) \nonumber \\
=&\mathcal{S}_0(t)\Big(P(t)\Big[\big(\Lambda(t)\Lambda^\intercal(t)\big)^{\frac{1}{2}}\Big]^+\Lambda(t)Q(t)\mbdW(t)+P(t)\mbdW^{**}(t)\Big) \nonumber \\
  = & P(t)\Big(\Lambda(t)\Lambda^\intercal(t)\Big)^{\frac{1}{2}}P^\intercal(t)P(t)\Big[\big(\Lambda(t)\Lambda^\intercal(t)\big)^{\frac{1}{2}}\Big]^+\Lambda(t)Q(t)\mbdW(t) \nonumber\\
& +P(t)\Big(\Lambda(t)\Lambda^\intercal(t)\Big)^{\frac{1}{2}}P^\intercal(t)P(t)\mbdW^{**}(t) \nonumber\\
 =& \left\{P(t)\Lambda(t)Q(t)\right\}\mbdW(t).
 \label{eq:S0dW}
 \end{align}
 From the SVD of $\mathcal{S}$=P$\Lambda$Q, we conclude that
 \begin{equation}
\mathcal{S}_0(t,\mbX(t))\mbdW^*(t) =\mathcal{S}(t,\mbX(t))\mbdW(t).
\end{equation}
Hence, $\mbdX=\mbf(t,\mbX(t))dt+\mathcal{S}_0(t,\mbX(t))\mbdW^*_t$, i.e., $\mbX(t)$ is a sample path solution of equation \eqref{eq:Fox}. 

Similarly, given a Wiener trajectory $\mbW^*(t)$ and the solution to eqn.~\eqref{eq:Fox} $\mbX^*(t)$, we can construct a vector $\mbW(t)$ of $m$ independent Winner processes as 
\begin{equation}\label{eq:Fox_14D}
    \mbW(t)=\int_0^t Q^\intercal(s)\Lambda^+(s)\big[\Lambda(s)\Lambda^\intercal(s)\big]^{1/2}P^\intercal(s)\mbdW^*(s)+\int_0^t Q^\intercal(s)\mbdW^{***}(s)
\end{equation}
for $0\leqslant t \leqslant  T$, where $\mbW^{***}(t)$ is a vector of length $m$ with the first $r$ entries equal to 0 and the next $m-r$ entries independent Wiener processes, and $\Lambda^+(s)$ is the pseudoinverse of $\Lambda(s)$.
Then, by an argument parallel  to \eqref{eq:S0dW}, we conclude that 
\begin{equation}
  \mathcal{S}(t,\mbX^*(t))dW(t) =\mathcal{S}_0(t,\mbX^*(t))\mbdW^*(t).  
\end{equation}
Hence, $\mbdX^*=\mbf(t,\mbX^*(t))dt+\mathcal{S}(t,\mbX^*(t))\mbdW(t)$, that is, $\mbX^*(t)$ is also a solution to  \eqref{eq:LA}. Therefore we can conclude that the channel-based Langevin model in eqn.~\eqref{eq:LA} is pathwise equivalent to the Fox and Lu's model.

To illustrate pathwise equivalence, Fig.~\ref{fig:pwc} plots trajectories of the $14\times28$D stochastic HH model and Fox and Lu's model, using noise traces dictated by the preceding construction. In panel A, we generated a sample path for eqn.~\eqref{eq:LA} and plot three variables in $\mbX$: the voltage $V$, \Na~channel open probability $M_{31}$ and \K~channel open probability $N_4$.
The corresponding trajectory, $\mbX^*$, for  Fox and Lu's model was generated from eqn.~\eqref{eq:Fox} and the corresponding Wiener trajectory was calculated using eqn.~\eqref{eq:14D_Fox}. 
The top three subplots in panel A superposed the voltage $V^*$, \Na~channel open probability $M_{31}^*$ and \K~channel open probability $N_4^*$ in $\mbX^*$ against those in $\mbX$.
The bottom three subplots in panel A plot the point-wise differences of each variable.
Eqns.~\eqref{eq:LA} and \eqref{eq:Fox} are numerically solved in Matlab using the Euler-Maruyama method with a time step $dt=0.001$ms.
The slight differences observed arise in part due to numerical errors in calculating the singular value decomposition of $\mathcal{S}$ (in eqn.~\eqref{eq:LA}); another source of error is the finite accuracy of the Euler-Maruyama method.\footnote{The forward Euler method is first order accurate for ordinary differential equations, but the forward Euler-Maruyama method is only $O(\sqrt{dt})$ accurate for stochastic differential equations \citep{Kloeden-Platen-big}.}  
As shown in Fig.~\ref{fig:pwc}, most differences occur near the spiking region, where the system is numerically very stiff and the numerical accuracy of the SDE solver accounts for most of the discrepancies (analysis of which is beyond the scope of this paper). We can conclude from the comparison in Fig.~\ref{fig:pwc} that the $14\times28$D Langevin model is pathwise equivalent with the Fox and Lu's model. Similarly, the same analogy applies for other channel-based Langevin models such that with the same diffusion matrix $\mathcal{D}(\mbX)$. 

\begin{figure}
\begin{center}
\hspace*{-0.8cm}
\includegraphics[scale=0.77]{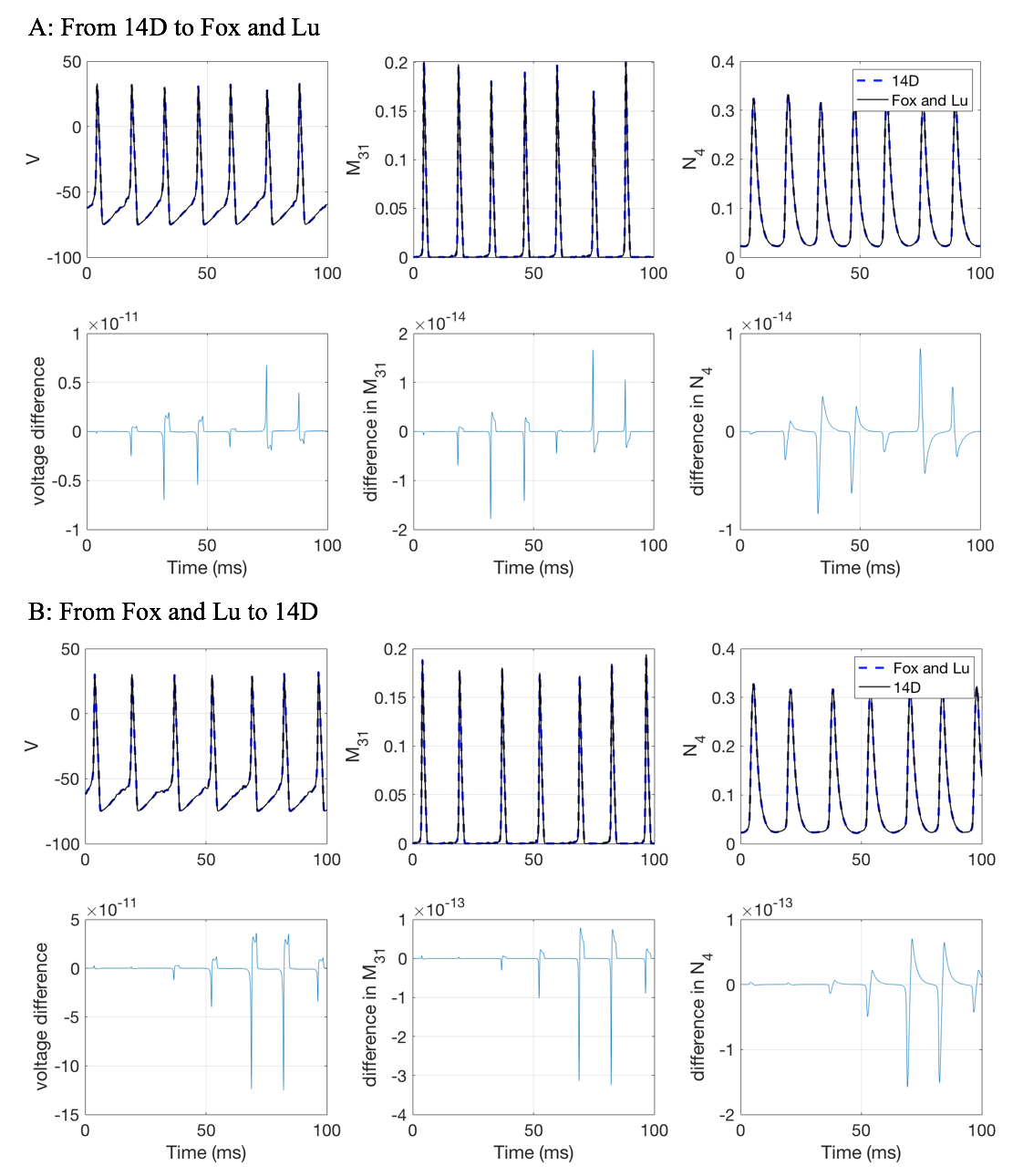}
\vspace*{-0.3cm}
\caption{{Pathwise equivalency of 14D HH model and Fox and Lu's model. \textbf{A:} Given a sample path of the $14\times28$D Langevin model in eqn.~\eqref{eq:LA}, we construct the noise by eqn.~\eqref{eq:14D_Fox} and generate the sample trajectory of Fox and Lu's model using eqn.~\eqref{eq:Fox}. \textbf{B:} Given a sample path of Fox and Lu's model in eqn.~\ref{eq:Fox}, we construct the noise by eqn.~\eqref{eq:Fox_14D} and generate the sample trajectory of the $14\times28$D Langevin model using eqn.~\eqref{eq:LA}. For both cases, we plot the voltage $V$, the open probability of the \Na~channel ($M_{31}$), and the open probability of the \K~channel ($N_4$). To show that these trajectories are pathwise equivalent, we superpose the trajectories for each variable and also plot the point-wise differences of each.  We obtain excellent agreement in both directions.}}
\label{fig:pwc}
\end{center}
\end{figure}

We have shown that our ``$14\times28$D" model, with a 14-dimensional state space and 28 independent noise sources (one for each directed edge) is pathwise equivalent to Fox and Lu's original 1994 model as well as other channel-based models (under corresponding boundary conditions) including \citep{Goldwyn2011PRE,GoldwynSheaBrown2011PLoSComputBiol,OrioSoudry2012PLoS1,Pezo2014Frontiers,Fox2018arXiv}.  
As we shall see in \S\ref{sec:modelcomp}, the pathwise equivalent models give statistically indistinguishable interspike interval distributions under the same BCs.  
We emphasize the importance of boundary conditions for pathwise equivalence.  
Two simulation algorithms with the same $A_i$ and $S_i$ matrices will generally have \emph{nonequivalent} trajectories if different boundary conditions are imposed.  
For example, \citep{Dangerfield2012APS} employs the same dynamics as \citep{OrioSoudry2012PLoS1} away from the boundary, where ion channel state occupancy approaches zero or one.  
But where the latter allow trajectories to move freely across this boundary (which leads only to small, short-lived excursions into ``nonphysical" values), Dangerfield imposes reflecting boundary conditions through a projection step at the boundary.
As we will see below (\S\ref{sec:modelcomp}), this difference in boundary conditions leads to a statistically significant difference in the ISI distribution, as well as a loss of accuracy when compared with the ``gold standard'' Markov chain simulation.

\section{Model Comparison}\label{sec:modelcomp}

In \S \ref{sec:stochastic_14DHH}, we studied the contribution of every directed edge to the ISI variability and proposed how stochastic shielding could be applied under current clamp. 
Moreover, in \S \ref{sec:path_equiv}, we proved that a family of Langevin models are pathwise equivalent.

Here we compare the accuracy and computational efficiency of several models, including the ``subunit model'' \citep{Fox1997BiophysicalJournal,GoldwynSheaBrown2011PLoSComputBiol}, Langevin models with different $S$ matrices or boudary conditions \citep{FoxLu1994PRE, GoldwynSheaBrown2011PLoSComputBiol,Dangerfield2012APS,OrioSoudry2012PLoS1,Pezo2014Frontiers}, the 14D HH model (proposed in \S \ref{subsec:14DHH}), the 14D stochastic shielding model with six independent noise sources (proposed in \S \ref{subsec:SS}), and the ``gold standard" Markov Chain model \rev{(discussed in \S\ref{subsec:RTC})}.
Where other studies have compared \rev{moment statistics} such as the mean firing frequency (under current clamp) and stationary channel open probababilies (under voltage clamp), we base our comparison on the entire interspike interval (ISI) distributions, under current clamp with a common fixed driving current.  We use two different comparisons of ISI distributions, the first based on the $L_1$ norm of the difference between two distributions (the Wasserstein distance, \citep{Wasserstein1969}), in \S\ref{subsec:ISIs} and the second based on the $L_\infty$ norm (the Kolmogorov-Smirnov test, \citep{Kolmogorov1933IIAG, Smirnov1948AMS}), in \S\ref{subsec:KS2}.  
We find similar results using both measures: as expected, the models that produce pathwise equivalent trajectories (Fox \& Lu '94, Orio \& Soudry, and our $14\times28$D model) have indistinguishable ISI statistics, while the non-equivalent models (Fox '97, Dangerfield, Goldwyn \& Shea-Brown, our \rev{$14\times6$D} stochastic-shielding model) have significantly different ISI distributions.  Of these, the \rev{$14\times6$D} SS model is the closest to the models in the $14\times28$D class, and as fast as any other model considered.


\subsection{\texorpdfstring{$\bf{L_1}$}{} Norm Difference of ISIs}\label{subsec:ISIs}


We first evaluate the accuracy of different stochastic simulation algorithms by comparing their ISI distributions under current clamp to that produced by a reference algorithm, namely the discrete-state Markov Chain (MC) algorithm.

Let $X_1, X_2, \ldots, X_n$ be $n$ independent samples of ISIs with a true cumulative distribution function $F$.
Let $F_n(\cdot)$ denote the corresponding empirical cumulative distribution function (ECDF) defined by
\begin{equation} \label{eq:ECDF}
 F_{n}(x)={\frac  1n}\sum _{{i=1}}^{n}{\mathbf  {1}}_{{\{X_{i}\leq x\}}},\qquad x\in {\mathbb  {R}},
\end{equation}
where we write $\mathbf{1}_A$ to denote the indicator function for the set $A$.   
Let $Q$ and $Q^M$ be the quantile functions of $F$ and $F^M$, respectively. The $L_1$-Wasserstein distance between two CDF's $F^M$ and $F$ can be written as \citep{Shorack2009SIAM} (page 64)
\begin{equation}
    \label{eq:L1Was}
    \rho_1(F,F^M)=\int_0^\infty \left|F(x)-F^M(x)\right|\,dx=\int_0^1 \left|Q(x)-Q^M(x)\right|\,dx.
\end{equation}
\rev{Note that $\rho_1$ has the same units as ``$dx$''.  Thus the $L_1$ distances reported in Tab.~\ref{tab:L1ISIs} have units of milliseconds.}

When two models have the same number of samples, $n$,  \eqref{eq:L1Was} can be estimated by
\begin{equation}
    \label{eq:L1Was_estimate}
   \int_0^1 \left|Q(x)-Q^M(x)\right|\,dx \approx \frac{1}{n} \sum_{i=1}^n \left|X_i-Y_i\right|:=\rho_1(F_n,F_n^M),
\end{equation}
where $X_1,\cdots,X_n$ and $Y_1,\cdots, Y_n$ are $n$ independent samples sorted in ascending order with CDF $F$ and $F^M$, respectively.

We numerically calculate $\rho_1(F_n,F_n^M)$ to compare several Langevin models against the MC model.
We consider the following models: ``Fox94" denotes the original model proposed by \citep{FoxLu1994PRE}, which requires a square root decomposition ($S=\sqrt{D}$) for each step in the simulation, see equations \eqref{eq:FoxandLu_dv}-\eqref{eq:FoxandLu_dK}. 
``Fox97" is the widely used ``subunit model" of \cite{Fox1997BiophysicalJournal}, see  equations \eqref{eq:Fox_dv}-\eqref{eq:Fox_dx}.
``Goldwyn" denotes the method taken from \citep{GoldwynSheaBrown2011PLoSComputBiol}, where they restrict the 14D system ($V$, 5 \K~gates and 8 \Na~gates) to the 4D multinomial submanifold ($V,\ m,\ n,\ \text{and}\ h$, see p.~\pageref{page:multinomial_mentioned_first} above), with gating variables truncated to $[0,1]$. 
We write ``Orio"  for the model proposed by \citep{OrioSoudry2012PLoS1}, where they constructed a rectangular matrix $S$ such that $SS^\intercal=D$ 
(referred to as $S_\text{paired}$ in Tab.~\ref{tab:L1ISIs})
combining fluctuations driven by pairs of  reciprocal edges, thereby avoiding taking matrix square roots at each time step. 
The model ``Dangerfield" represents \citep{Dangerfield2012APS}, which used the same $S$ matrix as in \citep{OrioSoudry2012PLoS1} but added a reflecting (no-flux) boundary condition via orthogonal projection
(referred to as  $S_{\text{EF}}$ in Tab.~\ref{tab:L1ISIs}).
Finally, we include the $14\times28$D model we proposed in \S \ref{subsec:14DHH}, or ``14D" (referred to as $S_\text{single}$ in Tab.~\ref{tab:L1ISIs});  ``SS" is the stochastic shielding model specified in \S \ref{subsec:SS}.

For each model, we ran 10,000 independent samples of the simulation, holding channel number, injected current \rev{($I_\text{app}=10$ nA)}, and initial conditions fixed.
\rev{Throughout the paper, we presume a fixed channel density of 60 $\text{channels}/\mu \text{m}^2$ for sodium and 18 $\text{channels}/\mu \text{m}^2$  for potassium in a membrane patch of area $100\,\mu\text{m}^2$, consistent with prior work such as \cite{GoldwynSheaBrown2011PLoSComputBiol, OrioSoudry2012PLoS1}.
The initial condition is taken to be the point on the deterministic limit cycle at which the voltage crosses upwards through $-60$ mV. 
An initial transient corresponding to 10-15 ISIs is discarded, to remove the effects of the initial condition.
See Tab.~\ref{tab:parameters} in Appendix \ref{axppend_alpha_beta} for a complete specification of simulation parameters.} 
We compared the efficiency and accuracy of each model through the following steps: 
\begin{enumerate}
    \item For each model, a single run simulates a total time of 84000 milliseconds (ms) with time step  0.008 ms, recording at least 5000 ISIs. 
    \label{sim:step}
    \item  For each model, repeat 10,000 runs in step one. 
    \item \rev{Create a reference ISI distribution by aggregating all 10,000 runs of the MC model, i.e.~based on roughtly $5\times 10^7$ ISIs.}
    \item \rev{For each of $10^4$ individual runs, align} all ISI data into a single vector and calculate the ECDF using equation \eqref{eq:ECDF}.
    \item Compare the ISI distribution of each model with the \rev{reference MC distribution}  by calculating the $L_1$-difference of the ECDFs using equation \eqref{eq:L1Was_estimate}.
    \item To compare the computational efficiency, we take the average execution time of the MC model as the reference. The relative computational efficiency is the ratio of the average execution time of a model with that of the MC model (c.~3790 sec.). 
\end{enumerate}

\begin{tiny}
\begin{table}[htbp]\centering
   \begin{tabular}{lcrccr} 
   Model  & Variables & $S$ Matrix&Noise Dim. & $L_1$ Norm (msec.) & Runtime \\
   &V+M+N & & Na+K & (Wasserstein Dist.) & (sec.)\\
   \hline
   MC & 1+8+5 & n/a & 20+8 & $2.27\,\text{e-}4\pm7.15\,\text{e-}5$ & 3790 \\ 
   Fox94 & 1+7+4 & $S=\sqrt{D}$ & 7+4 & $4.74\,\text{e-}2\pm1.93\,\text{e-}4$ & 2436 \\ 
   Fox97 & 1+2+1 & n/a & 3 & $8.01\,\text{e-}1\pm 9.48\,\text{e-}4$ & 67 \\ 
    Dangerfield& 1+8+5 & $S_\text{EF}$ & 10+4 & $2.18\,\text{e-}1\pm2.14\,\text{e-}4$ & 655 \\ 
    Goldwyn & 1+8+5 & $S=\sqrt{D}$ & 8+5 & $1.83\,\text{e-}1\pm1.93\,\text{e-}4$ & 2363 \\ 
   Orio & 1+8+5 & $S_{\text{paired}}$ & 10+4 & $4.52\,\text{e-}2\pm2.08\,\text{e-}4$ & 577 \\ 
     $14\times28$D & 1+8+5 & $S_{\text{single}}$ & 20+8 & $4.93\,\text{e-}2\pm1.94\,\text{e-}4$ & 605 \\ 
     SS & 1+8+5 & $S_{\text{ss}}$ & 4+2 & $7.62\,\text{e-}2\pm7.57\,\text{e-}5$ & 73 \\ 
  \hline 
   \end{tabular}
   \caption{Summary of the $L_1$-Wasserstein distances of ISI distributions for Langevin type Hodgkin-Huxley models compared to the MC model.
   Model (see text for details): MC: Markov-chain. Fox1994: model from \cite{FoxLu1994PRE}. Fox97:  \cite{Fox1997BiophysicalJournal}.  Goldwyn: \cite{GoldwynSheaBrown2011PLoSComputBiol}. Dangerfield:  \cite{Dangerfield2010PCS, Dangerfield2012APS}. $14\times28$D:  model proposed in \S \ref{subsec:14DHH}. SS:  stochastic-shielding model (\S \ref{subsec:SS}).
   Variables: number of degrees of freedom in Langevin equation representing voltage, sodium gates, and potassium gates, respectively. 
   $S$ Matrix: Form of the noise coefficient matrix in equations \eqref{eq:FoxandLu_dv}-\eqref{eq:FoxandLu_dK}.  Noise Dimensions: number of independent Gaussian white-noise sources represented for sodium and potassium, respectively. $L_1$ Norm: Empirically estimated $L_1$-Wasserstein distance between the model's ISI distribution and the MC model's ISI distribution. For MC-vs-MC, independent trials were compared.  $a\pm b$: mean$\pm$standard deviation.  Runtime (in sec.): see text for details. n/a: not applicable.  }
   \label{tab:L1ISIs}
\end{table}
\end{tiny}
Table \ref{tab:L1ISIs} gives the empiricially measured $L_1$ difference in ISI distribution between several pairs of models.\footnote{Runtimes in Tab.~\ref{tab:L1ISIs}, rounded to the nearest integer number of seconds, were obtained by averaging the runtimes on a distribution of heterogeneous \rev{compute nodes} from  Case Western Reserve University's  high-performance computing cluster.} 
\rev{The first row (``MC") gives the average $L_1$ distance between individual MC simulations and the  reference distribution generated by aggregating all MC simulations, in order to give an estimate of the intrinsic variability of the measure.}  
Figure \ref{fig:L1ISI} plots the $L_1$-Wasserstein differences versus the relative computational efficiency of several models against the MC model.
These results suggest that the Fox94, Orio, and $14\times28$D models are statistically indistinguishable, when compared with the MC model using the $L_1$-Wasserstein distance.
This result is expected in light of our results (\S \ref{sec:path_equiv}) showing that these three models are pathwise-equivalent.
(We will make pariwise statistical comparisons between the ISI distributions of each model in \S \ref{subsec:KS2}.)
Among these equivalent models, however, the $14\times28$D and Orio models are significantly faster than the original Fox94 model (and the Goldwyn model) because they avoid the matrix square root computation.
The Dangerfield model has speed similar to the $14\times28$D model, but the use of reflecting boundary conditions introduces significant inaccuracy in the ISI distribution.
The imposition of truncating boundary conditions in the Goldwyn model also appears to affect the ISI distribution.  

\begin{wrapfigure}{r}{0.55\textwidth}
\includegraphics[width=0.9\linewidth]{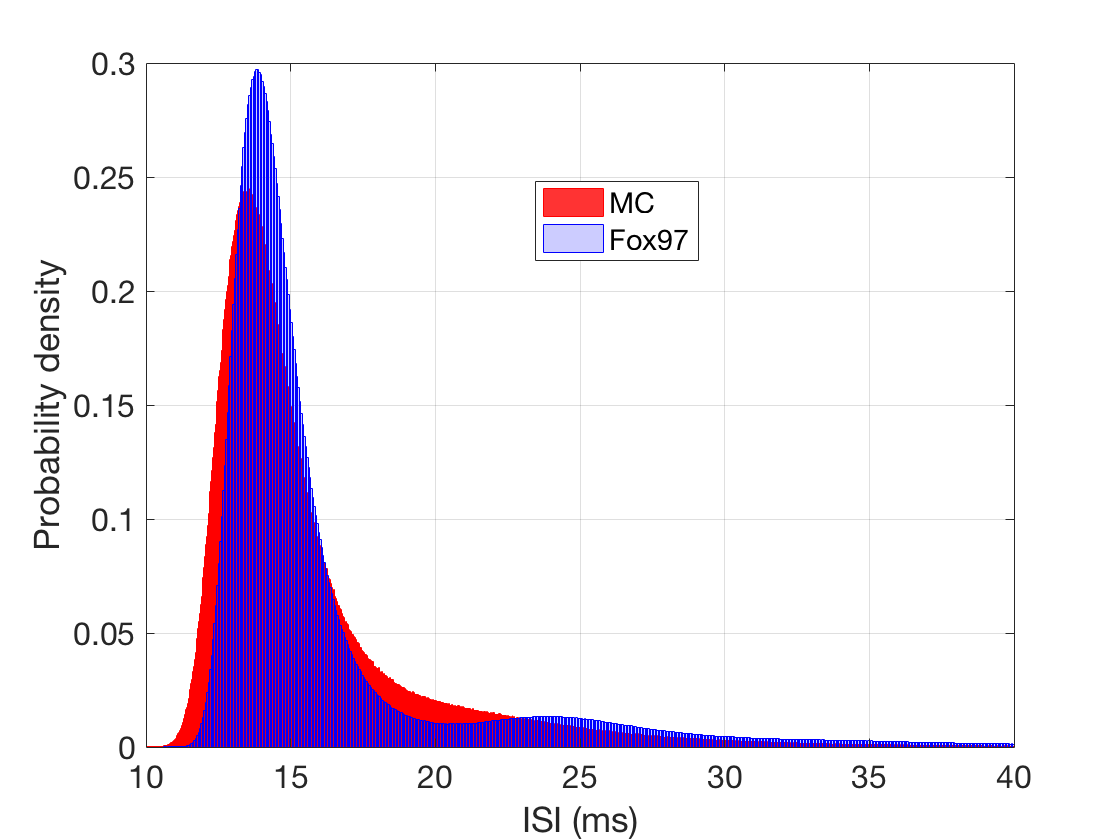} 
\caption{The probability density of interspike intervals (ISIs) for Fox97 (blue) and the MC model (red). The probability densities were calculated over more than $5.4\times 10^7$ ISIs.}\label{fig:Fox97vsMC}
\end{wrapfigure}
Of the models considered, the Fox97 subunit model is the fastest, however it makes a particularly poor approximation to the ISI distribution of the MC model. Note that the maximum $L_1$-Wasserstein distance between two distributions is 2.
The ISI distribution of Fox97 subunit model to that of the MC model is more than 0.8, which is ten times larger than the $L_1$-Wasserstein distance of the SS model, and almost half of the maximum distance. 
As shown in Fig.~\ref{fig:Fox97vsMC}, the Fox97 subunit model fails to achieve the spike firing threshold and produces longer ISIs.
Because of its inaccuracy, we do not include the subunit model in our remaining comparisons. 
The stochastic shielding model, on the other hand, has nearly the same speed \rev{as the Fox97 model}, but is over 100 times more accurate (in the $L_1$ sense).
The SS model is an order of magnitude faster than the $14\times28$D model, and has less than twice the $L_1$ discrepancy versus the MC model \rev{($L_1$ norm 76.2 versus 49.3 microseconds).} 
While this difference in accuracy is statistically significant, it may not be practically significant, depending on the application \rev{(see \S\ref{sec:discussion} for further discussion of this point).}

\begin{figure}
\begin{center}
\includegraphics[scale=0.3]{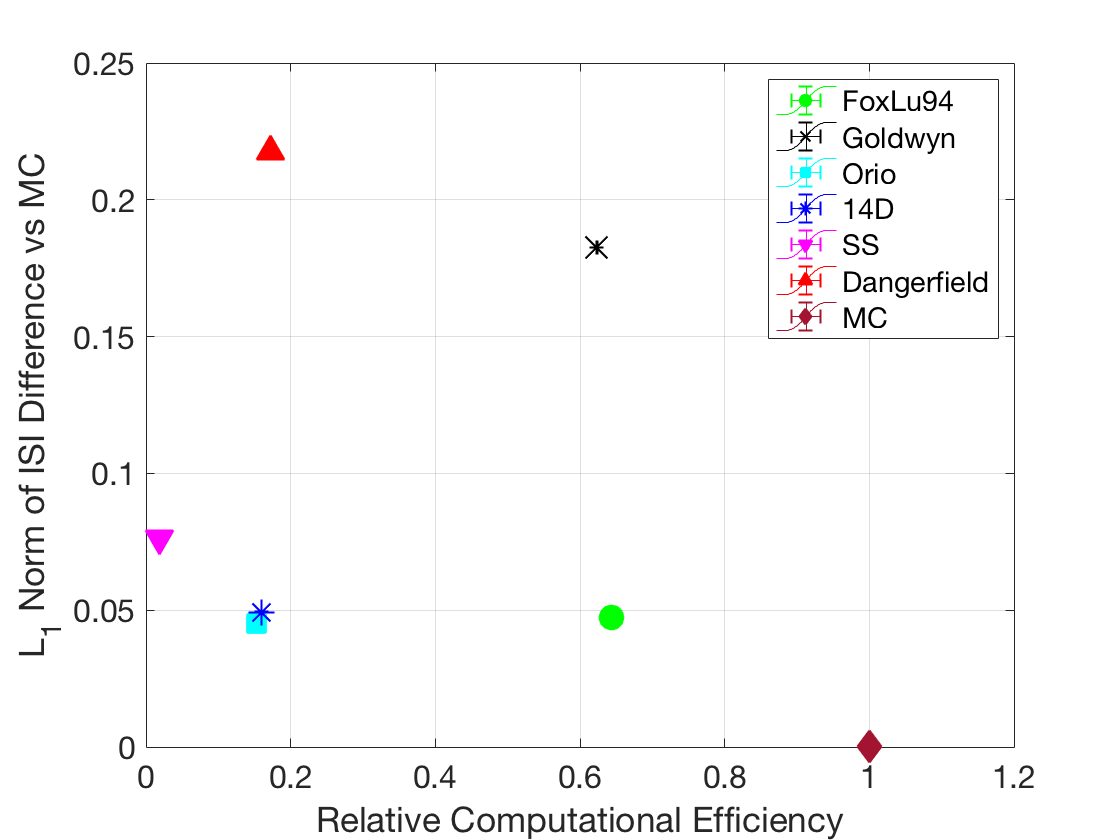}
\caption{The $L_1$-Wasserstein distances and relative computational efficiency vs. the MC model.
``Fox94" (green circle), ``Goldwyn" (black cross), ``Orio" (cyan square),  ``14D" (blue star), ``SS" (magenta downward pointing triangle), ``Dangerfield" (red upward pointing triangle), and the ``MC" (brown diamond) model.
The $L_1$ error for ISI distribution was computed using the $L_1$-Wasserstein distance \eqref{eq:L1Was_estimate}, with  discrete time Gillespie/Monte-Carlo simulations as
a reference.
The relative computational efficiency is the ratio of the recorded run time to the mean recorded time of the MC mode (3790 seconds).
The mean and 95\% confidence intervals were calculated using 100 repetitions of $10,000$ runs each ($5\times 10^9$ ISIs total). 
} 
\label{fig:L1ISI}
\end{center}
\end{figure}

\subsection{Two-sample Kolmogorov-Smirnov Test}\label{subsec:KS2}
In addition to using the $L_1$-Wasserstein distances to test the differences between two CDFs, we can also make a pairwise comparison between each model by applying the Dvoretzky-Kiefer-Wolfowitz inequality \citep{Dvoretzky1956AMS} and the two-sample Kolmogorov-Smirnov (KS) test \citep{Kolmogorov1933IIAG, Smirnov1948AMS}.  While the Wasserstein distance is based on the $L_1$ norm, the KS statistic is based on the $L_\infty$ (or supremum) norm.

The Dvoretzky-Kiefer-Wolfowitz inequality \citep{Dvoretzky1956AMS} 
establishes confidence bounds for the CDF. Specifically, the interval that contains the true CDF, $F(\cdot)$, with probability $1-\alpha$, is given by
\begin{equation} \label{eq:DKW}
    |F_{n}(x)-F(x)| \leq \varepsilon \;{\text{ where }}\varepsilon ={\sqrt {\frac {\ln {\frac {2}{\alpha }}}{2n}}}.
\end{equation}
When comparing samples $X^M_1,X^M_2,\ldots,X^M_n$ obtained from an approximate model $M$ against the gold standard, in \S\ref{subsec:ISIs} we computed the $L_1$ difference of the empirical density functions, as an approximation for the $L_1$ difference of the true distributions.
\rev{Instead,} we work \rev{here} with the $L_\infty$ norm,
\begin{equation}
\rho_\infty(F_n,F_n^M)=\lim_{p\to\infty} \left(\int_0^\infty\left|F_n^M(x)-F_n(x)\right|^p\,dx\right)^{1/p}=\sup_{0\le x < \infty}\left(\left|F_n^M(x)-F_n(x)\right|\right).
\end{equation}

For each $x\ge 0$,  equation \eqref{eq:DKW} bounds the discrepancy between the true and empirical distribution differences as follows.  
By the triangle inequality, and independence of the $X_i$ from the $X^M_i$, the inequality 
\begin{eqnarray}
\lvert F^M-F \rvert &=& \lvert F^M-F^M_n+F_n-F+F^M_n-F_n \rvert\nonumber \\
&\leq& \lvert F^M-F_n^M \rvert +\lvert F_n-F \rvert +\lvert F_n^M-F_n \rvert \nonumber\\
&\leq& 2\varepsilon +\lvert F_n^M-F_n \rvert, \label{eq:FF1}
\end{eqnarray}
holds with probability $(1-\alpha)^2$.  
Similarly,
\begin{eqnarray}\nonumber
\lvert F_n^M-F_n \rvert &=& \lvert F_n^M-F^M_n+F-F_n+F^M-F \rvert \\ \nonumber
&\leq& \lvert F^M-F_n^M \rvert +\lvert F_n-F \rvert +\lvert F^M-F \rvert \\
&\leq& 2\varepsilon +\lvert F^M-F \rvert \label{eq:FF2}
\end{eqnarray}
also holds 
with probability $(1-\alpha)^2$.
Together, \eqref{eq:FF1}-\eqref{eq:FF2} indicate that the discrepancy between the difference of empirical distributions and the difference of true distributions is bounded as 
\begin{equation}
  \Big\lvert  \lvert F^M-F \rvert-\lvert F_n^M-F_n \rvert \Big\lvert \leq 2\varepsilon
\end{equation}
with probability $(1-\alpha)^2$, for $\varepsilon ={\sqrt {\frac {\ln {\frac {2}{\alpha }}}{2n}}}$. 

We will use
the pointwise difference of the ECDF's for a large sample as an
estimate for the pointwise  difference between two true CDFs.
%
The  two-sample Kolmogorov-Smirnov (KS) test \citep{Kolmogorov1933IIAG, Smirnov1948AMS} offers a statistics to test whether two samples are from the same distribution. The two-sample KS \rev{statistic} is 
\begin{equation}\label{eq:KS2_D}
    D_{n,m}=\sup _{x}|F_{1,n}(x)-F_{2,m}(x)|,
\end{equation}
where $F_{1,n}$ and $F_{2,m}$ are two ECDFs for two samples defined in  \eqref{eq:ECDF}, and the $\sup$ is the supremum function. 
The reference statistic, $R_{n,m}(\alpha)$, depending on the significance level $\alpha$, is defined as 
\begin{equation}\label{eq:KS2_R}
   R_{n,m}(\alpha)=\sqrt{\frac{-\log(\alpha/2)}{2}}\sqrt{\frac{n+m}{nm}},
\end{equation}
where $n$ and $m$ are the sample sizes.
The null hypothesis that ``the two samples  come from the same distribution" is rejected at the significance level $\alpha$ if 
\begin{equation}\label{eq:KS2_comp}
    D_{n,m}>R_{n,m}(\alpha).
\end{equation}

\begin{figure}
\begin{center}
\hspace{-0.9cm}
\includegraphics[height=5.3cm,width=9.0cm]{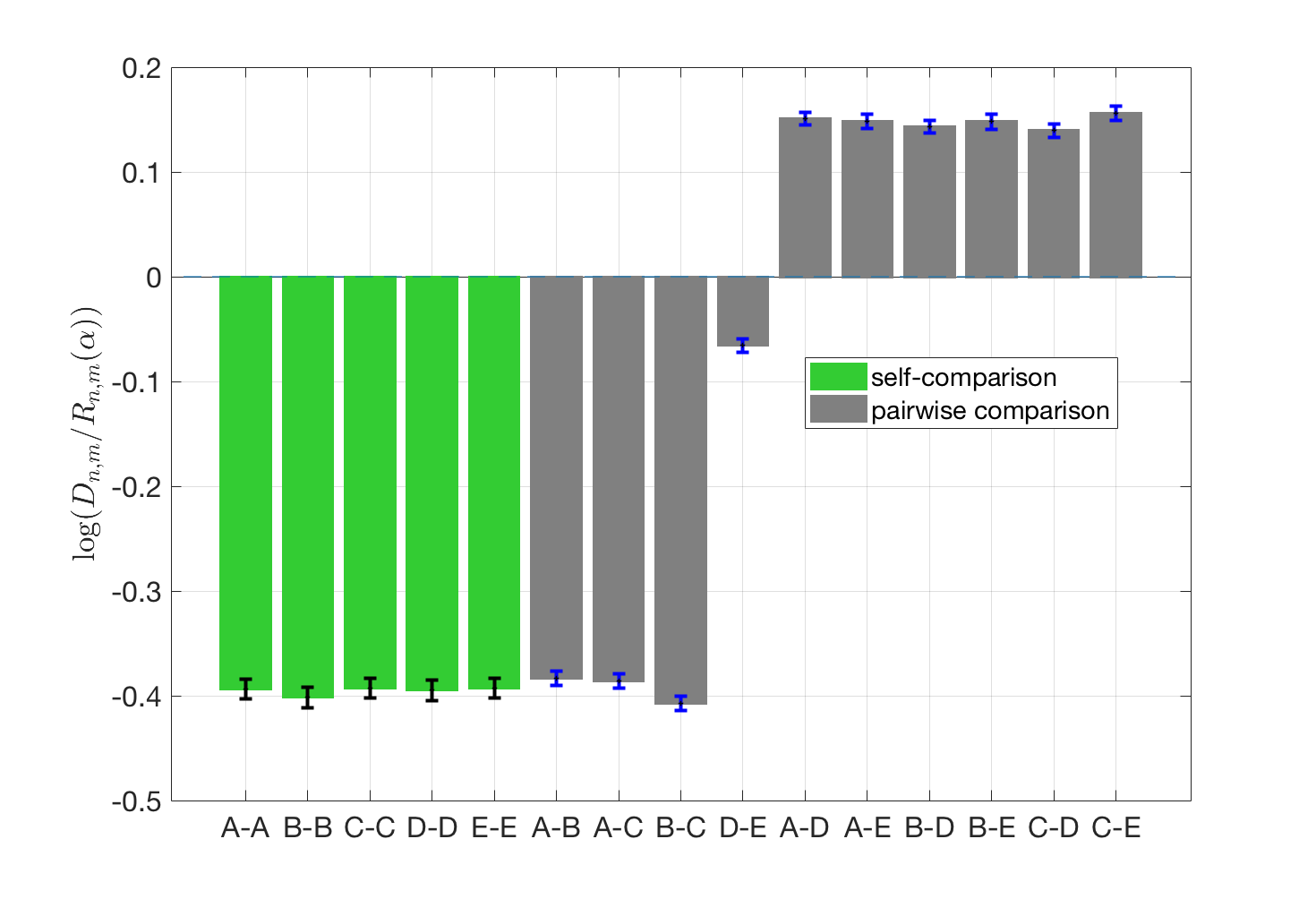}
\hspace{-0.5cm}
\includegraphics[height=5cm,width=6.6cm]{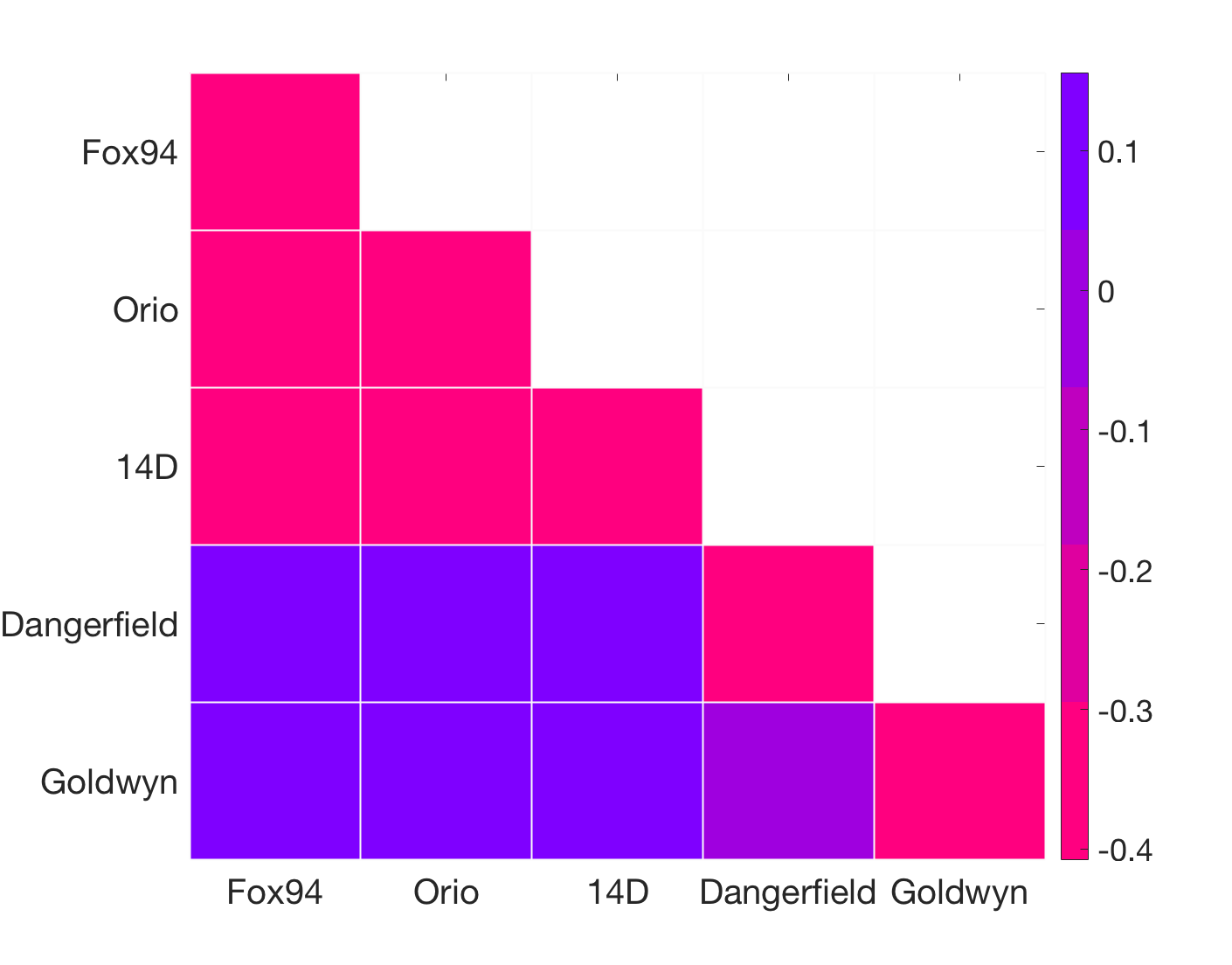}
\caption{Logarithm of the ratio of Kolmogorov-Smirnov test statistic $D_{n,m}/R_{n,m}(\alpha)$, eqns.~\eqref{eq:KS2_D}-\eqref{eq:KS2_R}, for samples from the ISI distribution for each pair of models.
\textbf{Left:} 
Box and whisker plots showing mean and 95\% confidence intervals based on 10,000 pairwise comparisons.  
The first five plots show self-comparisons (green bars); the remainder compare distinct pairs (grey bars). 
A:``Fox94" \citep{FoxLu1994PRE},
B:``Orio" \citep{OrioSoudry2012PLoS1},
C: ``14D" (the $14\times28$D model we proposed in \S \ref{subsec:14DHH}),
D: Dangerfield \citep{Dangerfield2012APS},
E: Goldwyn \citep{GoldwynSheaBrown2011PLoSComputBiol}. 
The models of Fox and Lu, Orio and Soudry, and our 14D model generate indistinguishable ISI distributions, but are distinguishable from  Dangerfield's model and Fox's 97 model.
\textbf{Right:} Mean logarithms (as in left panel) for all pairwise comparisons.  Fox94, Orio and $14\times28$D form a block of statistically indistinguishable samples.
}\label{fig:KS1}
\end{center}
\end{figure}

Figure \ref{fig:KS1} plots the logarithm of ratio of the two-sample KS statistics, $\frac{D_{n,m}}{R_{n,m}(0.01)}$, for ``Fox94" \citep{FoxLu1994PRE}, ``Goldwyn" \citep{ GoldwynSheaBrown2011PLoSComputBiol},``Dangerfiled" \citep{Dangerfield2012APS}, ``Orio" \citep{OrioSoudry2012PLoS1}, ``14D" (the $14\times28$D model we proposed in \S \ref{subsec:14DHH}). 
Data of ``self-comparison" (e.g. Fox94 vs. Fox 94) was obtained by comparing two ISI ECDF's from independent simulations. 
As shown \rev{in} Fig.~\ref{fig:KS1}, models that we previously proved were pathwise equivalent in \S \ref{sec:path_equiv}, namely the ``Fox94", ``Orio" and the ``14D" model, are not distinguishable at any confidence level justified by our data. 
Note that those three models use the same boundary conditions (free boundary condition as in \cite{OrioSoudry2012PLoS1}) and the ratio $D_{n,m}/R_{n,m}(\alpha)$ of pairwise comparison is on the same magnitude of that for the self-comparisons.
However, as pointed out above, these statistically equivalent simulation algorithms have different computational efficiencies (Fig.~\ref{fig:L1ISI}).
Among these methods, Orio and Soundry's algorithm (14 dimensional state space with 14 undirected edges as noise sources) and our method (14 dimensional state space with 28 directed edges as noise sources) have similar efficiencies, with Orio's method being about 5\% faster than ours method.  
Our $14\times28$D method provides the additional advantage that it facilitates further acceleration under the stochastic shielding approximation (see \S \ref{sec:discussion}).

In contrast to the statistically equivalent Orio, $14\times28$D and Fox '94 models, algorithms using different boundary conditions are not
pathwise equivalent, which is again verified in Fig.~\ref{fig:KS1}. Algorithms with subunit approximation and truncated boundary condition (i.e., ``Goldwyn") and the reflecting boundary condition (i.e. ``Dangerfield") are significantly different in accuracy (and in particular, they are less accurate) than models in the $14\times28$D class.


\section{Discussion \& Conclusions}\label{sec:discussion}
\subsection{Summary}
The exact method for Markov Chain (MC) simulation for an electrotonically compact (single compartment) conductance-based stochastic model under current clamp is a hybrid discrete (channel state) / continuous (voltage) model of the sort used by \citep{ClayDeFelice1983BiophysJ,NewbyBressloffKeener2013PRL,AndersonErmentroutThomas2015JCNS}. 
While MC methods are computationally expensive, simulations based on Gaussian/Langevin approximation can capture the effects of stochastic ion channel fluctuations with reasonable accuracy and excellent computational efficiency.  Since Goldwyn and Shea Brown's work focusing the attention of the computational neuroscience community on Fox and Lu's Langevin algorithm for the Hodgkin-Huxley system \citep{FoxLu1994PRE,GoldwynSheaBrown2011PLoSComputBiol}, several variants of this approach have appeared in the literature.  

In the present paper we advocate for a class of models combining the best features of conductance-based Langevin models with the recently developed stochastic shielding approximation \citep{SchmandtGalan2012PRL,SchmidtThomas2014JMN,SchmidtGalanThomas2018PLoSCB}.  
We propose a Langevin model with a 14-dimensional state space, representing the voltage, five states of the \K~channel, and eight states of the \Na-channel; and a 28-dimensional representation of the driving noise: one independent Gaussian noise term for each directed edge in the channel-state transition graph.  
We showed in \S\ref{sec:determ_14D} that the corresponding mean-field 14D ordinary differential equation model is consistent with the classical HH equations in the sense that the latter correspond to an invariant submanifold of the higher-dimensional model, to which all trajectories converge exponentially quickly.  
\rev{ Fig.~\ref{fig:HHcomp} illustrated the relation between the deterministic 4D and 14D Hodgkin Huxley systems.}
Building on this framework, we introduced the $14\times28$D model, with independent noise sources corresponding to each ion channel transition (\S\ref{sec:stochastic_14DHH}).  
We proved in \S\ref{sec:path_equiv} that, given identical boundary conditions, our $14\times28$D model is pathwise equivalent both to Fox and Lu's original Langevin model, and to A 14-state model with 14 independent noise sources due to \citep{OrioSoudry2012PLoS1}.  

\rev{The original 4D HH model, the 14D deterministic HH model, and the family of  equivalent 14D Langevin models we consider here, form a nested family, each contained within the next.  
Specifically, (i) the 14D ODE model is the ``mean-field" version of the 14D Langevin model, and (ii) the 4D ODE model forms an attracting invariant submanifold within the 14D ODE model, as we establish in our Lemma \ref{LemmaXinvari}.  
So in a very specific sense, the original HH equations ``live inside" the 14D Langevin equations.
Thus these three models enjoy a special relationship.  
In contrast, the 4D Langevin equations studied in \cite{Fox1997BiophysicalJournal} do not bear an especially close relationship to the other three. 
}

In \rev{addition} to rigorous mathematical analysis we also performed numerical comparisons (\S\ref{sec:modelcomp}) showing that, as expected, the pathwise equivalent models produced statistically indistinguishable interspike interval (ISI) distributions. 
Moreover, the ISI distributions for our model (and its equivalents) were closer to the ISI distribution of the ``gold standard" MC model under two different metric space measures.  Our method (along with Orio and Soudry's) proved computationally more efficient than Fox and Lu's original method and Dangerfield's model \citep{Dangerfield2012APS}.
In addition, our method lends itself naturally to model reduction (via the stochastic shielding approximation) to a significantly faster \rev{$14\times 6$D} simulation  that preserves a surprisingly high level of accuracy. 

\subsection{Discrete Gillespie Markov Chain algorithms}
The discrete-state Markov Chain algorithm due to Gillespie is often taken to be the gold standard simulation for a single-compartment stochastic conductance-based model.
Most former literature on Langevin HH models, such as \citep{GoldwynSheaBrown2011PLoSComputBiol, Linaro2011PublicLibraryScience, OrioSoudry2012PLoS1,Huang2013APS},
when establishing a reference MC model, consider a version of the discrete Gillespie algorithm that assumes a piecewise-constant propensity approximation, i.e.~that does not take into account that the voltage changes between transitions, which changes the transition rates. 
This approximation can lead to biophysically unrealistic voltage traces for very small system sizes (cf.~Fig.~2 of \citep{KisperskyWhite2008Scholarpedia}, top trace with $N=1$ ion channel) although the differences appear to be mitigated for  $N\gtrsim 40$ channels \citep{AndersonErmentroutThomas2015JCNS}.
In the present paper, our MC simulations are based on 6000 \Na~and 1800 \K~channels (as in \cite{GoldwynSheaBrown2011PLoSComputBiol}), and we too use the ISI distribution generated by a piecewise-constant propensity MC algorithm as our reference distribution. 
As shown in Tab.~\ref{tab:L1ISIs} and Fig.~\ref{fig:L1ISI}, the computation time for MC is one order of magnitude larger than efficient methods such as \citep{OrioSoudry2012PLoS1}, \citep{Dangerfield2012APS} and the $14\times28$D model.
The computational cost of the MC model increases dramatically as the number of ion channels grows, therefore, even the approximate MC algorithm is inapplicable for a large number of channels. 

\subsection{Langevin Models}
It is worth pointing out that the accuracy of Fox and Lu's original Langevin equations has not been fully appreciated. In fact, Fox and Lu's model \citep{FoxLu1994PRE} gives an approximation to the MC model that is just as accurate as \citep{OrioSoudry2012PLoS1} both in the gating variable statistics \citep{GoldwynSheaBrown2011PLoSComputBiol} and also in the ISI distribution sense (see \S \ref{sec:modelcomp}) -- because as we have established, these models are pathwise equivalent! 
However, the original implementation requires taking a matrix square root in every timestep in the numerical simulation, which significantly reduces its computational efficiency.

Models based on modifications of  \citep{FoxLu1994PRE}'s work can be divided into three classes: the subunit model \citep{Fox1997BiophysicalJournal}; effective noise models  \citep{Linaro2011PublicLibraryScience,Guler2013MITPress}, and channel-based Langevin models such as \citep{GoldwynSheaBrown2011PLoSComputBiol,Dangerfield2012APS,OrioSoudry2012PLoS1,Pezo2014Frontiers,Huang2013APS}.

\paragraph{Subunit model} 
The first modification of the Fox and Lu's model is the subunit model \citep{Fox1997BiophysicalJournal}, which keeps the original form of the HH model, and adds noise to the gating variables ($m,\ h,\ \text{and}\ n$) \citep{Fox1997BiophysicalJournal,GoldwynSheaBrown2011PLoSComputBiol}. The subunit approximation model was widely used because of its fast computational speed.
However, as \citep{Bruce2009ABE} and others pointed out, the inaccuracy of this model remains significant even for large number of channels. Moreover, \citep{GoldwynSheaBrown2011PLoSComputBiol} and \citep{Huang2015PhBio} found that the subunit model fails to capture the statistics of the HH \Na~and \K~ gates. 
In this paper, we also observed that the subunit model is more efficient than channel-based Langevin models, but tends to delay spike generation. As shown in Fig.~\ref{fig:Fox97vsMC}, the subunit model generates significantly longer ISIs than the MC model.

\paragraph{Effective noise models} Another modification to Fox and Lu's algorithm is to add colored noise to the channel open fractions. Though  colored noise models such as \citep{Linaro2011PublicLibraryScience,Guler2013MITPress} are not included in our model comparison, \citep{Huang2015PhBio} found that both these effective noise models generate shorter ISIs than the MC model with the same parameters. 
Though the comparison we provided in \S\ref{sec:modelcomp} only include the Fox and Lu 94, Fox97, Goldwyn, Dangerfield, Orio, SS and the $14\times28$D model, combining the results from \citep{GoldwynSheaBrown2011PLoSComputBiol} and \citep{Huang2015PhBio}, the $14\times28$D model could be compared to a variety of models including \citep{FoxLu1994PRE,Fox1997BiophysicalJournal,GoldwynSheaBrown2011PLoSComputBiol, Linaro2011PublicLibraryScience, OrioSoudry2012PLoS1,Dangerfield2012APS,Huang2013APS,Guler2013MITPress}.

\paragraph{Channel-based Langevin models} The main focus of this paper is the modification based on the original Fox and Lu's matrix decomposition method, namely, the channel-based (or conductance-based) Langevin models.
We proved in \S\ref{sec:path_equiv} that under the same boundary conditions, Fox and Lu's original model, Orio's model and  our $14\times28$D model are pathwise equivalent, which was also verified from our numerical simulations in \S\ref{sec:path_equiv} and \S\ref{sec:modelcomp}.
In \S\ref{sec:path_equiv}, we discussed channel-based Langevin models including \citep{FoxLu1994PRE,GoldwynSheaBrown2011PLoSComputBiol,Dangerfield2012APS,OrioSoudry2012PLoS1,Fox2018arXiv}.
We excluded Fox's more recent implementation \citep{Fox2018arXiv} in \S\ref{sec:modelcomp} for two reasons.  First, the algorithm is pathwise equivalent to others considered there.  And moreover, the method is vulnerable to numerical instability when performing the Cholesky decomposition. 
Specifically, some of the elements in the $S$ matrix from the Cholesky decomposition in \citep{Fox2018arXiv} involve square roots of differences of several quantities, with no guarantee that the differences will result in nonnegative terms --  even with strictly positive value of the gating variables.
Nevertheless, this model would be in the equivalence class and in any case would not be more efficient than the Orio's model, because of the noise dimension and complicated operations (involving taking multiple square roots) in each step.

\subsection{Model comparisons}
If two random variables have similar distributions, then they will have similar moments, but not vice-versa.  Therefore, comparison of the full interspike-interval distributions produced \rev{by} different simulation algorithms gives a more rigorous test than comparison of first and second moments of the ISI distribution.  
Most previous evaluations of competing Langevin approximations were based on the accuracy of low-order moments, for example the mean and variance of channel state occupancy under voltage clamp, or the mean and variance of the interspike interval distribution under current clamp
\citep{Goldwyn2011PRE,GoldwynSheaBrown2011PLoSComputBiol,SchmandtGalan2012PRL,Linaro2011PublicLibraryScience,Dangerfield2012APS,OrioSoudry2012PLoS1,Huang2013APS,Huang2015PhBio}. 
Here, we compare the accuracy of the different algorithms using the full ISI distribution, but using the $L_1$ norm of the difference (Wasserstein metric) and the $L_\infty$ norm (Kolmogorov-Smirnov test).   
\cite{GreenwoodMcDonnelWard2014NECO} previously compared the ISI distributions generated by the Markov chain (Gillespie algorithm) to the distribution generated by different types of Langevin approximations (LA), including the original Langevin models  \citep{FoxLu1994PRE,GoldwynSheaBrown2011PLoSComputBiol}, the channel-based  LA with colored noise \citep{Linaro2011PublicLibraryScience, Guler2013MITPress},  and LA with a $14\times 14$ variant of the diffusion coefficient matrix $S$  \citep{OrioSoudry2012PLoS1}.
They concluded that Orio and Soudry's method provided the best match to the Markov chain model, specifically ``Fox-Goldwyn, and Orio-Kurtz\rev{\footnote{\rev{We refer to this model as to as ``Orio"}}} methods both generate ISI histograms very close to those of Micro\rev{\footnote{\rev{This is the model we refer to as the Markov-chain model.}}}" \citep{GreenwoodMcDonnelWard2014NECO}.  
We note that the comparison reported in this paper simply superimposed plots of the ISI distributions, allowing a qualitative comparison, while our metric-space analysis is fully quantitative.
In any case, their conclusions are consistent with our findings; we showed in \S\ref{sec:path_equiv} that the Fox-Goldwyn and the Orio-Kurtz model are pathwise equivalent (when implemented with the same boundary conditions), which accounts for the similarity in the ISI histograms they generate.  
In fact, because of pathwise equivalence, we can conclude that the true distributions for these models are identical, and any differences observed just reflect finite sampling.

\subsection{Stochastic Shielding Method} 
The stochastic shielding (SS) approximation \citep{SchmandtGalan2012PRL} provides an efficient and accurate method for approximating the Markov process with only a subset of observable states. For conductance-based models, rather than aggregating ion channel states, SS effects dimension reduction by selectively eliminating those independent noise sources that have the least impact on current fluctuations. 
Recent work in \citep{Pezo2014Frontiers} compared previous methods such as \citep{Gillespie1977, OrioSoudry2012PLoS1, Dangerfield2012APS, Huang2013APS, SchmandtGalan2012PRL} in accuracy, applicability and simplicity as well as computational efficiency.  They concluded  that for mesoscopic numbers of channels, stochastic shielding methods combined with diffusion approximation methods can be an optimal choice.
Like \citep{OrioSoudry2012PLoS1}, the stochastic shielding method proposed by \citep{Pezo2014Frontiers} also assumed detailed balance of transitions between adjacent states and used edges that are directly connected to the open gates of HH \Na~and~\K.
We calculated the edge importance in \S\ref{subsec:SS} and found that the four (out of twenty) most important directed edges for the \Na~gates are not the four edges directly connected to the conducting state, as assumed in previous application of the SS method \citep{SchmandtGalan2012PRL}.

\subsection{Which model to use?} Among all modifications of  Fox and Lu's method considered here, Orio and Soudry's approach, and our $14\times28$D model provide the best approximation to the ``gold standard" MC model, with the greatest computational efficiency.  
Several earlier models were studied in the review paper by \citep{GoldwynSheaBrown2011PLoSComputBiol}, where they rediscovered that the original Langevin model proposed by Fox and Lu is the best approximation to the MC model among those considered. 
Later work \citep{Huang2015PhBio} further surveyed a wide range of Langevin approximations for the HH system including \citep{FoxLu1994PRE,Fox1997BiophysicalJournal,GoldwynSheaBrown2011PLoSComputBiol,Linaro2011PublicLibraryScience, Guler2013MITPress, OrioSoudry2012PLoS1,Huang2013APS} and explored models with different boundary conditions. 
\cite{Huang2015PhBio} concluded that the bounded and truncated-restored Langevin model \citep{Huang2013APS} and the unbounded \citep{OrioSoudry2012PLoS1}'s model provide the best approximation to the MC model.

As shown in \S\ref{sec:path_equiv} and \S\ref{sec:modelcomp}, the $14\times28$D Langevin model naturally derived from the channel structure is pathwise equivalent to the Fox and Lu `94, Fox `18, and Orio-Soudry models under the same boundary conditions. 
The  $14\times28$D model is more accurate than the reflecting boundary condition method of \citep{Dangerfield2012APS}, and also better than the approximation method proposed by \citep{GoldwynSheaBrown2011PLoSComputBiol},
when the entire ISI distribution is taken into account. 
We note that \citep{Huang2015PhBio} treated Goldwyn's method \citep{GoldwynSheaBrown2011PLoSComputBiol} as the original Fox and Lu model in their comparison, however, the simulation in \citep{GoldwynSheaBrown2011PLoSComputBiol} uses the $4D$ multinomial submanifold to update gating variables.
Our analysis and numerical simulations suggest that the original Fox and Lu model is indeed as accurate as the Orio-Soudry model, while the computational cost still remains a major concern. 

Though the $14\times28$D model has similar efficiency and accuracy with \cite{OrioSoudry2012PLoS1}, it \SP{has} several advantages.
First, the rectangular $S$ matrix (in eqn.~\eqref{eq:FoxandLu_dNa}-\eqref{eq:FoxandLu_dK}) in Orio's model merges the noise contributions of reciprocal pairs of edges. However, this dimension reduction assumes, in effect, that detailed balance holds along reciprocal edges, which our results show is not the case, under current clamp  (Fig.~\ref{fig:logSSNaK}).
Moreover, the $14\times28$D model arises naturally from the individual transitions of the exact evolution equations (eqn.\eqref{eq:RTCNa}-\eqref{eq:RTCK}) for the underlying Markov Chain model, which makes it conceptually easier to understand.
In addition, our method for defining the $14\times28$D Langevin model and finding the best SS model extends to channel-based models with arbitrary channel gating schemes beyond the standard HH model.
Given any channel state transition graph, the Langevin equations may be read off from the transitions, and the edge importance under current clamp can be evaluated by applying the stochastic shielding method to investigate the contributions of noise from each individual directed edge.
Finally, in exchange for a small reduction in accuracy, the stochastic shielding method affords a significant gain in efficiency.  
The $14\times28$D model thus offers a natural way to quantify the contributions of the microscopic transitions to the macroscopic voltage fluctuations in the membrane through the use of stochastic shielding. 
For general ion channel models, extending a biophysically-based Langevin model analogous to our $14\times28$D HH model, together with the  stochastic shielding method, may provide the best available tool for investigating how unobservable  microscopic behaviors (such as ion channel fluctuations)  affect the macroscopic variability in many biological systems.

\subsection{Limitations}
All Langevin models, including our proposed $14\times28$D model, proceed from the assumption that the ion channel population is large enough (and the ion channel state transitions frequent enough) that the Gaussian approximations by which the white noise forcing terms are derived, are justified.
Thus when the system size is too small, no Langevin system will be an appropriate.
 Fortunately the Langevin approximation appears quite accurate for realistic population sizes.  

The $14\times28$D model uses more noise sources than other approaches. However, stochastic shielding allows us to jettison noise sources that do not significantly impact the system dynamics (the voltage fluctuations and ISI distribution). 
Moreover, in order to compare the ISI distribution in detail among several variants of the Fox and Lu '94 model versus the Markov Chain standard, we have considered a single value of the driving current, while other studies have compared parametrized responses such as the firing rate, ISI variance, or other moments, as a function of applied current. 
Accurate comparisons require large ensembles of independent trajectories, forcing a tradeoff between precision and breadth; we opted \rev{here} for precise comparisons at a representative level of the driving current.

\rev{From a conceptual point of view, a shortcoming of most Langevin models is the tendency for some channel state variables $x$ to collide with the domain boundaries $x\in[0,1]$ and to cross them during numerical simulations with finite time steps.  
We adopted the approach advocated by \cite{OrioSoudry2012PLoS1} of using ``free boundaries" in which gating variables can make excursions into the (unphysical) range $x<0$ or $x>1$.  
Practically speaking, these excursions are always short, if the time step is reasonably small, as they tend to be self-correcting.\footnote{\rev{To avoid complex entries, we use $|x|$ when calculating entries in the noise coefficient matrix.}}  
Another approach is to construct reflecting boundary conditions; different implementations of this idea were used in \cite{Dangerfield2010PCS},
\cite{FoxLu1994PRE}, and \cite{SchmidGoychukHanggi2001EPL}.
Dangerfield's method proved both slower and less accurate than the free boundary method.  
As an alternative method, one uses a biased rejection sampling approach, testing each gating variable of the 14D model on each time step, and repeating the noise sample for any time step violating the domain conditions \cite{FoxLu1994PRE,SchmidGoychukHanggi2001EPL}. 
We found that this method had accuracy similar to that of Dangerfield's method ($L_1$-Wasserstein difference $\approx 4.4\text{e-}1$ msec, cf.~Tab.~\ref{tab:L1ISIs}) and runtime similar to that of the Fox and Lu 94 implementation, about 4 times slower than our 14D Langevin model.
}

\rev{
Yet another approach that in principle can guarantee the stochastic process remains within proscribed bounds, rederives a ``best diffusional approximation" Fokker-Planck equation based on matching a master-equation birth-and-death description of a (binomial) population of two-state ion channels, leading to modified drift and diffusion terms \citep{Goychuk2014stochastic}. 
This method does not appear to extend readily to the 14D setting, with underlying multinomial structure of the ion channel gates, so we do not dwell on it further.  
}

\rev{Table \ref{tab:L1ISIs} gives the accuracies with which each model reproduces the ISI distribution, compared to a standard reference distribution generated through a large number of samples of the MC method.
The mean $L_1$ difference between a single sample and the reference sample is about 0.227 microseconds. 
For a nonnegative random variable $T\ge 0$, the difference in the mean under two probability distributions is bounded above by the $L_1$ difference in their cumulative distribution functions.\footnote{\rev{For a nonnegative random variable $T$ with cumulative distribution function $F(t)=\Pr[T\le t]$, the mean satisfies $E[T]=\int_0^\infty(1-F(t))\,dt$ \citep{Grimmett+Stirzaker:2005:book}. Therefore the difference in mean under two distributions $F_1$ and $F_2$ satisfies
$|E_1[T]-E_2[T]|=\left|\int_0^\infty F_1(t)-F_2(t)\,dt\right|\le \rho_1(F_1,F_2).$}  }
Thus the $L_1$ norm gives an idea of the temporal accuracy with which one can approximate a given distribution by another.
The mean difference between the ISI distribution generated by a single run of the full $14\times28$D model is about 49 $\mu$sec, and the discrepancy produced by the (significantly faster) SS model is about 76 $\mu$sec.
When would this level of accuracy matter for the function of a neuron within a network?  
The barn owl \textit{Tyto alba} uses interaural time difference to localize its prey to within 1-2 degrees, a feat that requires encoding information in the precise timing of auditory system action potentials at the scale of 5-20 microseconds \citep{MoiseffKonishi1981JCompPhys,GerstnerKempterVanHemmenWagner1996Nature}.
For detailed studies of the effects of channel noise in this system, the superior accuracy of the MC model might be preferred.  
On the other hand, the timescale of information encoding in the human auditory nerve is thought to be in the millisecond range \citep{Goldwyn2010JCompNeu}; with precision in the feline auditory system as reported as low as 100 $\mu$s \citep{Imennov2009IEEE}, see also \citep{WooMillerAbbas2009IEEETransBME}.
For these and other mammalian systems, the stochastic shielding approximation should provide sufficient accuracy.
}

\subsection{Future Work}
In \S\ref{subsec:SS} we compared the ISI variance when noise was included one edge at a time, and found that the edges making the greatest contribution to population fluctuations under voltage clamp were not identical to the edges having the largest effect on ISI variance, when taken one at a time. 
However, the ISI, considered as a random variable determined through a  first-passage time process, depends on the entire trajectory, not just on the occupancy of the conducting states.  
The HH dynamics are strongly nonlinear, producing a limit cycle in the deterministic case, and it is not immediately clear whether the effects of channel noise on ISI variability should be additive.  
In future work, we plan to address the question of the additive contribution of individual/molecular noise sources on ISI variability.

A principle motivation for using the stochastic shielding algorithm is to develop fast and accurate algorithms for ensemble simulations of forward models for parameter estimation in a data assimilation framework. We expect that our method may prove useful for such studies based on current-clamp electrophysiological data in the future.

\section*{Acknowledgments}
The authors thank Dr.~David Friel (Case Western Reserve University) for illuminating discussions of the impact of channel noise on neural dynamics, \rev{and the anonymous referees for  many detailed and constructive comments. } 

This work was made possible in part by grants from the National Science Foundation (DMS-1413770 and DEB-1654989) and the Simons Foundation. PT thanks the Oberlin College Department of Mathematics for research support.  
This research has been supported in part by the Mathematical Biosciences Institute and the National Science Foundation under grant DMS-1440386. 
Large-scale simulations made use of the High Performance Computing Resource in the Core Facility for Advanced Research Computing at Case Western Reserve University.

\newpage
\begin{appendices}

\section{Table of Common Symbols and Notations}\label{notations}

\begin{table}[htbp]\centering
   \begin{tabular}{|l|l|} 
   \hline
   Symbol  & Meaning \\
   \hline
   $C$ & Membrane capacitance ($\mu F/cm^2$)\\
   \hline
   $v$ & Membrane potential ($mV$) \\
   \hline
   $V_\text{Na},V_\text{K},V_\text{L}$& 
   Ionic reversal potential for \Na, \K and leak ($mV$)\\
   \hline
   $I_\text{app}$&
   Applied current to the membrane ($nA/cm^2$)\\
   \hline
  $m,\ h,\ n$ & Dimensionless gating variables for \Na~and \K~channels \\
   \hline
    $\alpha_x,\beta_x\,, x\in\{m,n,h\}$& Voltage dependent rate constant ($1/msec$) \\
    \hline 
   $\mbx$ & 
   vector of state variables
   \\ 
   \hline
   $\text{M}=[\text{M}_1,\text{M}_2,\cdots,\text{M}_8]$  & 
   Eight-component state vector for the \Na~gates\\
   \hline
   $[m_{00},m_{10},m_{20},m_{30},m_{01},m_{11},m_{21},m_{31}]^
\intercal$ & Components for the
\Na gates\\
\hline
   $\text{N}=[\text{N}_1,\text{N}_2,\cdots,\text{N}_5]$ & Five-component state vector for the \K~gates\\
   \hline
$[n_0,n_1,n_2,n_3,n_4]^\intercal$ & Components for the \K~gates\\
\hline
$M_\text{tot},\ N_\text{tot}$& Total number of \Na~and \K~channels\\
\hline
   $\mathcal{X}$ & 4-dimensional manifold domain for 4D HH model \\
   \hline
   $\mathcal{Y}$ & 14-dimensional manifold domain for 14D HH model \\
   \hline
   $\Delta^k$ & $k$-dimensional simplex in $\R^{k+1}$ \\
   & given by $y_1+\ldots+y_{k+1}=1, y_i\ge 0$\\
   \hline
    $\mathcal{M}$ & 
    Multinomial submanifold within the 14D space\\
    \hline
   $A_\text{Na}$, $A_\text{K}$ & State-dependent rate matrix\\
   \hline
   $D$ & State diffusion matrix\\
   \hline
   $S$, $S_1$, $S_2$, $S_\text{Na}$, $S_\text{K}$ & Noise coefficient matrices \\
   \hline
   $\xi$ & Vector of independent $\delta$-correlated Gaussian white \\
   & noise with zero mean and unit variance\\
   \hline 
  $\mbX=[X_1,X_2,\ldots,X_d]$ & A d-dimensional random variable for sample path \\
    \hline 
  $\mbW=[W_1,W_2,\ldots,W_n]$ & A Wiener trajectory with $n$ components \\
  \hline 
  $\delta(\cdot)$ & The Dirac delta function \\
  \hline 
  $\delta_{ij}$ & The Kronecker delta \\
  \hline 
  $F_n$ &   Empirical cumulative distribution function with $n$\\
  &observations (in \S\ref{sec:modelcomp}, we use $m,\ n$ as sample sizes)   \\
  \hline
   \end{tabular}
   \caption{Common symbols and notations in this paper.}
   \label{tab:notations}
\end{table}

\section{\rev{Parameters} and Transition Matrices }\label{axppend_alpha_beta}

\begin{table}[htbp]\centering
   \rev{\begin{tabular}{|l|l|r|} 
   \hline
   Symbol  & Meaning & Value \\
   \hline
   $C$ & Membrane capacitance & 1 $\mu F/cm^2$ \\
    \hline
     $\bar{g}_\text{Na}$ & Maximal sodium conductance & 120 $\mu S/cm^2$ \\
     \hline
     $\bar{g}_\text{K}$ & Maximal potassium conductance & 36 $\mu S/cm^2$ \\
     \hline
     $g_\text{leak}$ & Leak conductance & 0.3 $\mu S/cm^2$ \\
   \hline
   $V_\text{Na}$ & 
   Sodium reversal potential for \Na & 50 $mV$\\
   \hline
    $V_\text{K}$ & 
   Potassium reversal potential for \K & -77 $mV$\\
   \hline
    $V_\text{leak}$ & 
   Leak reversal potential & -54.4 $mV$\\
   \hline
   $I_\text{app}$&
   Applied current to the membrane  &10 $nA/cm^2$\\
   \hline
   $\mathcal{A}$ & Membrane Area & $100\,\mu\text{m}^2$\\
   \hline
   $M_\text{tot}$& Total number of \Na~channels & 6,000\\
\hline
$N_\text{tot}$& Total number of\K~channels & 18,00\\
\hline
   \end{tabular}}
   \caption{\rev{Parameters used for simulations in this paper.}}
   \label{tab:parameters}
\end{table}

 Subunit kinetics for the Hodgkin and Huxley equations are given by
\begin{align}
  \alpha_m(v)&=  \frac{0.1(25-v)}{ \exp(2.5-0.1v)-1} \label{eq:rate4}  \\
\beta_m(v)&= 4\exp(-v/18) \label{eq:rate5}  \\
\alpha_h(v)&= 0.07\exp(-v/20) \label{eq:rate6}  \\
\beta_h(v)&= \frac{1}{ \exp(3-0.1v)+1} \label{eq:rate7}  \\
\alpha_n(v)&= \frac{0.01 (10-v)}{\exp(1-0.1v)-1} \label{eq:rate8}  \\
\beta_n(v)&= 0.125\exp(-v/80) \label{eq:rate9}    
\end{align}


 \[ A_\text{K}(v) =\begin{bmatrix}
   D_\text{K}(1)& \beta_n(v)             & 0                & 0                  & 0\\
   4\alpha_n(v)& D_\text{K}(2)&   2\beta_n(v)              & 0&                   0\\
    0&        3\alpha_n(v)&        D_\text{K}(3)& 3\beta_n(v)&          0\\
    0&        0&               2\alpha_n(v)&          D_\text{K}(4)& 4\beta_n(v)\\
    0&        0&               0&                 \alpha_n(v)&          D_\text{K}(5)
\end{bmatrix},
\]  

 \[ A_\text{Na} =\begin{bmatrix}
D_\text{Na}(1) & \beta_m&0 &0 &\beta_h&0&0&0\\
3\alpha_m&D_\text{Na}(2)&2\beta_m&0&0&\beta_h&0&0\\
0&2\alpha_m&D_\text{Na}(3) &3\beta_m&0&0&\beta_h&0 \\
0&0&\alpha_m&D_\text{Na}(4)&0&0&0&\beta_h \\
\alpha_h&0&0&0&D_\text{Na}(5)&\beta_m&0&0\\
0&\alpha_h&0&0&3\alpha_m&D_\text{Na}(6)&2\beta_m&0\\
0&0&\alpha_h&0&0&2\alpha_m&D_\text{Na}(7)&3\beta_m\\
0&0&0&\alpha_h&0&0&\alpha_m&D_\text{Na}(8)\\
\end{bmatrix},
\]  
where the diagonal elements $$D_\text{ion}(i)=-\sum_{j:j\neq i}A_\text{ion}(j,i),\quad \quad \text{ion}\in \{\text{Na},\text{K}\}.$$

\section{Proof of Lemma \ref{LemmaXinvari}}\label{append_LemmaXinvari}
For the reader's convenience we restate
\begin{customlemma}{\ref{LemmaXinvari}}Let $\mathcal{X}$ and $\mathcal{Y}$ be the lower-dimensional and higher-dimensional Hodgkin-Huxley manifolds given by \eqref{eq:XYdefinitions}, and let $F$ and $G$ be the vector fields on $\mathcal{X}$ and $\mathcal{Y}$ defined by \eqref{eq:HH4D_V}-\eqref{eq:rate3} and \eqref{14dhh1}-\eqref{14dhh3}, respectively. 
Let $H:\mathcal{X}\to\mathcal{M}\subset\mathcal{Y}$ and $R:\mathcal{Y}\to\mathcal{X}$ be the mappings given in Tables  \ref{tab:4D-to-14D-for-HH} and \ref{tab:14D-to-4D-for-HH}, respectively, 
and define the multinomial submanifold $\mathcal{M}=H(\mathcal{X})$.  Then $\mathcal{M}$ is forward-time--invariant under the flow generated by $G$.  Moreover, the vector field $G$,  when restricted to $\mathcal{M}$, coincides  with the vector field induced by $F$ and the map $H$.  That is, for all $y\in\mathcal{M}$, $G(y)=D_x H(R(y))\cdot F(R(y)).$
\end{customlemma}

The main idea of the proof is to show that  show that for every $y\in\mathcal{Y}$, $G(y)$ is contained in the span of the four vectors $\left\{\frac{\partial H}{\partial x_i}(R(y))\right\}_{i=1}^4.$

\begin{proof}
The map from the 4D HH model to the 14D HH model is given in Tab.~\ref{tab:4D-to-14D-for-HH} as $\{H:x\to y\given x\in\mathcal{X}, y\in \mathcal{Y}\}$, and the map from the 14D HH model to the 4D HH model is given in Tab.~\ref{tab:14D-to-4D-for-HH} as $\{R:y\to x \given x\in\mathcal{X},y\in \mathcal{Y}\}$. 
The partial derivatives $\frac{\partial H}{\partial x}$ of the map $H$ are given by  
\begin{align*}
 &\frac{dm_{00}}{dm}=-3(1-m)^2(1-h)&&    \frac{dm_{00}}{dh}=  -(1-m)^3\\
 & \frac{dm_{10}}{dm}=3(1-h)(3m^2-4m+1) &&  \frac{dm_{10}}{dh}=  -3(1-m)^2m\\
 & \frac{dm_{20}}{dm}=3(1-h)(2m-3m^2)  &&   \frac{dm_{20}}{dh}=  -3(1-m)m^2\\
 & \frac{dm_{30}}{dm}=3(1-h)m^2   && \frac{dm_{30}}{dh}=  -m^3\\
 & \frac{dm_{01}}{dm}=-3h(1-m)^2    && \frac{dm_{01}}{dh}=  (1-m)^3\\
 & \frac{dm_{11}}{dm}=3h(3m^2-4m+1)  &&  \frac{dm_{11}}{dh}=  3(1-m)^2m\\
 & \frac{dm_{21}}{dm}= 3h(2m-3m^2)  &&   \frac{dm_{21}}{dh}=  3(1-m)m^2\\
 & \frac{dm_{31}}{dm}=3hm^2   && \frac{dm_{31}}{dh}=  m^3.
\end{align*}
We can write $\partial H/\partial x$ in matrix form as:
\[\frac{\partial H}{\partial x}=\begin{bmatrix}
1&0&0&0 \\
0&-3(1-m)^2(1-h) & -(1-m)^3&0\\
0&3(1-h)(3m^2-4m+1)&-3(1-m)^2m&0\\
0&3(1-h)(2m-3m^2)&-3(1-m)m^2&0\\
0&3(1-h)m^2&-m^3&0\\
0&-3h(1-m)^2&(1-m)^3&0\\
0&3h(3m^2-4m+1)& 3(1-m)^2m&0\\
0&3h(2m-3m^2)&3(1-m)m^2&0\\
0&3hm^2&m^3&0 \\
0&0&0&-4(1-n)^3\\ 
0&0&0&4(1-n)^2(1-4n)\\
0&0&0&12n(1-n)(1-2n)\\
0&0&0&4n^2(3-4n)\\
0&0&0&4n^3
\end{bmatrix}.
\]
We write out the vector fields \eqref{14dhh2} and \eqref{14dhh3} component by component:
\begingroup
\allowdisplaybreaks
\begin{eqnarray}\label{14d_4d}
\frac{d\text{M}_1}{dt}&=&\beta_m\text{M}_2+\beta_h\text{M}_5-(3\alpha_m+\alpha_h)\text{M}_1 \nonumber \\
&=&-3(1-m)^2(1-h)\left[(1-m)\alpha_m-m\beta_m\right]+(1-m)^3\left[h\beta_h-(1-h)\alpha_h\right] \nonumber \\
\frac{d\text{M}_2}{dt}&=&3\alpha_m\text{M}_1+2\beta_m\text{M}_3+\beta_h\text{M}_6-(2\alpha_m+\beta_m+\alpha_h)\text{M}_2 \nonumber \\
&=&3(1-h)(3m^2-4m+1)\left[(1-m)\alpha_m-m\beta_m\right]+3(1-m)^2m\left[h\beta_h-(1-h)\alpha_h\right]  \nonumber \\
\frac{d\text{M}_3}{dt}&=&2\alpha_m\text{M}_2+3\beta_m\text{M}_4+\beta_h\text{M}_7-(\alpha_m+2\beta_m+\alpha_h)\text{M}_3, \nonumber \\
&=&3(1-h)(2m-3m^2)\left[(1-m)\alpha_m-m\beta_m\right]+3(1-m)m^2\left[h\beta_h-(1-h)\alpha_h\right]  \nonumber \\
\frac{d\text{M}_4}{dt}&=&\alpha_m\text{M}_3+\beta_h\text{M}_8-(3\beta_m+\alpha_h)\text{M}_4, \nonumber \\
&=&3(1-h)m^2\left[(1-m)\alpha_m-m\beta_m\right] +m^3\left[h\beta_h-(1-h)\alpha_h\right]\nonumber\\
\frac{d\text{M}_5}{dt}&=&\beta_m\text{M}_6+\alpha_h\text{M}_1-(3\alpha_m+\beta_h)\text{M}_5,\nonumber \\
&=& -3h(1-m)^2\left[(1-m)\alpha_m-m\beta_m\right] +(1-m)^3\left[h\beta_h-(1-h)\alpha_h\right] \nonumber\\
\frac{d\text{M}_6}{dt}&=&3\alpha_m\text{M}_5+2\beta_m\text{M}_7+\alpha_h\text{M}_2-(2\alpha_m+\beta_m+\beta_h)\text{M}_6, \nonumber \\
&=&3h(3m^2-4m+1)\left[(1-m)\alpha_m-m\beta_m\right] -3(1-m)^2m\left[h\beta_h-(1-h)\alpha_h\right]\nonumber\\
\frac{d\text{M}_7}{dt}&=&2\alpha_m\text{M}_6+3\beta_m\text{M}_8+\alpha_h\text{M}_3-(\alpha_m+2\beta_m+\beta_h)\text{M}_7, \nonumber \\
&=&3h(2m-3m^2)\left[(1-m)\alpha_m-m\beta_m\right] - 3(1-m)m^2\left[h\beta_h-(1-h)\alpha_h\right]\nonumber\\
\frac{d\text{M}_8}{dt}&=&\alpha_m\text{M}_7+\alpha_h\text{M}_4-(3\beta_m+\beta_h)\text{M}_8, \nonumber\\
&=&3hm^2\left[(1-m)\alpha_m-m\beta_m\right] -m^3\left[h\beta_h-(1-h)\alpha_h\right]\nonumber\\
\frac{d\text{N}_1}{dt}&=&\beta_n\text{N}_2-4\alpha_n\text{N}_1=-4(1-n)^3[\alpha_n(1-n)-n\beta_n],\nonumber\\
\frac{d\text{N}_2}{dt}&=&4\alpha_n\text{N}_1+2\beta_n\text{N}_3-(3\alpha_n+\beta_n)\text{N}_2=4(1-n)^2(1-4n)[\alpha_n(1-n)-n\beta_n], \nonumber \\
\frac{d\text{N}_3}{dt}&=&3\alpha_n\text{N}_2+3\beta_n\text{N}_4-(2\alpha_n+2\beta_n)\text{N}_3=12n(1-n)(1-2n)[\alpha_n(1-n)-n\beta_n], \nonumber \\
\frac{d\text{N}_4}{dt}&=&2\alpha_n\text{N}_3+4\beta_n\text{N}_5-(3\alpha_n+3\beta_n)\text{N}_4=4n^2(3-4n)[\alpha_n(1-n)-n\beta_n], \nonumber \\
\frac{d\text{N}_5}{dt}&=&\alpha_n\text{N}_4-4\beta_n\text{N}_5=4n^3[\alpha_n(1-n)-n\beta_n].\nonumber
\end{eqnarray}
\endgroup
By extracting common factors from the previous expressions it is clear that $G(y)$ may be written, for all $y\in\mathcal{Y}$, as 
\begin{align}
   G(y)=&\frac{-\bar{g}_{Na}\text{M}_8(V-V_{Na})-\bar{g}_{K}\text{N}_5(V-V_K)-g_L(V-V_L)+I_\text{app}}{C}\left\{\frac{\partial H}{\partial v}(R(y))\right\} \nonumber\\
   &+\left[(1-m')\alpha_m-m'\beta_m\right] \left\{\frac{\partial H}{\partial m}(R(y))\right\} \nonumber\\
   &-\left[h'\beta_h-(1-h')\alpha_h\right] \left\{\frac{\partial H}{\partial h}(R(y))\right\} \nonumber\\
   & +[\alpha_n(1-n')-n'\beta_n]\left\{\frac{\partial H}{\partial n}(R(y))\right\}\label{eq:Gy}
\end{align}
where $m'=(\text{M}_2+\text{M}_6)/3+2(\text{M}_3+\text{M}_7)/3+(\text{M}_4+\text{M}_8),$ $h'=\text{M}_5+\text{M}_6+\text{M}_7+\text{M}_8$ and $n'=\text{N}_2/4+\text{N}_3/2+3\text{N}_4/4+\text{N}_5$.
Thus $G(y)$ is in the span of the column vectors $\partial H/\partial v$, $\partial H/\partial m$, $\partial H/\partial n$, and $\partial H/\partial h$, as was to be shown.

On the other hand, the vector field for the 4D HH ODE (\ref{eq:HH4D_V}-\ref{eq:rate3}) is given by
\[
F=\begin{bmatrix}
\left(-\bar{g}_{Na}m^3h(V-V_{Na})-\bar{g}_{K}n^4(V-V_K)-g_L(V-V_L)+I_\text{app}\right)/C\\
   \alpha_m(V)(1-m)-\beta_m(V)m   \\
  \alpha_h(V)(1-h)-\beta_h(V)h  \\
    \alpha_n(V)(1-n)-\beta_n(V)n 
     \end{bmatrix}.
\]
Referring to  \eqref{eq:Gy}, we see that $G(y)=D_x H(R(y))\, F(R(y))$.    
 Thus we complete the proof of Lemma \ref{LemmaXinvari}.
\end{proof}

\section{Diffusion Matrix of the 14D Model}\label{app:SNaSK}
As defined in equations \eqref{eq:define_M} and \eqref{eq:define_N}
\begin{align}
\mbM&=[m_{00},m_{10},m_{20},m_{30},m_{01},m_{11},m_{21},m_{31}]^\intercal,  \\
\mbN&=[n_0,n_1,n_2,n_3,n_4]^\intercal,
\end{align}
the diffusion matrices $D_\text{Na}$ and $D_\text{K}$ are given by

\begin{align*}
D_\text{K}=
&\left[
\begin{array}{ccc}
D_\text{K}(1,1)&-4\alpha_n n_0-\beta_n n_1 &0\\
   -4\alpha_n n_0-\beta_n n_1& 
   D_\text{K}(2,2)&
   3\alpha_n n_1 -2\beta_n n_2 \\
  0 &-3\alpha_n n_1 -2\beta_n n_2&D_\text{K}(3,3)\\
    0 &0&-2\alpha_n n_2-3\beta_n n_3\\
    0&0&0 \\
    \end{array} 
\right.\cdots\\
&\quad\quad\quad\quad\quad\quad\quad\quad\quad\cdots \left. 
\begin{array}{cc}
0&0\\
 0&0 \\
-2\alpha_n n_2-3\beta_n n_3&0\\
    D_\text{K}(4,4)& -\alpha_n n_3-4\beta_n n_4\\
    -\alpha_n n_3-4\beta_n n_4&D_\text{K}(5,5) \\
    \end{array}
    \right],
\end{align*}


{\footnotesize{
\[
D_\text{Na}^{(1:4)}=\begin{bmatrix}
D_\text{Na}(1,1)&-3\alpha_m m_{00}-\beta_m m_{10} &0&0\\
   -3\alpha_m m_{00}-\beta_m m_{10}& 
   D_\text{Na}(2,2)&-2\alpha_m m_{10}-2\beta_m m_{20} &0 \\
  0 & -2\alpha_m m_{10}-2\beta_m m_{20} &D_\text{Na}(3,3)&-\alpha_m m_{20}-3\beta_m m_{30}\\
    0 &0&-\alpha_m m_{20}-3\beta_m m_{30} &D_\text{Na}(4,4) \\
    -\alpha_h m_{00}-\beta_h m_{01}&0&0&0  \\
    0&-\alpha_h m_{10}-\beta_h m_{11}&0&0 \\
    0&0&-\alpha_h m_{20}-\beta_h m_{21}&0
     \end{bmatrix},
\]}}

{\footnotesize{
\[D^{(5:8)}_\text{Na}=\begin{bmatrix}
-\alpha_h m_{00}-\beta_h m_{01}&0&0&0\\
 0&-\alpha_h m_{10}-\beta_h m_{11}&0&0 \\
  0&0& -\alpha_h m_{20}-\beta_h m_{21}&0\\
   0&0&0&-\alpha_h m_{30}-\beta_h m_{31} \\
   D_\text{Na}(5,5)&-3\alpha_m m_{01}-\beta_m m_{11}&0&0  \\
    -3\alpha_m m_{01}-\beta_m m_{11}&D_\text{Na}(6,6) &-2\alpha_m m_{11}-2\beta_m m_{21} &0 \\
   0&-2\alpha_m m_{11}-2\beta_m m_{21}&D_\text{Na}(7,7) & -\alpha_m m_{21}-3\beta_m m_{31}
     \end{bmatrix},
\]}}


where $$D_\text{ion}(i,i)=-\sum_{j\::\:j\neq i}D_\text{ion}(j,i),\text{ for ion }\in\{\text{Na},\text{K}\}.$$ 
%

The matrices $S_\text{K}$ and $S_\text{Na}$ for the $14\times28$D HH model are given by
\begin{align*}
S^{(1:5)}_\text{Na}=&\left[
\begin{array}{ccccc}
-\sqrt{\alpha_h m_{00}}& \sqrt{\beta_h m_{01}}&-\sqrt{3\alpha_m m_{00}}&\sqrt{\beta_m m_{10}}&
0\\
   0& 0&\sqrt{3\alpha_m m_{00}}&-\sqrt{\beta_m m_{10}}&
   -\sqrt{\alpha_h m_{10}} \\
  0 &0 &0&0&0\\
    0 &0&0&0&0\\
    \sqrt{\alpha_h m_{00}}&-\sqrt{\beta_h m_{01}}&0&0&0 \\
     -\sqrt{\beta_h m_{11}}&0&0&0&0 \\
      0&0&0&\sqrt{\alpha_h m_{20}}&-\sqrt{\beta_h m_{21}} \\
       0&0&0&0&0 \\
    \end{array} 
\right]\\
S_\text{Na}^{(6:10)}=&\left[ 
\begin{array}{ccccc}
0&0&0&0&0\\
 \sqrt{\beta_h m_{11}}
   &-\sqrt{2\alpha_m m_{10}}&\sqrt{2\beta_m m_{20}}&0&0 \\
  \sqrt{2\alpha_m m_{10}}
  &-\sqrt{2\beta_m m_{20}}&-\sqrt{\alpha_h m_{20}}&\sqrt{\beta_h m_{21}}\\
    0&0&0 &0&0\\
   0&0&0&0&0 \\
     -\sqrt{\beta_h m_{11}}&0&0&0&0 \\
      0&0&0&\sqrt{\alpha_h m_{20}}&-\sqrt{\beta_h m_{21}} \\
       0&0&0&0&0 \\
    \end{array}
\right]\\
S^{(11:15)}_\text{Na}=&\left[ 
\begin{array}{ccccc}
0&0&0&0&0\\
0&0&0&0&0\\
   -\sqrt{\alpha_m m_{20}}&
   \sqrt{3\beta_m m_{30}}&0&0
   &0\\
   \sqrt{\alpha_m m_{20}}&
   -\sqrt{3\beta_m m_{30}}
   &-\sqrt{\alpha_h m_{30}}
   &\sqrt{\beta_h m_{31}}&0\\
 0&0&0&0
   0&\\
0&0&0&0&\sqrt{3\alpha_m m_{01}}\\
 0&0&0&0
   &0\\
   0&0&\sqrt{\alpha_h m_{30}}&-\sqrt{\beta_h m_{31}}
   &0\\
    \end{array}
\right]\\
S^{(16:20)}_\text{Na}=&\left[
\begin{array}{ccccc}
0&0&0&0&0\\
0&0&0&0&0\\
0&0&0&0&0\\
0&0&0&0&0\\
-\sqrt{3\alpha_m m_{01}}
   &\sqrt{\beta_m m_{11}}&0&0&0\\
   -\sqrt{\beta_m m_{11}}&-\sqrt{2\alpha_m m_{11}}&\sqrt{2\beta_m m_{21}}&0&0\\
   0&\sqrt{2\alpha_m m_{11}}&-\sqrt{2\beta_m m_{21}}&-\sqrt{\alpha_m m_{21}}&\sqrt{3\beta_m m_{31}}\\
   0&0&0&\sqrt{\alpha_m m_{21}}&-\sqrt{3\beta_m m_{31}}\\
    \end{array}
    \right],
\end{align*}
where $S^{(i:j)}_\text{Na}$ is the i$^{th}$-j$^{th}$ column of $S_\text{Na}$, and
\begin{align*}
S_\text{K}=&\left[
\begin{array}{cccc}
-\sqrt{4\alpha_n n_0}& \sqrt{\beta_n n_1}&0&0\\
   \sqrt{4\alpha_n n_0}& -\sqrt{\beta_n n_1}&-\sqrt{3\alpha_n n_1}&\sqrt{2\beta_n n_2} \\
  0 &0 &\sqrt{3\alpha_n n_1}&-\sqrt{2\beta_n n_2}\\
    0 &0&0&0 \\
    0&0&0&0 \\
   \end{array} 
\right.\cdots\\
&\quad\cdots \left. 
\begin{array}{cccc}
0&0&0&0\\
0&0&0&0 \\
\sqrt{2\alpha_n n_2}&-\sqrt{3\beta_n n_3}&-\sqrt{\alpha_n n_3}&\sqrt{4\beta_n n_4} \\
    0&0&\sqrt{\alpha_n n_3}&-\sqrt{4\beta_n n_4} \\
  \end{array}
    \right].
\end{align*}

\end{appendices}

\bibliography{Shusen}
\bibliographystyle{apalike}

\end{document}